\documentclass[10pt, a4paper]{article}

\usepackage[utf8]{inputenc}

\usepackage{amsmath, amsthm}
\usepackage{amsfonts}
\usepackage{amssymb}
\usepackage{amsxtra}
\usepackage{mathrsfs}
\usepackage{mathtools}

\usepackage{bbm}
\usepackage{bm}
\usepackage{color}
\usepackage{float}
\usepackage{graphicx}
\usepackage{caption}
\usepackage{subfigure}
\usepackage{cancel}
\usepackage{comment}
\usepackage{latexsym}
\usepackage{soul}
\usepackage{todonotes}
\usepackage{accents}

\usepackage{authblk}
\usepackage{natbib}

\usepackage{tikz}
\usetikzlibrary{arrows.meta,calc,positioning}

\usepackage{slashed}
\usepackage{leftidx}
\usepackage{tensor}
\usepackage{epstopdf}

\usepackage{enumerate}
\usepackage{enumitem}

\usepackage{hyperref}

\newtheorem{theorem}{Theorem}[section]
\newtheorem{lemma}[theorem]{Lemma}
\newtheorem{corollary}[theorem]{Corollary}
\newtheorem{proposition}[theorem]{Proposition}

\newtheorem{definition}[theorem]{Definition}
\newtheorem{example}[theorem]{Example}

\newtheorem{convention}[theorem]{Convention}

\theoremstyle{remark}
\newtheorem{remark}[theorem]{Remark}

\numberwithin{equation}{section}
\usepackage[left=2cm,right=2cm,top=3cm,bottom=3cm]{geometry}

\newcommand{\mr}{\mathring}

\makeatletter

\newcommand{\Rmnum}[1]{\expandafter\@slowromancap\romannumeral #1@}
\makeatother

\begin{document}

\title{A construction of charged naked singularities for the Einstein--Maxwell-charged scalar field equations in $3+1$ dimensions}

\author{Weihao Zheng\thanks{wz344@math.rutgers.edu}}

\affil{\small Department of Mathematics, Rutgers University, Hill Center, 110 Frelinghuysen Road, Piscataway, NJ, USA}
\date{\today}

\maketitle
\begin{abstract}
We construct a class of global naked singularity solutions for the $(3+1)$-dimensional spherically symmetric Einstein--Maxwell--charged scalar field equations. The solutions arise from $C^{1,\alpha}$ asymptotically flat one-ended characteristic initial data with $0<\alpha\ll1$ and develop a curvature singularity that is not enclosed by an event horizon. The resulting spacetimes are genuinely charged, with a nonzero Maxwell field that satisfies a quantitative positive lower bound up to the singularity and admits an explicit asymptotic profile. {As part of the construction, we develop a general framework for constructing naked singularity exteriors, namely, the causal future region of the past light cone emanating from the singularity. The framework accommodates a broad range of singular characteristic data and, in particular, allows for singular behavior of the Maxwell field near the singularity.}
\end{abstract}
\section{Introduction}
Understanding the global evolution of regular initial data for the Einstein equations and the formation of singularities is one of the fundamental problems in General Relativity. The Weak Cosmic Censorship conjecture (WCC) for the Einstein equations (coupled to reasonable matter models), formulated in~\cite{chris99B,penroseWCC}, asserts that if a singularity exists in the evolution of \textbf{generic} regular initial data, then it must be hidden inside a black hole region. 

The word ``generic'' is included in the formulation of WCC to allow for the possible existence of some exceptional initial data whose maximal developments contain singularities not enclosed by an event horizon, namely, naked singularity spacetimes. {The geometry of these spacetimes allows the causal signals to ``escape'' from the singularity and propagate to null infinity, thus rendering those singularities \textit{naked} to distant observers.} Indeed, over the past several decades, various naked singularity spacetimes have been constructed as solutions to the Einstein equations; see \cite{christodoulou1994examples,choptuik1} for the Einstein--scalar field system, \cite{guo2023naked} for the Einstein--Euler system, and \cite{rodnianski2023naked,yakov22} for the Einstein vacuum equations.

In this paper, we are interested in charged naked singularities for the Einstein--Maxwell-charged scalar field equations: \begin{align}
    &Ric_{\mu\nu}(g)-\frac{1}{2}R(g)g_{\mu\nu} = T_{\mu\nu}^{EM}+T_{\mu\nu}^{SC},\label{eq1}\\&
    T_{\mu\nu}^{EM} = 2\left(g^{\alpha\beta}F_{\alpha\mu}F_{\beta\nu}-\frac{1}{4}F^{\alpha\beta}F_{\alpha\beta}g_{\mu\nu}\right),\quad \nabla^{\mu}F_{\mu\nu} = iq_{0}\left(\frac{\phi\overline{D_{\nu}\phi}-\overline{\phi}D_{\nu}\phi}{2}\right),\ F = dA,\\&
    T_{\mu\nu}^{SC} = 2\left(\Re(D_{\mu}\phi\overline{D_{\nu}\phi})-\frac{1}{2}g^{\alpha\beta}D_{\alpha}\phi\overline{D_{\beta}\phi}g_{\mu\nu}\right), D_{\mu} = \nabla_{\mu}+iq_{0}A_{\mu},\\&g^{\mu\nu}D_{\mu}D_{\nu}\phi = 0,\label{eq2}
\end{align}
where $(\mathcal{M},g)$ is a Lorentzian spacetime, the two-form $F$ is the Maxwell field, the function $\phi$ is the complex-valued scalar field, the one-form $A$ is the electromagnetic potential, and the real number $q_{0}\neq0$ is the scalar field charge. The solution to~\eqref{eq1}--\eqref{eq2} consists of a Lorentzian manifold $(\mathcal{M},g)$ associated with a real-valued two form $F$ and a complex function $\phi:\mathcal{M}\rightarrow\mathbb{C}$.
\begin{theorem}[Rough version of Theorem~\ref{main theorem1}]
\label{thm: rough version for whole spacetime}
Fix a nonzero scalar field charge $q_{0}\neq 0$. There exists a class of spherically symmetric naked singularity solutions to the Einstein--Maxwell-charged scalar field equations, with $C^{1,\alpha}$ initial data imposed on the initial outgoing cone $C_{out}$ emanating from the center $\{r = 0\}$ for $0<\alpha\ll1$. The spacetime has an incomplete null infinity $\mathcal{I}$, with a finite affine length. Moreover, the naked singularity $\mathcal{O}$ is genuinely charged, in the sense that $\lim_{p\rightarrow\mathcal{O}}\vert g^{\alpha\mu}g^{\beta\nu}F_{\alpha\beta}F_{\mu\nu}\vert(p)\approx 1$ for $p$ on the past light cone emanating from $\mathcal{O}$. The spacetime can be extended to the future light cone emanating from $\mathcal{O}$ in $C^{1,\alpha^{\prime}}$ for some $\alpha^{\prime}\in(0,1)$ and the area radius $r$ is strictly increasing along the cone.
\end{theorem}
The Penrose diagram of this spacetime is depicted in Figure~\ref{fig: penrose diagram}.
\begin{figure}[htbp]
    \centering
    \begin{tikzpicture}[ scale=0.6, >=Latex, every node/.style={font=\small} ]
\coordinate (O) at (0,0);
\coordinate (A) at (4,-4);
\coordinate (B) at (6,6);
\coordinate (C) at (10,2);
\coordinate (Ce) at (0,-8);
\draw (O)--(Ce) node[midway, left]{$\Gamma: = \{r = 0\}$};

\draw (O)--(A) node[midway, above, sloped]
    {\rotatebox{90}{$\underline{C}$}};
\draw[dashed] (O)--(B);
\draw (Ce)--(C) node[midway, below, sloped]{$C_{out}$};
\draw[dashed] (B) -- (C)
    node[midway, above, sloped]
    {\rotatebox{90}{$\mathcal{I}$}};
\node[left] at (O) { $\text{singularity }\mathcal{O}$};
\coordinate (A1) at (6,-2);
\coordinate (A2) at (8,0);

\coordinate (A3) at (3,0);
\coordinate (A4) at (1.5,1.5);

\coordinate (L1) at (2.5,-1.5);
\coordinate (L2) at (3.5,-0.5);

\coordinate (L3) at (1.5,0.5);

\coordinate (L4) at (5,3);
\coordinate (ext) at (4,0);
\node at (ext) {$\mathcal{R}_{ext}$};
\coordinate (int) at (2,-4);
\node at (int) {$\mathcal{R}_{int}$};

\coordinate (f) at (4,4);

\coordinate (L5) at (4,2);
\coordinate (L6) at (7,3);

\end{tikzpicture}
    \caption{Penrose diagram}
    \label{fig: penrose diagram}
\end{figure}
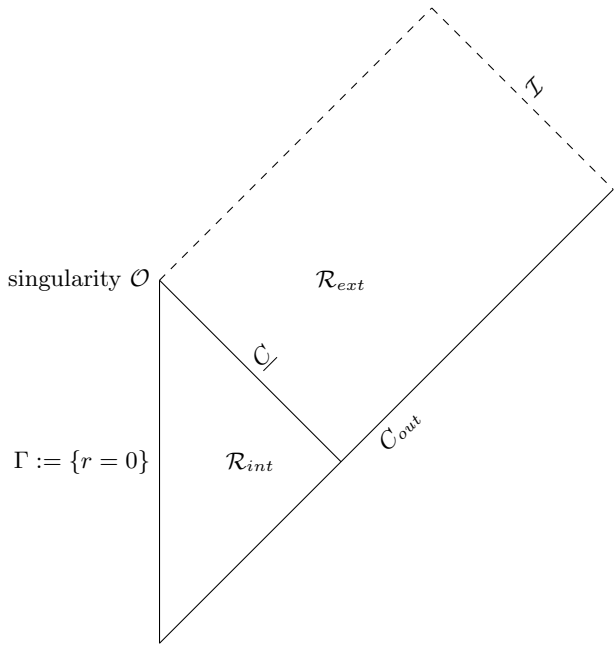
\paragraph{Previous works on self-gravitating charged scalar field}
The decay properties and singularity structure of solutions to the Einstein--Maxwell--charged scalar field equations have been studied {mathematically in~\cite{kommemi2013global,van2018stability,van2021mass,van2023breakdown,van2025asymptotically,van2025coexistence,an2022trapped,shen2025formation,giorgi2025cauchy,kehle2025gravitational,gajic2019interior,kauffman2023global}}. In particular, in the small-data regime, global existence and decay to Minkowski spacetime were established in~\cite{chae2003global,kauffman2023global}. By contrast, for initial data that are sufficiently large in a suitable sense, the formation of trapped surfaces was proved in~\cite{an2022trapped,shen2025formation,giorgi2025cauchy}. {The construction of charged naked singularity solutions therefore requires balancing two competing effects: the initial data must be sufficiently large to produce a singularity, while remaining sufficiently controlled to avoid the formation of trapped surfaces that would enclose the singularity within a black hole.}

\paragraph{The self-gravitating real scalar field system}
{In the case of an uncharged, real-valued self-gravitating scalar field with vanishing Maxwell field ($F\equiv 0$ in~\eqref{eq1}--\eqref{eq2}), the equations~\eqref{eq1}--\eqref{eq2} reduce to the Einstein--scalar field equations; see~\eqref{eq: intro: esf1}--\eqref{eq: intro: ESF2}.} For such a system, the small-data global existence and the detailed structure of singularities have been completely identified by Christodoulou in his series of seminal works~\cite {chris86,chris91,chris93,christodoulou1994examples,chris99B,chris99}. In particular, he constructed \textit{continuously self-similar} naked singularity solutions for the Einstein-scalar field equations in~\cite{christodoulou1994examples}; see Section~\ref{intro: sec: chris original construction} for a brief review of his construction. Christodoulou's solutions can also be viewed as solutions to the Einstein--Maxwell--charged scalar field equations by embedding the real scalar field in the complex one and setting $F\equiv0$. The solutions constructed in Theorem~\ref{thm: rough version for whole spacetime}, however, are genuinely charged. More precisely, the Maxwell field is everywhere nontrivial and admits a quantitative nonzero lower bound when approaching the singularity $\mathcal{O}$.

\paragraph{The exterior and interior regions}
{A singularity in a Lorentzian spacetime is called the first singularity if its causal past does not contain any singularity. As proved rigorously in~\cite{kommemi2013global}, in a spherically symmetric spacetime, such a singularity $\mathcal{O}$ must appear at the center of the spherical symmetry $\Gamma:=\{r = 0\}$.} The modern framework for constructing naked singularity spacetimes for the Einstein equations has been shaped largely by the pioneering works
\cite{christodoulou1994examples,rodnianski2023naked,yakov22}.
In this framework, the construction is naturally divided into two parts:
\begin{itemize}
    \item The construction of the exterior region $\mathcal{R}_{\mathrm{ext}}$, defined as the causal future of the past light cone emanating from the naked singularity $\mathcal{O}$.

    \item The construction of the interior region $\mathcal{R}_{\mathrm{int}}$, defined as the causal past of the past light cone emanating from the naked singularity $\mathcal{O}$.
\end{itemize}
See Figure~\ref{fig: penrose diagram} for an illustration of these two regions. Christodoulou's original construction of continuously self-similar naked singularities for the Einstein-scalar field equations
\cite{christodoulou1994examples} treated both the interior and exterior regions and relied crucially on exact continuous self-similarity and spherical symmetry, rather than on solving an initial value problem; see Section~\ref{intro: sec: chris original construction}. By contrast, the methodology for constructing the exterior region was substantially generalized by Rodnianski and Shlapentokh-Rothman~\cite{rodnianski2023naked}. {Their approach consists of prescribing suitable characteristic initial data on $
\underline{C}\cup\bigl(C_{\mathrm{out}}\cap\mathcal{R}_{\mathrm{ext}}\bigr)$, where, as depicted in Figure~\ref{fig: penrose diagram}, $\underline{C}$ is the ingoing null cone and $C_{\mathrm{out}}$ is the outgoing null cone emanating from the center. The data on $\underline{C}$ are singular at the endpoint $\mathcal{O}$.} One then establishes global existence for the resulting characteristic initial value problem. {This approach has several advantages:
\begin{itemize}
    \item It does not rely on exact self-similarity and spherical symmetry. Therefore, it is considerably more robust and flexible for extensions to the Einstein equations coupled with other matter models or the Einstein vacuum equations~\cite{rodnianski2023naked}.
    \item The freedom to prescribe suitable data on the initial outgoing cone $C_{\mathrm{out}}\cap\mathcal{R}_{\mathrm{ext}}$ allows one to incorporate sufficiently regular perturbations supported in the exterior region directly into the construction. In this sense, the stability of the resulting spacetimes under such perturbations can be obtained within the same framework; see~\cite{rodnianski2023naked} for the Einstein vacuum equations and~\cite{singh2024construction,an2026naked} for the corresponding exterior stability results for the Einstein-scalar field equations.
    \item The freedom to prescribe singular data on the initial ingoing cone $\underline{C}$ makes it possible to accommodate a broader range of singular behaviors. In particular, the regularity of the initial data on $C_{\mathrm{out}}\cap\mathcal{R}_{\mathrm{ext}}$ is closely related to the singular behavior prescribed on $\underline{C}$. This flexibility allows the approach of~\cite{rodnianski2023naked} to be adapted to the construction of exterior regions of more general naked singularity spacetimes, including examples arising from smooth initial data on the outgoing cone; see~\cite{cicortas2024discretely} for exteriors of naked singularity solutions for the Einstein-scalar field equations with discretely self-similar behavior on $\underline{C}$ and smooth outgoing initial data.
\end{itemize}}
Building on this framework, we give a systematic treatment of the exterior region of naked singularity spacetimes for the Einstein--Maxwell--charged scalar field equations, covering, in particular, a scenario with stronger singular behavior of the Maxwell field.
\begin{theorem}[Rough version of Theorem~\ref{main theorem2}]
\label{intro: rough version of the exterior construction}
Fix a nonzero scalar field charge $q_{0}\neq0$. {Then there exists a broad class of spherically symmetric exterior regions of naked singularity spacetimes solving the Einstein--Maxwell--charged scalar field equations. These exterior regions arise from $C^{1,\alpha}$ initial data on $C_{\mathrm{out}}\cap\mathcal{R}_{\mathrm{ext}}$ with $0<\alpha\ll1$. Moreover, letting $r$ denote the area radius, the Maxwell field satisfies the following singular behavior along the past light cone emanating from the singularity $\mathcal{O}$:
\begin{equation*}
    \left|
    g^{\alpha\mu}g^{\beta\nu}
    F_{\alpha\beta}F_{\mu\nu}
    \right|^{\frac12}(p)
    \approx
    \left|\log r(p)\right|
    \qquad\text{as } p\to\mathcal{O}.
\end{equation*}
Furthermore, the exterior solutions constructed in this theorem are stable under sufficiently regular perturbations supported in $C_{\mathrm{out}}\cap\mathcal{R}_{\mathrm{ext}}$.}
\end{theorem}
{The above theorem concerns only the exterior region; constructing a compatible interior solution remains an open problem, which we leave for future work. A related exterior-only construction was obtained in~\cite{cicortas2024discretely} for the Einstein--scalar field equations, where they constructed discretely self-similar (see Section~\ref{sec: intro: critical collapse}) exterior naked singularity solutions without a corresponding interior fill-in.}

\subsection{Comparison with naked singularities for the Einstein-scalar field equations and the Einstein vacuum equations}
\subsubsection{$k$-self-similar naked singularities constructed in~\cite{christodoulou1994examples}}
\label{intro: sec: chris original construction}
In this section, we briefly review the continuously self-similar naked singularity solutions, also known as the $k$-self-similar naked singularities, constructed by Christodoulou~\cite{christodoulou1994examples} for the Einstein-scalar field equations (i.e., $F=q_{0} = 0$ in \eqref{eq1}--\eqref{eq2}).

We write the spacetime metric $g$ in the so-called double-null gauge $(u,v,\theta,\varphi)$: \begin{equation}
\label{eq: intro double null gauge}
    g = -\frac{1}{2}\Omega^{2}(u,v)du\otimes dv-\frac{1}{2}\Omega^{2}(u,v)dv\otimes du+r^{2}(u,v)d\sigma^{2},\quad d\sigma^{2}: = d\theta^{2}+\sin^{2}\theta d\varphi^{2}.
\end{equation}
Then the Einstein-scalar field equations become \begin{align}
  &r\partial_{u}\partial_{v}r+\partial_{u}r\partial_{v}r = -\frac{1}{4}\Omega^{2},\quad \partial_{u}\left(\frac{\partial_{u}r}{\Omega^{2}}\right) = -\frac{r}{\Omega^{2}}\vert\partial_{u}\phi\vert^{2},\quad \partial_{v}\left(\frac{\partial_{v}r}{\Omega^{2}}\right) = -\frac{r}{\Omega^{2}}\vert\partial_{v}\phi\vert^{2},\label{eq: intro: esf1}\\&
  \partial_{u}\partial_{v}\log\Omega^{2} = \frac{\Omega^{2}}{2r^{2}}+\frac{2\partial_{u}r\partial_{v}r}{r^{2}}-2\partial_{u}\phi\partial_{v}\phi,\quad r\partial_{u}\partial_{v}\phi+\partial_{u}r\partial_{v}\phi+\partial_{v}r\partial_{u}\phi = 0.\label{eq: intro: ESF2}
\end{align}
The above equations admit two symmetries: \begin{itemize}
    \item Scaling symmetry: if $(r,\Omega^{2},\phi)(u,v)$ is a solution, then so is $(\lambda r,\Omega^{2},\phi)\left(\frac{u}{\lambda},\frac{v}{\lambda}\right)$ for any $0\neq \lambda\in\mathbb{R}$;
    \item Translation symmetry: if $(r,\Omega^{2},\phi)(u,v)$ is a solution, then so is $(r,\Omega^{2},\phi+a)(u,v)$ for any $a\in\mathbb{R}$.
\end{itemize}
The $k$-self-similarity considered in~\cite{christodoulou1994examples} is to assume that the solution is invariant under the above two symmetries with $a = -k\log\lambda$: \begin{equation}
    r(u,v) = \lambda r\left(\frac{u}{\lambda},\frac{v}{\lambda}\right),\quad \Omega^{2}(u,v) = \Omega^{2}\left(\frac{u}{\lambda},\frac{v}{\lambda}\right),\quad \phi(u,v) = \phi\left(\frac{u}{\lambda},\frac{v}{\lambda}\right)-k\log\lambda.
\end{equation}
As a consequence of the $k$-self-similarity, the solution $(r,\Omega^{2},\phi)$ takes the following special form \begin{equation}
    r(u,v) = (-u)\mr{r}(z),\quad \Omega^{2}(u,v) = \mr{\Omega}^{2}(z),\quad \phi(u,v) = \mr{\phi}(z)-k\log(-u),\label{eq:intro:k self-similarity}
\end{equation}
where $z: = \frac{v}{(-u)}$ is called the self-similar coordinate. From the form~\eqref{eq:intro:k self-similarity}, one can further fix the gauge choice of the center $\Gamma: = \{r = 0\} = \{\mr{r} = 0\}$ to be $\{z = -1\}$.

The Einstein--scalar field equations~\eqref{eq: intro: esf1}--\eqref{eq: intro: ESF2} then reduce to a system of ODEs for $(\mr{r},\mr{\Omega}^{2},\mr{\phi})$. In his original construction of naked singularities~\cite{christodoulou1994examples}, Christodoulou introduced two further renormalized variables, thereby reducing this system to a two-dimensional autonomous ODE system. A detailed analysis of the resulting phase portrait then yields the naked singularity solutions for $0<k^{2}<1/3$, and the singularity $\mathcal{O}$ corresponds to $(u,v) = (0,0)$. One should note that although we have presented the main ideas of Christodoulou's construction in double-null coordinates, the original argument in~\cite{christodoulou1994examples} was carried out in the so-called Bondi coordinates; see~\cite{singh2024construction} for the corresponding formulation and argument in double-null coordinates.

Notably, similar to the spacetimes in Theorem~\ref{thm: rough version for whole spacetime} and Theorem~\ref{intro: rough version of the exterior construction}, the $k$-self-similar naked singularity spacetime in~\cite{christodoulou1994examples} arises from non-smooth initial data with $C^{1,\frac{k^{2}}{1-k^{2}}}$-Hölder regularity of the scalar field on the initial outgoing cone emanating from the center.

However, the above approach does not apply to the Einstein--Maxwell--charged scalar field equations, since the nonzero scalar-field charge $q_{0}\neq0$ breaks both the translation and scaling symmetries. Nevertheless, to construct the interior solutions of the charged naked singularities in Theorem~\ref{thm: rough version for whole spacetime}, {we first apply Christodoulou's approach to the corresponding Einstein--Maxwell system with zero Maxwell field and an uncharged scalar field.} In this setting, this system reduces to~\eqref{eq: intro: esf1}--\eqref{eq: intro: ESF2}, with the scalar field $\phi$ now being complex-valued. In particular, the scaling and translation symmetries remain available, and hence one can define an analogous notion of $k$-self-similarity; see Definition~\ref{def: k1k2 self-similarity}. Although the naked singularity spacetimes constructed in Theorem~\ref{thm: rough version for whole spacetime} are genuinely charged, the contributions from the Maxwell field $F$ are sub-leading when compared to the asymptotic behavior of $(r,\Omega,\phi)$ near the singularity. This allows us to construct the interior solutions in Theorem~\ref{thm: rough version for whole spacetime} perturbatively around the continuously self-similar solutions of the Einstein--Maxwell system with zero Maxwell field and an uncharged scalar field; see Section~\ref{sec: setup of the perturbations} for an overview of this perturbative argument.
\subsubsection{Naked singularity exteriors for the Einstein vacuum equations constructed in~\cite{rodnianski2023naked}}
{Christodoulou's pioneering construction in~\cite{christodoulou1994examples} is closely tied to the specific structure of the Einstein-scalar field system, exploiting both continuous self-similarity and spherical symmetry to reduce the equations to a $2\times2$ system of ODEs. Consequently, the approach does not directly extend to the Einstein equations coupled to other matter models or to the Einstein vacuum equations.} In the vacuum setting, however, Birkhoff's theorem implies that imposing spherical symmetry eliminates all nontrivial dynamical degrees of freedom. Moreover, the notion of $k$-self-similarity in~\cite{christodoulou1994examples} relies on the translation symmetry of the scalar field and therefore has no direct analogue in vacuum. The seminal work of Rodnianski and Shlapentokh-Rothman~\cite{rodnianski2023naked} introduced a geometric notion of $k$-self-similarity, also referred to as \emph{twisted self-similarity}, and developed a more robust framework for constructing the exterior region of naked singularity spacetimes. Roughly speaking, their approach is to prescribe suitable characteristic initial data on $\underline{C}\cup(C_{out}\cap\mathcal{R}_{ext})$ that obey the bound suggested by the twisted self-similarity and already contain a curvature singularity along the initial ingoing cone $\underline{C}$, and then establish global existence toward the future for the resulting hyperbolic evolution. 

For the exterior constructions in Theorem~\ref{thm: rough version for whole spacetime} and Theorem~\ref{intro: rough version of the exterior construction}, we adapt and generalize the approach of~\cite{rodnianski2023naked}; see Section~\ref{sec: overview of the exterior construction} for details. The exterior spacetime in Theorem~\ref{thm: rough version for whole spacetime} can also be constructed perturbatively, following the approach described in Section~\ref{intro: sec: chris original construction}. We emphasize, however, that the exterior spacetime constructed in Theorem~\ref{intro: rough version of the exterior construction} is genuinely non-perturbative.

\subsection{Naked singularities arising from the critical collapse}
\label{sec: intro: critical collapse}
Another important source of naked singularities in gravitational collapse arises from critical phenomena at the threshold of black-hole formation. For the Einstein--scalar field equations, it is known that sufficiently small initial data lead to dispersion~\cite{chris86}, whereas sufficiently large initial data can lead to black-hole formation~\cite{chris91}. In his seminal numerical work~\cite{choptuik1}, Choptuik investigated the transition between small-data dispersion and large-data black-hole formation for the spherically symmetric Einstein--scalar field equations and discovered a naked singularity spacetime at the threshold, a phenomenon now referred to as critical collapse. In contrast to the spacetimes constructed in~\cite{christodoulou1994examples}, the naked singularity spacetime in~\cite{choptuik1} arises from smooth initial data and exhibits discrete, rather than continuous, self-similarity. {Beyond these differences in regularity and symmetry, the naked singularity spacetimes in~\cite{choptuik1} and~\cite{christodoulou1994examples} are also genuinely distinct from the perspective of stability. With the gauge of the initial geometry fixed, the former is expected to be unstable under rescaling of the scalar-field initial data, whereas the latter have been shown to be stable~\cite{zheng2026nonlinear,singh2026nonlinear}, thereby ruling out a critical-collapse interpretation for the naked singularities constructed in~\cite{christodoulou1994examples}.} The exterior region of a discretely self-similar naked singularity spacetime for the Einstein--scalar field equations was constructed rigorously in~\cite{cicortas2024discretely}.

Critical phenomena have also been studied in the presence of electromagnetic charge. For the Einstein--Maxwell-charged scalar field equations, small-data dispersion~\cite{chae2003global,kauffman2023global} and large-data black hole formation~\cite{an2022trapped,shen2025formation} have also been established mathematically. The work~\cite{gundlach1996charge} studied the transition between dispersion and black-hole formation and showed that the critical behavior remains governed by the uncharged Choptuik solution. In particular, the Maxwell field becomes sub-leading in the critical limit, which is reminiscent of the behavior of the Maxwell field in our Theorem~\ref{thm: rough version for whole spacetime}. This picture was subsequently confirmed numerically in~\cite{hod1997critical}. The rigorous mathematical construction of these charged naked singularity spacetimes from the critical collapse remains an interesting open problem.
\subsection{Stability and instability of naked singularities}
\label{intro: sec: stability}
A fundamental question concerning any construction of naked singularities is whether the resulting spacetime is stable under perturbations of the initial data. This issue is closely tied to weak cosmic censorship, whose genericity formulation allows the existence of exceptional naked singularity solutions while requiring that such behavior be destroyed by sufficiently general perturbations. 

For the Einstein--scalar field equations, Christodoulou proved that the naked singularity spacetimes arising from his construction~\cite{christodoulou1994examples} are unstable under sufficiently rough spherically symmetric exterior perturbations~\cite{chris99}. Subsequent works~\cite{li2025interior,liuli,liuli_outsidesymm,an_highcodim,singh2} further refined the regularity and support conditions on the perturbations and established the following instability results:
\begin{itemize}
\item For generic spherically symmetric exterior perturbations with Hölder regularity below that of the initial data in~\cite{christodoulou1994examples}, the naked singularities are unstable~\cite{singh2,liuli};
\item For generic spherically symmetric perturbations supported in the interior and with Hölder regularity below that of the initial data in~\cite{christodoulou1994examples}, the naked singularities are unstable~\cite{li2025interior};
\item For generic rough exterior perturbations outside spherical symmetry, the naked singularities are unstable~\cite{liuli_outsidesymm,an_highcodim}.
\end{itemize}
Although the perturbations considered in the above instability results are rougher than the background spacetime, they are still regarded as ``regular'' in view of the local well-posedness result for the spherically symmetric Einstein-scalar field equations established in~\cite{chris93}, which applies to scalar field initial data whose derivatives belong to the space of functions of bounded variation.

However, additional subtleties arise when one considers more regular perturbations. It was shown in~\cite{singh2026nonlinear,zheng2026nonlinear} that the naked singularities in~\cite{christodoulou1994examples} are stable for all initial data in a small open neighborhood of the background data, with respect to the topology of a localized Hölder space capturing the sharp regularity of the background initial data. This leaves open the question of the precise functional framework in which WCC should be formulated; see also~\cite{an2026naked} for partial stability results under sufficiently regular exterior perturbations outside of spherical symmetry.

For the naked singularity spacetimes constructed in Theorem~\ref{thm: rough version for whole spacetime} and Theorem~\ref{intro: rough version of the exterior construction}, \textbf{stability under sufficiently regular spherically symmetric exterior perturbations follows directly from our construction.} For the spacetime constructed in Theorem~\ref{thm: rough version for whole spacetime}, we further expect an instability result analogous to~\cite{chris99} and a stability result analogous to~\cite{zheng2026nonlinear}. We leave a systematic study of these stability and instability properties to future work.

\paragraph{Acknowledgments}
The author would like to thank his advisor, Maxime Van de Moortel, for his kind support, patience, and valuable discussions. The author thanks Jaydeep Singh for many enlightening discussions. The author gratefully acknowledges the support
from the NSF Grant DMS-2247376.
\paragraph{Declaration on AI use.}
All the mathematical work in the current paper was completed at a time when AI tools were considerably less capable than they are today. All mathematical research and arguments were carried out independently by the author. AI was used only to improve the exposition of the manuscript.
\section{Preliminaries}
\subsection{Reduction of the Einstein--Maxwell-charged scalar field equations under spherical symmetry}
\label{sec: reduction of EMcsf}
Recall that the Einstein--Maxwell-charged scalar field equations take the form \begin{align}
&Ric_{\mu\nu}(g)-\frac{1}{2}R(g)g_{\mu\nu} = T_{\mu\nu}^{EM}+T_{\mu\nu}^{SC},\label{eq:whole equation1}\\&
T_{\mu\nu}^{EM} = 2\left(g^{\alpha\beta}F_{\alpha\mu}F_{\beta\nu}-\frac{1}{4}F^{\alpha\beta}F_{\alpha\beta}g_{\mu\nu}\right),\quad \nabla^{\mu}F_{\mu\nu} = iq_{0}\left(\frac{\phi\overline{D_{\nu}\phi}-\overline{\phi}D_{\nu}\phi}{2}\right),\ F=dA,\\&
T^{SC}_{\mu\nu} = 2\left(\Re{\left(D_{\mu}\phi\overline{D_{\nu}\phi}\right)}-\frac{1}{2}\left(g^{\alpha\beta}D_{\alpha}\phi\overline{D_{\beta}\phi}\right)g_{\mu\nu}\right),\quad D_{\mu} = \nabla_{\mu}+iq_{0}A_{\mu},\\&
g^{\mu\nu}D_{\mu}D_{\nu}\phi = 0,\label{eq:whole equation5}
\end{align}
where $g$ is the metric of the spacetime $\mathcal{M}$, $\phi$ is the scalar field, $F$ is the electromagnetic field, $A$ is the electromagnetic potential, and $q_{0}$ is the scalar field charge. Sometimes we refer to $q_{0}$ as the coupling constant. In spherical symmetry, we use double-null coordinates $(u,v,\theta,\varphi)$ in which the metric takes the form
\begin{equation*}
g = -\frac{1}{2}\Omega^{2}du\otimes dv-\frac{1}{2}\Omega^{2}dv\otimes du+r^{2}(u,v)d\sigma^{2},
\end{equation*}
where $d\sigma^{2} = d\theta^{2}+\sin^{2}\theta d\varphi^{2}$ is the standard sphere metric and $r = r(u,v)$ is the area radius function. We then consider the quotient manifold $\mathcal{Q} = \mathcal{M}/SO(3)$, equipped with the coordinates $(u,v)$ and metric \begin{equation}
g_{\mathcal{Q}} = -\frac{1}{2}\Omega^{2}du\otimes dv-\frac{1}{2}\Omega^{2}dv\otimes du.\label{eq: pre: global double-null gauge}
\end{equation}
On the quotient manifold $\mathcal{Q}$, we further assume the electromagnetic field $F$ takes the form \begin{equation}
F = \frac{Q(u,v)\Omega^{2}(u,v)}{2r^{2}}du\wedge dv,\label{def: form of Maxwell field}
\end{equation}
where the function $Q(u,v)$ is called the spacetime charge. $F$ and the one-form $A$ are related by \begin{equation}
F_{uv} = \partial_{u}A_{v}-\partial_{v}A_{u},\quad A = A_{u}du+A_{v}dv.\label{eq:form of Maxwell potential}
\end{equation}
Using the spherical symmetry, the form of the electromagnetic field $F$ \eqref{def: form of Maxwell field}, and the form of the electromagnetic potential $A$ \eqref{eq:form of Maxwell potential}, the Einstein--Maxwell-charged scalar field equations \eqref{eq:whole equation1}-\eqref{eq:whole equation5} can be reduced to:\begin{align}
\partial_{u}\left(\frac{\partial_{u}r}{\Omega^{2}}\right) &= -\frac{r}{\Omega^{2}}\left\vert D_{u}\phi\right\vert^{2},\label{eq:spherical symmetric equaion1}\\
\partial_{v}\left(\frac{\partial_{v}r}{\Omega^{2}}\right)& = -\frac{r}{\Omega^{2}}\left\vert D_{v}\phi\right\vert^{2},\label{eq:dv-Ray equation}\\
r\partial_{u}\partial_{v}r +\partial_{u}r\partial_{v}r & = -\frac{\Omega^{2}}{4}\left(1-\frac{Q^{2}}{r^{2}}\right),\label{eq:wave equation for r}\\
\partial_{u}\partial_{v}\log\Omega^{2}& = \frac{2\partial_{u}r\partial_{v}r}{r^{2}}+\frac{\Omega^{2}}{2r^{2}}\left(1-\frac{2Q^{2}}{r^{2}}\right)-2\Re{\left(D_{u}\phi\overline{D_{v}\phi}\right)},\label{eq:wave equation for omega}\\D_{u}D_{v}\phi+\frac{\partial_{u}rD_{v}\phi}{r}+\frac{\partial_{v}rD_{u}\phi}{r}& =\frac{iq_{0}Q\Omega^{2}}{4r^{2}}\phi,\label{eq:wave equation for phi},\\
\partial_{u}Q& = -q_{0}r^{2}\Im{\left(\phi\overline{D_{u}\phi}\right)},\label{eq:u-transport equation for Q}\\\partial_{v}Q &= q_{0}r^{2}\Im{\left(\phi\overline{D_{v}\phi}\right)},\label{eq:v-transport equation for Q}
\\\
\partial_{u}A_{v}-\partial_{v}A_{u}& = \frac{Q\Omega^{2}}{2r^{2}}.\label{eq:spherical symmetric equation last}
\end{align}
The equations \eqref{eq:spherical symmetric equaion1}-\eqref{eq:dv-Ray equation} are called the Raychaudhuri equations, the equations \eqref{eq:wave equation for r}-\eqref{eq:wave equation for phi} are the wave equations for $r$, $\Omega^{2}$, and $\phi$, respectively, the equations \eqref{eq:u-transport equation for Q}-\eqref{eq:v-transport equation for Q} are the transport equations for the spacetime charge $Q$, and the equation \eqref{eq:spherical symmetric equation last} comes from the relation between $F$ and $A$.

Note that the Einstein--Maxwell-charged scalar field equations under spherical symmetry \eqref{eq:spherical symmetric equaion1}-\eqref{eq:spherical symmetric equation last} are invariant under the gauge transformation\begin{equation}
A_{new} = A-d\chi,\quad \phi_{new} = e^{iq_{0}\chi}\phi,\label{eq: gauge invariance of the EMcsf}
\end{equation}
for any real-valued smooth function $\chi$. Hence, we can make the following gauge choice:\begin{equation}
A_{v}\equiv 0.\label{eq: gauge choice of A}
\end{equation}
Under this gauge choice of $A$, we can rewrite the wave equation for $\phi$ \eqref{eq:wave equation for phi} as:\begin{equation}
r\partial_{u}\partial_{v}\phi+\partial_{u}r\partial_{v}\phi+\partial_{v}r\partial_{u}\phi+iq_{0}\left(A_{u}\partial_{v}(r\phi)-\frac{Q\Omega^{2}}{4r}\phi\right) = 0.
\label{eq: pre wave equation for phi}
\end{equation}
We decompose the scalar field $\phi$ into its real and imaginary parts $\phi = \phi_{1}+i\phi_{2}$. Then \begin{align}
r\partial_{u}\partial_{v}\phi_{1}+\partial_{u}r\partial_{v}\phi_{1}+\partial_{v}r\partial_{u}\phi_{1} = &q_{0}\left(A_{u}\partial_{v}(r\phi_{2})-\frac{Q\Omega^{2}}{4r}\phi_{2}\right),\\
r\partial_{u}\partial_{v}\phi_{2}+\partial_{u}r\partial_{v}\phi_{2}+\partial_{v}r\partial_{u}\phi_{2}=&-q_{0}\left(A_{u}\partial_{v}(r\phi_{1})-\frac{Q\Omega^{2}}{4r}\phi_{1}\right).
\end{align}

Instead of working with $\Omega^{2}$, it is often useful to work with the Hawking mass $m$ and the mass ratio $\mu$:\begin{align}
m:&=\frac{r}{2}\left(1-g_{Q}(\nabla r,\nabla r)\right) = \frac{r}{2}\left(1+\frac{4\partial_{u}r\partial_{v}r}{\Omega^{2}}\right),\\\mu:&=\frac{2m}{r} = 1+\frac{4\partial_{u}r\partial_{v}r}{\Omega^{2}}.
\end{align}
For the Einstein--Maxwell-charged scalar field equations, we also define the following renormalized Hawking mass: \begin{equation}
    \varpi: = m+\frac{Q^{2}}{2r}.
\end{equation}

It is straightforward to derive the following transport equations for $m$, $\mu$, and $\varpi$ from \eqref{eq:whole equation1}-\eqref{eq:whole equation5}:\begin{align}
\partial_{u}m+\frac{r}{\partial_{u}r}\left\vert D_{u}\phi\right\vert^{2}m &= \frac{1}{2}\left(\frac{r^{2}}{\partial_{u}r}\left\vert D_{u}\phi\right\vert^{2}+\frac{\partial_{u}r}{r^{2}}Q^{2}\right),\label{eq:u-transport equation for m}\\
\partial_{v}m+\frac{r}{\partial_{v}r}\left\vert D_{v}\phi\right\vert^{2}m&=\frac{1}{2}\left(\frac{r^{2}}{\partial_{v}r}\left\vert D_{v}\phi\right\vert^{2}+\frac{\partial_{v}r}{r^{2}}Q^{2}\right),\label{eq: v-transport equation for m}\\
\partial_{u}\mu+\left(\frac{\partial_{u}r}{r}+\frac{r}{\partial_{u}r}\left\vert D_{u}\phi\right\vert^{2}\right)\mu& = \frac{r}{\partial_{u}r}\left\vert D_{u}\phi\right\vert^{2}+\frac{\partial_{u}r}{r^{3}}Q^{2},\\
\partial_{v}\mu+\left(\frac{\partial_{v}r}{r}+\frac{r}{\partial_{v}r}\left\vert D_{v}\phi\right\vert^{2}\right)\mu& = \frac{r}{\partial_{v}r}\left\vert D_{v}\phi\right\vert^{2}+\frac{\partial_{v}r}{r^{3}}Q^{2}\label{eq: v-transport equation for mu},\\
\partial_{u}\varpi+\frac{r}{\partial_{u}r}\vert D_{u}\phi\vert^{2}\varpi&=-q_{0}Qr\Im\left(\phi\overline{D_{u}\phi}\right)+\frac{\vert D_{u}\phi\vert^{2}Q^{2}}{2\partial_{u}r}+\frac{1}{2}\frac{r^{2}}{\partial_{u}r}\vert D_{u}\phi\vert^{2},\\
\partial_{v}\varpi+\frac{r}{\partial_{v}r}\vert D_{v}\phi\vert^{2}\varpi&=q_{0}Qr\Im\left(\phi\overline{D_{v}\phi}\right)+\frac{\vert D_{v}\phi\vert^{2}Q^{2}}{2\partial_{v}r}+\frac{1}{2}\frac{r^{2}}{\partial_{v}r}\vert D_{v}\phi\vert^{2}.
\end{align}
Throughout the paper, we will use the following convention:
\begin{convention}
For the scalar field $\phi$, we use the notation:\begin{equation*}
\phi_{1}: = \Re{\phi},\quad \phi_{2}: = \Im{\phi},\quad \phi = \rho e^{i\theta}.
\end{equation*}
\end{convention}
\subsection{Einstein--Maxwell-uncharged scalar field equations and $(k_{1},k_{2})$-self-similarity}
\subsubsection{Einstein--Maxwell-uncharged scalar field equations under spherical symmetry}
In the construction of the interior solution of the naked singularity, we first take the coupling constant $q_{0} = 0$. Under the regular center assumption, i.e., $Q\equiv0$ at the axis, we have that $Q\equiv0$ in the whole spacetime when $q_{0} = 0$. Similar to Section~\ref{sec: reduction of EMcsf}, under the spherically symmetric assumption
\begin{equation*}
    g = -\frac{1}{2}\hat{\Omega}^{2}d\hat u\otimes d\hat v-\frac{1}{2}\hat{\Omega}^{2}d\hat v\otimes d\hat u+r^{2}(d\theta^{2}+\sin^{2}\theta d\varphi^{2}),
\end{equation*}
we can reduce the Einstein--Maxwell-uncharged scalar field equations to the following equations:\begin{align}
    \partial_{\hat u}\left(\frac{\partial_{\hat u}r}{\hat\Omega^{2}}\right) =& -\frac{r}{\hat\Omega^{2}}\left\vert\partial_{\hat u}\phi\right\vert^{2},\label{eq: uncharged u Ray equation}\\
     \partial_{\hat v}\left(\frac{\partial_{\hat v}r}{\hat\Omega^{2}}\right)=&-\frac{r}{\hat\Omega^{2}}\left\vert\partial_{\hat v}\phi\right\vert^{2},\label{eq: uncharged v Ray equation}\\ r\partial_{\hat u}\partial_{\hat v}r +\partial_{\hat u}r\partial_{\hat v}r=&-\frac{\hat\Omega^{2}}{4},\label{eq: uncharged r wave equation},\\
     \partial_{\hat{u}}\partial_{\hat{v}}\log\hat{\Omega}^{2}=&\frac{2\partial_{\hat{u}}r\partial_{\hat{v}}r}{r^{2}}+\frac{\hat{\Omega}^{2}}{2r^{2}}-2\Re\left(D_{\hat u}\phi\overline{D_{\hat v}\phi}\right),\\r\partial_{\hat u}\partial_{\hat v}\phi_{1}+\partial_{\hat u}r\partial_{\hat v}\phi_{1}+\partial_{\hat v}r\partial_{\hat u}\phi_{1}=&0,\label{eq: uncharged phi1 wave equation}\\r\partial_{\hat u}\partial_{\hat v}\phi_{2}+\partial_{\hat u}r\partial_{\hat v}\phi_{2}+\partial_{\hat v}r\partial_{\hat u}\phi_{2}=&0,\label{eq: uncharged phi2 wave equation}\\
     \partial_{\hat u}m+\frac{r}{\partial_{\hat u}r}\vert\partial_{\hat{u}}\phi\vert^{2} m=&\frac{1}{2}\frac{r^{2}}{\partial_{\hat u}r}\left\vert\partial_{\hat u}\phi\right\vert^{2},\label{eq: uncharged du m equation}\\\partial_{\hat v}m+\frac{r}{\partial_{\hat v}r}\vert\partial_{\hat{v}}\phi\vert^{2}m=&\frac{1}{2}\frac{r^{2}}{\partial_{\hat v}r}\left\vert\partial_{\hat v}\phi\right\vert^{2}.\label{eq: uncharged dvm equation}
\end{align}
\begin{remark}
    The above equations~\eqref{eq: uncharged u Ray equation}--\eqref{eq: uncharged dvm equation} closely resemble the spherically symmetric Einstein–scalar field equations, except that the scalar field $\phi = \phi_{1}+i\phi_{2}$ is complex-valued rather than necessarily real-valued.

\end{remark}
\subsubsection{Notion of $(k_{1},k_{2})$-self-similarity}
The spherically symmetric Einstein--Maxwell-uncharged scalar field equations \eqref{eq: uncharged u Ray equation}-\eqref{eq: uncharged dvm equation} are invariant under the following scaling and translation:\begin{equation}
\begin{aligned}
    &r(\hat u,\hat v)\rightarrow ar\left(\frac{\hat u}{a},\frac{\hat v}{a}\right),\quad \hat\Omega^{2}(\hat u,\hat v)\rightarrow \hat\Omega^{2}\left(\frac{\hat u}{a},\frac{\hat v}{a}\right),\\& \phi_{1}(\hat u,\hat v)\rightarrow \phi_{1}\left(\frac{\hat u}{a},\frac{\hat v}{a}\right)+b_{1},\quad \phi_{2}(\hat u,\hat v)\rightarrow \phi_{2}\left(\frac{\hat u}{a},\frac{\hat v}{a}\right)+b_{2},
    \end{aligned}
\end{equation}
where $a,b_{1},b_{2}\in\mathbb{R}$. This motivates us to introduce the following notion of $(k_{1},k_{2})$-self-similarity:
\begin{definition}
\label{def: k1k2 self-similarity}
    We say that a solution $(r,\hat\Omega^{2},\phi,m)$ to \eqref{eq: uncharged u Ray equation}-\eqref{eq: uncharged dvm equation} is $(k_{1},k_{2})$-self-similar if \begin{equation}
    \label{eq: k1k2 self simiar}
        \begin{aligned}
            &r(\hat u,\hat v) = ar\left(\frac{\hat u}{a},\frac{\hat v}{a}\right),\quad \hat\Omega^{2}(\hat u,\hat v) = \hat\Omega^{2}\left(\frac{\hat u}{a},\frac{\hat v}{a}\right),\quad m(\hat u,\hat v) = am\left(\frac{\hat u}{a},\frac{\hat v}{a}\right),\\&
            \phi_{1}(\hat u,\hat v) = \phi_{1}\left(\frac{\hat u}{a},\frac{\hat v}{a}\right)-k_{1}\log a,\quad \phi_{2}(\hat u,\hat v) = \phi_{2}\left(\frac{\hat u}{a},\frac{\hat v}{a}\right)-k_{2}\log a,
        \end{aligned}
    \end{equation}
    for $a\in\mathbb{R}_{+}$ and $(k_{1},k_{2})\in\mathbb{R}\times \mathbb{R}$.
\end{definition}
A direct consequence of the $(k_{1},k_{2})$-self-similarity is the following.
\begin{lemma}
If a solution $(r,\hat\Omega^{2},\phi,m)$ to \eqref{eq: uncharged u Ray equation}-\eqref{eq: uncharged dvm equation} is $(k_{1},k_{2})$-self-similar, then there exist functions $\mathring{r}(\hat z)$, $\mathring{\hat\Omega}(\hat z)$, $\mr{m}(\hat z)$, $\mr{\hat\phi}_{1}(\hat z)$, and $\mr{\hat\phi}_{2}(\hat z)$, such that \begin{equation}
\begin{aligned}
    &r(\hat u,\hat v) = (-\hat u)\mr{r}(\hat z),\quad \hat\Omega^{2}(\hat u,\hat v) = \mr{\hat\Omega}^{2}(\hat z),\quad m(\hat u,\hat v) = (-\hat u)\mr{m}(\hat z),\\& \phi_{1}(\hat u,\hat v) = \mr{\hat\phi}_{1}(\hat z)-k_{1}\log(-\hat u),\quad \phi_{2}(\hat u,\hat v) =\mr{\hat\phi}_{2}(\hat z)-k_{2}\log(-\hat u),
    \end{aligned}
\end{equation} 
where $\hat z: = \frac{\hat v}{(-\hat u)}$ is called the self-similar coordinate. Moreover, we have the following identities \begin{align}
    &\partial_{\hat u}r(\hat u,\hat v) = -\mr{r}(\hat z)+\hat z\frac{d\mr{r}}{d\hat z}(\hat z)=:\hat\nu(\hat u,\hat v) = \mr{\hat\nu}(\hat z),\label{eq: expression of dur}\\&
    \partial_{\hat v}r(\hat u,\hat v) = \frac{d\mr{r}}{d\hat z}(\hat z) = :\hat \lambda(\hat u,\hat v) = \mr{\hat\lambda}(\hat z),\label{eq: expression of dvr}\\&
    \mr{r}+\mr{\hat \nu} = \hat z\mr{\hat\lambda}.\label{eq: algebraic relation of derivatives of r}
\end{align}
\end{lemma}
\begin{proof}
This can be proved by taking $a = (-\hat u)$ in \eqref{eq: k1k2 self simiar}. A straightforward computation gives \eqref{eq: expression of dur}-\eqref{eq: algebraic relation of derivatives of r}.
\end{proof}
Under the $(k_{1},k_{2})$-self-similarity, we can further reduce the equations \eqref{eq: uncharged u Ray equation}-\eqref{eq: uncharged dvm equation} to \begin{align}
    \hat z\mr{r}\frac{d\mr{\hat\lambda}}{d\hat z} =& -\left(\frac{1}{4}\mr{\hat\Omega}^{2}+\mr{\hat \lambda}\mr{\hat \nu}\right),\label{eq: dzlambda}\\\mr{r}\frac{d\mr{\hat \nu}}{d\hat z}=&-\left(\frac{1}{4}\mr{\hat\Omega}^{2}+\mr{\hat\lambda}\mr{\hat\nu}\right),\label{eq:dznu}\\2\mr{\hat\Omega}^{-1}\mr{\hat\nu}\hat z\frac{d\mr{\hat\Omega}}{d\hat z}=&\hat z\frac{d\mr{\hat\nu}}{d\hat z}+\mr{r}\left(\hat z\mr{\hat\phi}_{1}^{\prime}+k_{1}\right)^{2}+\mr{r}\left(\hat z\mr{\hat \phi}_{2}^{\prime}+k_{2}\right)^{2},\label{eq: first equation of dzomega}\\2\mr{\hat\Omega}^{-1}\mr{\hat\lambda}\frac{d\mr{\hat\Omega}}{d\hat z}=&\frac{d\mr{\hat\lambda}}{d\hat z}+\mr{r}(\mr{\hat\phi}_{1}^{\prime})^{2}+\mr{r}(\mr{\hat\phi}_{2}^{\prime})^{2},\label{eq: second equation of dzomega}\\\mr{r}\hat z\frac{d\mr{\hat\phi}_{1}^{\prime}}{d\hat z}=&-2\mr{\hat{\lambda}}\hat z\mr{\hat\phi}_{1}^{\prime}-k_{1}\mr{\hat\lambda},\label{eq: dzphi1 equation}\\\mr{r}\hat z\frac{d\mr{\hat\phi}_{2}^{\prime}}{d\hat z}=&-2\mr{\hat \lambda}\hat z\mr{\hat\phi}_{2}^{\prime}-k_{2}\mr{\hat\lambda}\label{eq: dzphi2 equation},
\end{align}
where $\mr{\hat\phi}_{1}^{\prime}: = \frac{d}{d\hat z}\mr{\hat\phi}_{1}$ and $\mr{\hat{\phi}}_{2}^{\prime}: = \frac{d}{d\hat z}\mr{\hat\phi}_{2}$. The equations \eqref{eq: dzlambda}-\eqref{eq:dznu} are from the wave equation for $r$ \eqref{eq: uncharged r wave equation}; the equation \eqref{eq: first equation of dzomega} is from the Raychaudhuri equation \eqref{eq: uncharged u Ray equation}; the equation \eqref{eq: second equation of dzomega} is from the Raychaudhuri equation \eqref{eq: uncharged v Ray equation}; the equation for $\phi_{1}$ \eqref{eq: dzphi1 equation} is from the wave equation for $\phi_{1}$ \eqref{eq: uncharged phi1 wave equation}; the equation for $\phi_{2}$ \eqref{eq: dzphi2 equation} is from the wave equation for $\phi_{2}$ \eqref{eq: uncharged phi2 wave equation}.

Combining the equations \eqref{eq: dzlambda}-\eqref{eq: second equation of dzomega}, we can derive the following algebraic relation: \begin{equation}
    \frac{1}{4}\mr{\hat\Omega}^{2}+\mr{\hat\lambda}\mr{\hat\nu} = \hat z\mr{r}^{2}(\mr{\hat\phi}_{1}^{\prime})^{2}+\hat z\mr{r}^{2}(\mr{\hat\phi}_{2}^{\prime})^{2}+2k_{1}\hat z\mr{r}\mr{\hat\lambda}\mr{\hat\phi}_{1}^{\prime}+2k_{2}\hat z\mr{r}\mr{\hat\lambda}\mr{\hat\phi}_{2}^{\prime}+(k_{1}^{2}+k_{2}^{2})\mr{r}\mr{\hat\lambda}.
\end{equation}
We can further fix the gauge condition that the center $\Gamma$ corresponds to the set $\{\hat z = -1\}$. Then the equations \eqref{eq: dzlambda}-\eqref{eq: dzphi2 equation} imply the following boundary conditions at the center.
\begin{proposition}
    Under the gauge choice that $\mr{r}(-1) = 0$ and $\mr{\hat\Omega}(-1) = 1$, if the center is regular, then we have the following condition on $\mr{\hat\lambda}(-1)$, $\mr{\hat\nu}(-1)$, $\mr{\hat\phi}_{1}^{\prime}(-1)$, and $\mr{\hat\phi}_{2}^{\prime}(-1)$:\begin{equation}
        \mr{r}(-1) = 0,\ \mr{\hat\Omega}(-1) = 1,\ \mr{\hat\lambda}(-1) = \frac{1}{2},\ \mr{\hat\nu}(-1) = -\frac{1}{2},\ \mr{\hat\phi}_{1}^{\prime}(-1) = \frac{1}{2}k_{1},\ \mr{\hat\phi}_{2}^{\prime}(-1) = \frac{1}{2}k_{2}.
    \end{equation}
\end{proposition}
\begin{proof}
    By the relation \eqref{eq: algebraic relation of derivatives of r}, at the regular center, we have that \begin{equation*}
        \mr{\hat\nu}(-1)+\mr{\hat\lambda}(-1) = 0.
    \end{equation*}
    Since the left-hand side of \eqref{eq: dzlambda} vanishes at the regular center, we have that \begin{equation*}
        \mr{\hat\nu}(-1)\mr{\hat\lambda}(-1) = -\frac{1}{4}\mr{\hat\Omega}^{2} = -\frac{1}{4}.
    \end{equation*}
    Hence, we have that \begin{equation*}
        \mr{\hat\lambda}(-1) = -\mr{\hat\nu}(-1) = \frac{1}{2}.
    \end{equation*}
    Substituting the regularity condition and the boundary conditions for $\mr{\hat\nu}$ and $\mr{\hat\lambda}$ into the equations \eqref{eq: dzphi1 equation}-\eqref{eq: dzphi2 equation}, we can derive that \begin{equation*}
        \mr{\hat\phi}_{1}^{\prime} = \frac{1}{2}k_{1},\quad \mr{\hat\phi}_{2}^{\prime} = \frac{1}{2}k_{2}.
     \end{equation*}
\end{proof}
Observing that after dividing the equations \eqref{eq: dzphi1 equation} and \eqref{eq: dzphi2 equation} by $k_{1}$ and $k_{2}$ respectively, we have that $\frac{\mr{\hat\phi}_{1}^{\prime}}{k_{1}}$ and $\frac{\mr{\hat\phi}_{2}^{\prime}}{k_{2}}$ satisfy the same equation with the same boundary condition at the center. Hence, we can prove the following proposition:
\begin{proposition}
\label{prop: relation between phi1 and phi2}
    If the solution $(r,\hat\Omega,\phi,m)$ satisfies the $(k_{1},k_{2})$-self-similarity, then we have that \begin{equation}
       \frac{\mr{\hat\phi}_{1}^{\prime}}{k_{1}} = \frac{\mr{\hat\phi}_{2}^{\prime}}{k_{2}}
    \end{equation}
    on the domain of existence.
\end{proposition}
\begin{proof}
    Dividing the equations \eqref{eq: dzphi1 equation} and \eqref{eq: dzphi2 equation} by $k_{1}$ and $k_{2}$ and then subtracting the two resulting equations, we have \begin{equation*}
        \mr{r}\hat z\frac{d}{d\hat z}\left(\frac{\mr{\hat\phi}_{1}^{\prime}}{k_{1}}-\frac{\mr{\hat\phi}_{2}^{\prime}}{k_{2}}\right) = -2\mr{\hat\lambda}\hat z\left(\frac{\mr{\hat\phi}_{1}^{\prime}}{k_{1}}-\frac{\mr{\hat\phi}_{2}^{\prime}}{k_{2}}\right).
    \end{equation*}
    Since at the center \begin{equation*}
       \left( \frac{\mr{\hat\phi}_{1}^{\prime}}{k_{1}}-\frac{\mr{\hat\phi}_{2}^{\prime}}{k_{2}}\right)\Bigg|_{\Gamma} = 0,
    \end{equation*}
    we can conclude that \begin{equation*}
        \frac{\mr{\hat\phi}_{1}^{\prime}}{k_{1}} = \frac{\mr{\hat\phi}_{2}^{\prime}}{k_{2}}
    \end{equation*}
    on the domain of existence.
\end{proof}


\subsection{Geometric conventions}
In double-null coordinates, for convenience in the discussion below, we use the following geometric convention:
\begin{align*}
&\Sigma_{u_{0}} = \{u = u_{0},v\geq 0\},\\&
\underline{\Sigma}_{v_{0}} = \{-1\leq u<0,v = v_{0}\},\\&
\mathcal{Q}_{u_{0},v_{0}} = \{-1\leq u\leq u_{0},0\leq v\leq v_{0}\}.
\end{align*}
We will impose our characteristic initial data on the bifurcating hypersurfaces $\Sigma_{-1}\cup \underline{\Sigma}_{0}$, intersecting at the sphere $S_{-1,0}$. $\Sigma_{u_{0}}$ is called the outgoing null cone and $\underline{\Sigma}_{v_{0}}$ is called the ingoing null cone.
\subsection{Characteristic local well-posedness for the Einstein--Maxwell-charged scalar field system}
In this section, we discuss the free initial data and the local well-posedness result for the spherically symmetric Einstein--Maxwell-charged scalar field equations.

The first local well-posedness result concerns two characteristic hypersurfaces \begin{equation*}
    \underline{\Sigma}_{v_{0}}^{u_{0},u_{1}}: = \{\ u_{0}\leq u\leq u_{1},\ v = v_{0}\},\quad \Sigma_{u_{0}}^{v_{0},v_{1}}:= \{u = u_{0},\ v_{0}\leq v\leq v_{1}\}
\end{equation*}
intersecting at $(u_{0},v_{0})$ with $r(u_{0},v_{0})>0$. We fix the gauge choice of the electromagnetic field $A_{v}\equiv0$. The free data for this problem consist of the gauge choice of $\partial_{u}r$ and $A_{u}$ on the ingoing cone $\underline{\Sigma}_{v_{0}}^{u_{0},u_{1}}$, the gauge choice of $\partial_{v}r$ on the outgoing cone $\Sigma_{u_{0}}^{v_{0},v_{1}}$, the initial data for $\partial_{u}\phi$ on $\underline{\Sigma}_{v_{0}}^{u_{0},u_{1}}$, the initial data for $\partial_{v}\phi$ on $\Sigma_{u_{0}}^{v_{0},v_{1}}$, and the data for $\mu,\phi,r,Q$ on the intersecting sphere $(u_{0},v_{0})$. Without loss of generality, we further fix the gauge choice \begin{equation}
    \partial_{u}r(u,v_{0}) = -1,\quad A_{u}(u,v_{0}) = 0,\quad \partial_{v}r(u_{0},v)
 = 1.\label{eq: gauge choice in LWP}
 \end{equation}
A solution to the Einstein--Maxwell--charged scalar field equations is called a $BV$ solution if it arises from $BV$ characteristic initial data satisfying

$$
\partial_u\phi(u,v_0)\in BV_u,
\qquad
\partial_v\phi(u_0,v)\in BV_v,
$$
and if the quantities $(\partial r,\partial\phi)$ have bounded total variation in terms of the initial data throughout the domain of existence. In view of the transport equations~\eqref{eq:u-transport equation for Q}--\eqref{eq:v-transport equation for Q}, the boundedness and regularity of $Q$ are then determined.

We have the following local well-posedness result.
\begin{proposition}
\label{prop: BVlwp}
    Fixing the gauge choice $A_{v}\equiv 0$ and \eqref{eq: gauge choice in LWP}, for initial data $\partial_{u}\phi(u,v_{0})\in BV_{u},\ \partial_{v}\phi(u_{0},v)\in BV_{v}$, there exists a unique BV solution to the Einstein--Maxwell-charged scalar field equations in the domain \begin{equation*}
        [u_{0},u_{1}]\times[v_{0},v_{0}+\delta)
    \end{equation*}
    for some $\delta$ sufficiently small. Moreover, if the initial data are $C^{1}$, then the solution is also $C^{1}$.
\end{proposition}
\begin{proof}
    See Appendix~\ref{app: BVLWP}.
\end{proof}
\begin{figure}[htbp]
    \centering
    \begin{tikzpicture}[ scale=0.8, >=Latex, every node/.style={font=\small} ]
\coordinate (O) at (0,0);
\coordinate (A) at (-2,2);
\coordinate (B) at (2,2);
\draw (O)--(A) node[midway, below, sloped] {$v=v_{0}$};
\draw (O)--(B) node[midway, below,sloped]{$u = u_{0}$};
\coordinate (C) at (0.5,0.5);
\coordinate (D) at (-1.5,2.5);
\draw (C)--(D) node[midway, above, sloped] {$v=v_{0}+\delta$};
\draw (A)--(D) node[midway, above, sloped] {$u = u_{1}$};
\fill[gray!20] (O)--(A)--(D)--(C)--cycle;
\end{tikzpicture}
    \caption{BV local well-posedness in Proposition~\ref{prop: BVlwp}}
    \label{fig:BV lwp}
\end{figure}
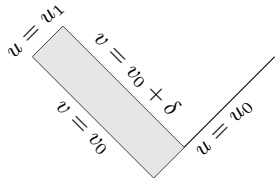

The second local well-posedness result we shall use concerns a mixed timelike--characteristic initial value problem, in which the spherical center $\Gamma$ is fixed by a gauge choice and initial data are prescribed on an outgoing null cone emanating from the center. Fix the gauge choice of the center to be \begin{equation*}
    \Gamma = \{r = 0\} = \{u = v\}.\label{eq: gauge choice of the center in LWP}
\end{equation*}
We assume the initial outgoing cone is \begin{equation}
    \Sigma_{0}^{v} = \{u = 0,\ 0\leq v\leq v_{0}\}.
\end{equation}
We further fix the gauge choice of the electromagnetic field $A_{v}\equiv 0$. Then the regular-center condition requires that on $\Gamma$ we should have \begin{equation*}
     Q|_{\Gamma} = m|_{\Gamma} = \mu|_{\Gamma} = 0.
\end{equation*}
The free data in this setting consist of the gauge choice of $\partial_{v}r(0,v)$ and the initial data for $\partial_{v}(r\phi)(0,v)$. A solution to the Einstein--Maxwell-charged scalar field equations is called $C^{1}$ if the solution is compatible at the center and $(\partial r,\partial\phi)\in C^{1}$ in the domain of existence. We have the following proposition.
\begin{proposition}
\label{prop: C1lwp}
    Fixing the gauge choice of the electromagnetic field $A_{v}\equiv 0$ and $A_{u}|_{\Gamma}$, and the gauge choice of the center~\eqref{eq: gauge choice of the center in LWP}, for initial data $\partial_{v}(r\phi)(0,v)\in C^{1}_{v}$, there exists a unique $C^{1}$ solution to the Einstein--Maxwell-charged scalar field equations in the domain \begin{equation*}
        \{0\leq u<\delta,\ u\leq v\leq v_{0}\}
    \end{equation*}
    for some $\delta$ sufficiently small.
\end{proposition}
\begin{proof}
This follows from an argument analogous to that used to establish the corresponding $C^1$ local well-posedness result for the Einstein--scalar field equations in~\cite{chris86}.
\end{proof}
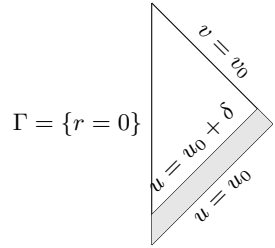
\begin{figure}[htbp]
    \centering
    \begin{tikzpicture}[ scale=0.8, >=Latex, every node/.style={font=\small} ]
\coordinate (O) at (0,0);
\coordinate (A) at (0,4);
\coordinate (B) at (2,2);
\draw (O)--(A) node[midway,left] {$\Gamma = \{r = 0\}$};
\draw (O)--(B) node[midway, below, sloped] {$u = u_{0}$};
\draw (A)--(B) node[midway, above, sloped]{$v = v_{0}$};
\coordinate (C) at (0,0.5);
\coordinate (D) at (1.75,2.25);
\draw (C)--(D) node[midway, above, sloped] {$u = u_{0}+\delta$};
\fill[gray!20] (O)--(B)--(D)--(C)--cycle;
\end{tikzpicture}
    \caption{$C^{1}$ local well-posedness in Proposition~\ref{prop: C1lwp}}
    \label{fig:C1 lwp}
\end{figure}

Note that, in the setting of this mixed timelike--characteristic initial value problem for the Einstein--scalar field equations, Christodoulou~\cite{chris93} proved local existence and uniqueness in the $BV$ class for initial data satisfying $\partial_v(r\phi)\in BV$. We expect an analogous low-regularity local well-posedness result to hold for the Einstein--Maxwell--charged scalar field equations. For the purposes of the present paper, however, the $C^1$ local well-posedness result in Proposition~\ref{prop: C1lwp} is sufficient, and we do not pursue the sharp low-regularity theory here.

The third local well-posedness result concerns a backward mixed timelike-characteristic initial value problem. \begin{proposition}
\label{prop: mixed lwp}
    Fix the gauge choice of the center to be as in~\eqref{eq: gauge choice of the center in LWP} and the gauge choice of the electromagnetic potential $A_{v}\equiv0$. Given compatible $C^{1}$ initial data for $(\partial_{v}r,\partial_{v}\phi)(u_{0},v)$ on the outgoing cone \begin{equation*}
        \{u = u_{0},\ u_{0}\leq v\leq v_{0}\},
    \end{equation*}
    and $(\partial_{u}r,\partial_{u}\phi)$ on the ingoing cone \begin{equation*}
        \{u_{0}\leq u\leq u_{1},\ v=v_{0}\},
    \end{equation*}
    there exists a unique $C^{1}$ solution to the Einstein--Maxwell-charged scalar field equations arising from these initial data in the region \begin{equation*}
        \{u_{0}-\delta\leq u\leq u_{0},\ u\leq v\leq v_{0}\}
    \end{equation*}
    for some $\delta>0$ sufficiently small.
\end{proposition}
\begin{proof}
    By the local well-posedness result in Proposition~\ref{prop: C1lwp}, there exists a positive constant $\delta_{1}$ sufficiently small, for which a $C^{1}$ solution in the region \begin{equation*}
        (u_{0}-\delta_{1},u_{0}]\times[u_{0},v_{0}]
    \end{equation*}
    exists. In particular, the solution in this region gives the value of $(\partial_{u}r,\partial_{u}\phi)$ on the ingoing cone \begin{equation*}
        \{u_{0}-\delta<u\leq u_{0},\ v = u_{0}\}.
    \end{equation*}
    Then using the local well-posedness result in Proposition~\ref{prop: C1lwp}, there exists a unique $C^{1}$ solution to the Einstein--Maxwell-charged scalar field equations in the region\begin{equation*}
        \{u_{0}-\delta_{1}<u\leq v,\ u_{0}-\delta_{2}\leq v\leq u_{0}\}
    \end{equation*}
    for some $\delta_{2}>0$ sufficiently small. Taking $\delta = \min\{\delta_{1},\delta_{2}\}$ concludes the proof.
\end{proof}
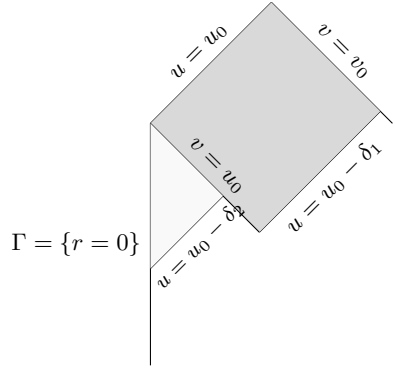
\begin{figure}[htbp]
    \centering
    \begin{tikzpicture}[ scale=0.8, >=Latex, every node/.style={font=\small} ]
\coordinate (O) at (0,0);
\coordinate (A) at (2,2);
\coordinate (B) at (0,-4);
\coordinate (C) at (4,0);
\draw (O)--(A) node[midway, above,sloped] {$u = u_{0}$};
\draw (O)--(B) node[midway, left] {$\Gamma = \{r = 0\}$};
\draw (A)--(C) node[midway, above, sloped] {$v = v_{0}$};
\coordinate (D) at (1.8,-1.8);
\coordinate (E) at (3.8,0.2);
\draw (D)--(E) node[midway,below, sloped]{$u = u_{0}-\delta_{1}$};
\fill[gray!30] (O)--(D)--(E)--(A)--cycle;
\draw (O)--(D) node[midway,above, sloped]{$v = u_{0}$};
\coordinate (F) at (1.2,-1.2);
\coordinate (G) at (0,-2.4);
\draw (F)--(G) node[midway,below,sloped]{\footnotesize{$u = u_{0}-\delta_{2}$}};
\fill[gray!5] (O)--(F)--(G)--cycle;
\end{tikzpicture}
    \caption{Local well-posedness in Proposition~\ref{prop: mixed lwp}}
    \label{fig:mixed lwp}
\end{figure}

\section{Main theorem}
In this section, we state our main theorems. First, the following theorem constructs the interior and exterior of naked singularity spacetimes for the Einstein--Maxwell-charged scalar field equations:
\begin{theorem}
    \label{main theorem1}
    Fix a nonzero scalar field charge $q_{0}\neq0$ and parameters $(k_{1},k_{2})$ with $0<k_{1}^{2}+k_{2}^{2}\ll1$, and let $k: = (k_{1}^{2}+k_{2}^{2})^{\frac{1}{2}}$. There exists a class of asymptotically flat and spherically symmetric naked singularity spacetimes solving the Einstein--Maxwell-charged scalar field equations~\eqref{eq:whole equation1}--\eqref{eq:whole equation5}, covered by a global double-null gauge~\eqref{eq: pre: global double-null gauge}. Let $z: = \frac{v}{(-u)^{q_{k}}}$ where $q_{k}: = 1-k^{2}$. Then the gauge choice of the global double-null coordinates can be renormalized so that the center of the spherical symmetry is given by $\{z= - 1\}$ and the singularity $\mathcal{O}$ is located at $(u,v) = (0,0)$.
    
Let $\mathcal{Q}$ be the projected manifold under the $SO(3)$ action and $\mathcal{Q}^{(in)}$ and $\mathcal{Q}^{(ex)}$ be the interior and exterior regions, respectively. Further let $C_{out} := \{u = -e^{-1}\}$ be an outgoing cone emanating from the center. Then this class of naked singularity spacetimes can be viewed as solutions to~\eqref{eq:whole equation1}--\eqref{eq:whole equation5} with $C^{1,\frac{k^{2}}{1-k^{2}}}$ initial data imposed on $C_{out}$. More precisely, we have \begin{align*}
        &r(-e^{-1},v),\ \Omega^{2}(-e^{-1},v),\ \phi(-e^{-1},v)\in C_{v}^{\infty}\left([-e^{-q_{k}},0)\cup(0,\infty)\right),\\&
        \sup_{v}\vert v\vert^{j-q_{k}}\vert\partial_{v}^{j+1}r\vert(-e^{-1},v)+\sup_{v}\vert v\vert^{j-q_{k}}\vert\partial_{v}^{j}\Omega^{2}\vert(-e^{-1},v)+\sup_{v}\vert v\vert^{j-q_{k}}\vert\partial_{v}^{j}\phi\vert(-e^{-1},v)<\infty,\quad \forall j\geq 2.
    \end{align*}
    For each spacetime in this class of naked singularity solutions to~\eqref{eq:whole equation1}--\eqref{eq:whole equation5}, we have the following quantitative bounds:\begin{itemize}
        \item \textup{(Quantitative bounds toward the singularity in the interior region)} For the quantities $(r,\mu,\phi,Q,F)$ in the interior region $\mathcal{Q}^{(in)}$, we have \begin{align*}
            &\vert r\vert\lesssim (-u)(1+z),\quad -\partial_{u}r\approx 1,\quad \partial_{v}r\approx (-u)^{k^{2}},\quad \mu\lesssim k^{2},\\&
            \vert\partial_{u}\phi_{1}\vert\lesssim \frac{k}{(-u)},\quad \vert\partial_{u}\phi_{2}\vert\lesssim\frac{k}{(-u)},\quad \left\vert\partial_{v}\phi_{1}\right\vert\lesssim\frac{1}{(-u)^{q_{k}}},\quad \left\vert\partial_{v}\phi_{2}\right\vert\lesssim\frac{1}{(-u)^{q_{k}}},\\&
            \vert Q\vert\lesssim u^{2},\quad \vert Q\vert(u,0)\approx u^{2},\quad \vert g^{\alpha\beta}g^{\mu\nu}F_{\alpha\mu}F_{\beta\nu}\vert\lesssim 1,\quad  \vert g^{\alpha\beta}g^{\mu\nu}F_{\alpha\mu}F_{\beta\nu}\vert(u,0)\approx 1.
        \end{align*}
        \item \textup{(Asymptotic flatness)} The initial data are asymptotically flat, in the sense that \begin{align*}
            \vert\partial_{v}\phi\vert(-e^{-1},v)\lesssim v^{-2},\quad \lim_{v\rightarrow\infty} m(-e^{-1},v)<\infty,\quad \lim_{v\rightarrow\infty}Q(-e^{-1},v)<\infty,\quad \lim_{v\rightarrow\infty}r(-e^{-1},v) = \infty
.        \end{align*}
\item \textup{(Free of trapped surfaces)} The spacetime is free of trapped surfaces \begin{equation*}
    \partial_{v}r>0,\quad \partial_{u}r<0.
\end{equation*}
\item \textup{(Incomplete future null infinity)} The spacetime has finite affine length of the future null infinity\begin{equation*}
    \lim_{v\rightarrow\infty}\int_{-e^{-1}}^{0}\frac{\Omega^{2}(u,v)}{\Omega^{2}(-e^{-1},0)}du<\infty.
\end{equation*}
    \end{itemize}

\end{theorem}
Adapting the approach in~\cite{rodnianski2023naked}, we can further construct general exterior regions of charged naked singularity spacetimes. To construct these exterior solutions, we impose suitable initial data on the characteristic hypersurface \begin{equation*}
    \{u = -e^{-1},\ v\geq 0\}\cup\{v = 0,\ -e^{-1}\leq u<0\}.
\end{equation*}
Fixing the gauge for the electromagnetic potential $A$ as in~\eqref{eq: gauge choice of A}, the free initial data in our setting consist of the gauge choices for $r$ and the initial data for $\phi$ on both the ingoing and outgoing cones, together with the value of $\Omega^{2}$ at their intersection $(u,v)=(-e^{-1},0)$. Then using the equations~\eqref{eq:spherical symmetric equaion1}--\eqref{eq:spherical symmetric equation last}, one can determine the values of \begin{align*}
    &\partial_{v}r(u,0),\ \partial_{v}\phi(u,0),\ Q(u,0),\ \Omega^{2}(u,0),\ \mu(u,0),\\&
    \partial_{u}r(-e^{-1},v),\ \partial_{u}\phi(-e^{-1},v),\ Q(-e^{-1},v),\ A_{u}(-e^{-1},v),\ \mu(-e^{-1},v).
\end{align*}
Let $\underline{v}$ be a small parameter. We define the following two auxiliary functions $(\partial_{v}\phi_{1})_{S}$ and $(\partial_{v}\phi_{2})_{S}$: \begin{align}
    (\partial_{v}\phi_{1})_{S}(u,v) = &\frac{1}{(-u)}\int_{u}^{-\left(\frac{v}{\underline{v}}\right)^{\frac{1}{1-q^{2}}}}(-\partial_{v}r)(s,0)\partial_{u}\phi_{1}(s,0)-q_{0}\frac{Q(s,0)\Omega^{2}(s,0)}{4(-s)}\phi_{2}(s,0)ds,\\
    (\partial_{v}\phi_{2})_{S}(u,v)=&\frac{1}{(-u)}\int_{u}^{-\left(\frac{v}{\underline{v}}\right)^{\frac{1}{1-q^{2}}}}(-\partial_{v}r)(s,0)\partial_{u}\phi_{2}(s,0)+q_{0}\frac{Q(s,0)\Omega^{2}(s,0)}{4(-s)}\phi_{1}(s,0)ds,
\end{align}
where $0<q\ll1$ is a small parameter and $(u,v)\in \left\{0\leq\frac{v}{(-u)^{1-q^{2}}}\leq\underline{v}\right\}$. Note that $$(\partial_{v}\phi_{1})_{S}(u,(-u)^{1-q^{2}}\underline{v}) = (\partial_{v}\phi_{2})_{S}(u,(-u)^{1-q^{2}}\underline{v}) = 0.$$

We have the following definition of the admissible initial data on the initial characteristic hypersurface.
\begin{definition}
\label{def: admissible initial data}
    A gauge choice of $r$ together with initial data on the characteristic hypersurface\begin{equation*}
    \{u = -e^{-1},\ v\geq 0\}\cup\{v = 0,\ -e^{-1}\leq u<0\}.
\end{equation*}
is called admissible if \begin{itemize}
    \item \textup{(Gauge choice of $r$)} The gauge choice of $r$ is taken to be \begin{equation*}
        r(u,0) = (-u),\quad \partial_{v}r(-e^{-1},v) = const.
    \end{equation*}
    \item \textup{(Initial data for $\partial_{u}\phi$ on the ingoing cone)} The initial data for $\phi = \rho e^{i\theta}$ on the ingoing cone are taken to be \begin{equation*}
        \partial_{u}\rho(u,0) = \frac{k}{(-u)}-\rho_{1}(u),\quad \rho\partial_{u}\theta(u,0) = \frac{\kappa}{(-u)}-\theta_{1}(u),
    \end{equation*}
    where $k$ and $\kappa$ are small parameters, and $\rho_{1}$ and $\theta_{1}$ are $C^{1}$ functions on $[-e^{-1},0)$ with \begin{equation*}
        \vert\rho_{1}(u)\vert\lesssim\frac{q}{\vert u\vert\vert\log u\vert^{2}},\quad \vert \theta_{1}(u)\vert\lesssim\frac{q^{2}}{\vert u\vert\vert\log(-u)\vert},\quad q: =\sqrt{k^{2}+\kappa^{2}}.
    \end{equation*}
    \item \textup{(Initial data for $\partial_{v}\phi$ on the outgoing cone)} The initial data for $\phi = \phi_{1}+i\phi_{2}$ on the outgoing cone are taken to be 
\begin{align*}
\partial_{v}\phi_{i}(-e^{-1},v)
=\epsilon f(v)+
\begin{cases}
(\partial_{v}\phi_{i})_{S}(-e^{-1},e^{-(1-q^{2})}\underline{v}), & 0\leq v\leq\underline{v},\\
\epsilon \chi(v), & \underline{v} \leq v < \Lambda,\\
\epsilon v^{-2}, & v\geq \Lambda,
\end{cases}
\quad i = 1,2,
\end{align*}
where the function $\chi$ is a function smoothly connecting the first piece and the last piece of the initial data, the function $f$ is a function with sufficiently high regularity and the bound $\vert f\vert\lesssim v^{-2}$, the constant $\epsilon$ is a small real number, and $\Lambda$ is a large real number. Moreover, we assume $\lim_{v\rightarrow\infty}\phi(-e^{-1},v) = 0$.
\item \textup{(Initial data for $\phi$ and $\Omega^{2}$ on the intersecting sphere $(-e^{-1},0)$)} Since the Einstein--Maxwell-charged scalar field equations are not translation invariant, one has to fix the value of $\phi$ on the intersecting sphere. The choice of $\phi(-e^{-1},0)$ and $\Omega^{2}(-e^{-1},0)$ will be \begin{equation*}
    \phi(-e^{-1},0) = 0,\quad \Omega^{2}(-e^{-1},0) = \text{any constant}.
\end{equation*}
\end{itemize}

\end{definition}

Now we are ready to state our theorem on the construction of general exterior solutions.
\begin{theorem}
    \label{main theorem2}
    Fix a nonzero scalar field charge $q_{0}\neq0$ and parameters $(k,\kappa)$. Let $q = \sqrt{k^{2}+\kappa^{2}}$. Imposing the admissible initial data as in Definition~\ref{def: admissible initial data}, there exists a global solution to the Einstein--Maxwell-charged scalar field equations without a trapped surface.
\end{theorem}
A corollary of the above theorem is the existence of an exterior solution describing a charged naked singularity spacetime with a stronger Maxwell field than those constructed in Theorem~\ref{main theorem1}.
\begin{corollary}
    If $k\kappa\neq0$ in Theorem~\ref{main theorem2}, then on the past light cone emanating from the singularity $\mathcal{O}$, we have \begin{equation*}
       \vert Q(u,0)\vert\approx u^{2}\vert\log(-u)\vert,\quad \vert g^{\alpha\beta}g^{\mu\nu}F_{\alpha\mu}F_{\beta\nu}\vert^{\frac{1}{2}}\approx \vert \log(-u)\vert.
    \end{equation*}
\end{corollary}
\begin{proof}
    See Example~\ref{example: exact loglog behavior}.
\end{proof}

\section{Interior construction}
\subsection{Construction of the interior solution to the Einstein--Maxwell-uncharged scalar field equations}
\label{sec: construction of uncharged solution}
In this section, we first construct $(k_{1},k_{2})$-self-similar solutions to the Einstein--Maxwell-uncharged scalar field equations, i.e., the equations \eqref{eq: uncharged u Ray equation}-\eqref{eq: uncharged dvm equation}. Let \begin{equation}
\alpha = \frac{\mr{r}}{\hat z\mr{\hat \lambda}},\quad \theta_{1} = \hat z\alpha\mr{\hat \phi}_{1}^{\prime},\quad \theta_{2} = \hat z\alpha\mr{\hat \phi}_{2}^{\prime}.
\end{equation}
Then using the equations \eqref{eq: dzlambda}-\eqref{eq: dzphi2 equation}, we can derive the following equations for $(\alpha,\theta_{1},\theta_{2})$:\begin{align}
    &\frac{d\alpha}{d\hat z} = \frac{1}{\hat z}\left[(\theta_{1}+k_{1})^{2}+\left(\theta_{2}+k_{2}\right)^{2}+(1-k_{1}^{2}-k_{2}^{2})(1-\alpha)\right],\label{eq: alpha equation}\\&
    \frac{d\theta_{1}}{d\hat z} = \frac{1}{\hat z\alpha}\left[\alpha\left((k_{1}^{2}+k_{2}^{2})\theta_{1}-k_{1}\right)+\theta_{1}\left((\theta_{1}+k_{1})^{2}+(\theta_{2}+k_{2})^{2}-(1+k_{1}^{2}+k_{2}^{2})\right)\right],\label{eq: theta 1 equation}\\&
    \frac{d\theta_{2}}{d\hat z} = \frac{1}{\hat z\alpha}\left[\alpha\left((k_{1}^{2}+k_{2}^{2})\theta_{2}-k_{2}\right)+\theta_{2}\left((\theta_{1}+k_{1})^{2}+(\theta_{2}+k_{2})^{2}-(1+k_{1}^{2}+k_{2}^{2})\right)\right].\label{eq: theta 2 equation}
\end{align}
By Proposition \ref{prop: relation between phi1 and phi2}, we have \begin{equation}
    \frac{\theta_{1}}{k_{1}} = \frac{\theta_{2}}{k_{2}}.\label{eq: relation between theta1 and theta2}
\end{equation}
Let $s = \log \mr{r}(\hat z)$. Then we have \begin{equation}
\label{eq: relation between ds and dz}
    \frac{d}{d\hat z} = \frac{ds}{d\hat z}\frac{d}{ds} = \frac{1}{\mr{r}}\frac{d\mr{r}}{d\hat z}\frac{d}{ds} = \frac{\mr{\hat \lambda}}{\mr{r}}\frac{d}{ds} = \frac{1}{\hat z\alpha}\frac{d}{ds}.
\end{equation}
Hence, under the relations \eqref{eq: relation between theta1 and theta2} and \eqref{eq: relation between ds and dz}, we can reduce the equations \eqref{eq: alpha equation}-\eqref{eq: theta 2 equation} to a $2\times 2$ autonomous system of ODEs \begin{align}
    \frac{d\alpha}{ds}& = \alpha\left[\left(1+\left(\frac{k_{2}}{k_{1}}\right)^{2}\right)(\theta_{1}+k_{1})^{2}+(1-k_{1}^{2}-k_{2}^{2})(1-\alpha)\right],\label{eq: alpha equation for the autonomous system}\\
    \frac{d\theta_{1}}{ds}& =\alpha\left((k_{1}^{2}+k_{2}^{2})\theta_{1}-k_{1}\right)+\theta_{1}\left(\left(1+\left(\frac{k_{2}}{k_{1}}\right)^{2}\right)(\theta_{1}+k_{1})^{2}-(1+k_{1}^{2}+k_{2}^{2})\right),\label{eq: theta equation for the autonomous system}
\end{align}
with the boundary conditions: \begin{equation}
    \lim_{s\rightarrow-\infty}e^{-s}\alpha = -2,\quad \lim_{s\rightarrow -\infty}e^{-s}\theta_{1} = k_{1}.\label{eq: initial data}
\end{equation}
Without loss of generality, we further assume that $k_{1}>0$ and $k_{2}>0$. The following proposition is a direct consequence of this assumption.

\begin{proposition}
\label{prop: positivity of theta}
   Using the $\hat z$-variable, in the region $[-1,\hat{z})$ where $\mr{\hat{\lambda}}>0$, the quantities $\mr{\hat\phi}_{1}^{\prime}$ and $\mr{\hat\phi}_{2}^{\prime}$ are always positive. Consequently, the quantities $\theta_{1}$ and $\theta_{2}$ are positive in this region.
\end{proposition}
\begin{proof}
    We can rewrite the equation \eqref{eq: dzphi1 equation} as \begin{equation*}
        \frac{d}{d\hat z}\left(\mr{r}^{2}\mr{\hat \phi}_{1}^{\prime}\right) = \frac{k_{1}\mr{r}\mr{\hat\lambda}}{(-\hat z)}>0.
    \end{equation*}
    Therefore, by the boundary condition \eqref{eq: initial data}, we have $\mr{r}^{2}\mr{\hat\phi}_{1}^{\prime}>0$. Hence, we can conclude that $\mr{\hat\phi}_{1}^{\prime}$ and $\theta_{1}$ are positive. Similarly, we can show that $\mr{\hat\phi}_{2}^{\prime}$ and $\theta_{2}$ are positive. This concludes the proof.
\end{proof}

Note that the above change of variables is only valid for values of $s$ where $\mr{\hat\lambda} = \frac{d\mr{r}}{d\hat z}>0$ in the region $(-\infty,s)$. However, the following proposition shows that for $k_{1}^{2}+k_{2}^{2}<1$, the change of variable is always valid in the domain of existence of \eqref{eq: alpha equation for the autonomous system}-\eqref{eq: theta equation for the autonomous system}.
\begin{proposition}
\label{prop: monotonicity of z}
    Let $k_{1}^{2}+k_{2}^{2}<1$. Then for any $s$ in the domain of existence of \eqref{eq: alpha equation for the autonomous system}-\eqref{eq: theta equation for the autonomous system}, we have that $\alpha<0$ and $\alpha$ is decreasing. Moreover, the corresponding function $\hat z = \hat z(s)$ is an increasing non-positive function of $s$.
\end{proposition}
\begin{proof}
    Assuming that $\alpha$ is not always negative on the domain of existence, there exists $s_{1}$ such that $\alpha<0$ on $(-\infty,s_{1})$ and $\alpha(s_{1}) = 0$. Then by the boundary condition of $\alpha$ and the mean value theorem, there exists $s_{2}\in(-\infty,s_{1})$ such that $\frac{d\alpha}{ds}(s_{2})>0$. By the definition of $s_{2}$, we have that \begin{equation*}
        \alpha\left[\left(1+\left(\frac{k_{2}}{k_{1}}\right)^{2}\right)(\theta_{1}+k_{1})^{2}+(1-k_{1}^{2}-k_{2}^{2})(1-\alpha)\right](s_{2})<0,
    \end{equation*}
    which contradicts $\frac{d\alpha}{ds}(s_{2})>0$. The monotonicity of $\alpha$ follows from the equation \eqref{eq: alpha equation for the autonomous system}. 
    
    In view of the relation \eqref{eq: relation between ds and dz}, we have \begin{equation*}
    d\log(-\hat z) = \alpha(s)ds.
\end{equation*}
Solving the above equation from the center, we have \begin{equation}
    \hat z(s) = -e^{\int_{-\infty}^{s}\alpha(\widetilde{s})d\widetilde{s}}.\label{eq: z in terms of s}
\end{equation}
Hence, in the domain of existence, the change of variable is always valid and $\hat z(s)\leq 0$.
\end{proof}
Next, we analyze the maximal domain of existence of $(\alpha,\theta_{1})$ with the initial data \eqref{eq: initial data}. Let $(-\infty,s_{*})$ be the maximal interval of existence for equations \eqref{eq: alpha equation for the autonomous system}-\eqref{eq: theta equation for the autonomous system}. By the standard theory of ODEs, either $s_{*} = \infty$ or $s_{*}<\infty$ and at least one of the quantities $\alpha,\theta_{1}$ blows up when $s\rightarrow s_{*}$. The next few propositions aim at analyzing the detailed blowup behaviors of $\alpha$ and $\theta_{1}$ when $s_{*}<\infty$.

\begin{proposition}
\label{prop: z = 0 for alpha blowup}
    For $k_{1}^{2}+k_{2}^{2}<1$, if $s_{*}<\infty$ and $\alpha\rightarrow-\infty$ when $s\rightarrow s_{*}$, then we have $\hat z_{*}:= \lim_{s\rightarrow s_{*}}\hat z(s) = 0$.
\end{proposition}
\begin{proof}
    Let $\beta = \frac{1}{(-\alpha)}$. By Proposition \ref{prop: monotonicity of z} and the equation~\eqref{eq: alpha equation for the autonomous system}, we have that $\beta>0$ decreases to $0$ when $s\rightarrow s_{*}$. Then by the equation \eqref{eq: alpha equation for the autonomous system}, we have \begin{equation*}
        \frac{d\beta}{ds}>-(1-k_{1}^{2}-k_{2}^{2})(1+\beta).
    \end{equation*}
By the monotonicity of $\beta$, there exists $s_{0}$ such that on $(s_{0},s_{*})$, we have $0<\beta<1$. Hence, for any $s\in(s_{0},s_{*})$, we have \begin{equation*}
    \frac{d\beta}{ds}>-2(1-k_{1}^{2}-k_{2}^{2}).
\end{equation*}
Integrating from $s = s_{*}$, we have \begin{equation*}
    0<\beta(s)<2(1-k_{1}^{2}-k_{2}^{2})(s_{*}-s),\quad s_{0}<s<s_{*}.
\end{equation*}
Hence, on $(s_{0},s_{*})$, we have \begin{equation*}
    \alpha<-\frac{1}{2(1-k_{1}^{2}-k_{2}^{2})}\frac{1}{s_{*}-s}.
\end{equation*}
Using \eqref{eq: z in terms of s}, we have that \begin{equation*}
    \hat z(s_{0})e^{\int_{s_{0}}^{s}\alpha(\bar{s})d\bar{s}}\leq\frac{\hat z(s_{0})}{(s_{*}-s_{0})^{\frac{1}{2(1-k_{1}^{2}-k_{2}^{2})}}}(s_{*}-s)^{\frac{1}{2(1-k_{1}^{2}-k_{2}^{2})}}\leq \hat z(s)\leq 0.
\end{equation*}
Therefore, we can conclude that $\lim_{s\rightarrow s_{*}}\hat z(s) = 0$.
\end{proof}

The next proposition shows that when $s_{*}<\infty$, only the quantity $\alpha$ blows up, whereas the quantity $\theta_{1}$ stays bounded.\begin{proposition}
\label{prop: boundedness of theta}
    For $k_{1}^{2}+k_{2}^{2}<1$, either $s_{*} = \infty$ or $s_{*}$ is finite and $\alpha\rightarrow -\infty$ while $\theta_{1}$ remains uniformly bounded as $s\rightarrow s_{*}$. In particular, this means $\hat z_{*} =  0$.
\end{proposition}
\begin{proof}
By  Proposition \ref{prop: z = 0 for alpha blowup}, if one can show that $\alpha$ goes to negative infinity as $s\rightarrow s_{*}$, then $\hat{z}_{*} = 0$.

Assume that $\hat z_{*}<0$ and $\alpha$ is bounded. Using $\alpha = \frac{\mr{r}}{\hat z\mr{\hat\lambda}}$, we have \begin{equation*}
    \frac{d}{d\hat z}\log(\mr{r}) = \frac{1}{\hat z\alpha}(\hat z),
\end{equation*}
from which we can conclude that $\mr{\hat\lambda}>0$ has an upper bound and a lower bound away from $0$ on $[-1,z_{*})$. We can rewrite the equation \eqref{eq: dzphi1 equation} as \begin{equation*}
    \frac{d}{d\hat z}\left(\mr{r}^{2}\mr{\hat \phi}_{1}^{\prime}\right) = \frac{k_{1}\mr{r}\mr{\hat \lambda}}{(-\hat{z})}.
\end{equation*}
Hence, we can conclude that $\mr{\hat\phi}_{1}^{\prime}$ is also uniformly bounded. Therefore, $\theta_{1}$ is bounded on $(-\infty,s_{*})$, which contradicts the fact that $s_{*}$ is the maximal domain of existence. 

Next, we assume that $\hat z_{*} = 0$ and $\alpha$ is bounded: \begin{equation*}
    -A\leq \alpha<0, \quad \text{for }s<s_{*}.
\end{equation*}
Then by the boundary condition \eqref{eq: initial data}, there exists $s_{1}$ such that for any $s\in(-\infty,s_{1})$, we have \begin{equation*}
    \alpha\geq -e^{s}.
\end{equation*}
Therefore, we have the following estimate on $z$: \begin{equation*}
    \lim_{s\rightarrow s_{*}}\hat z(s) = -e^{\int_{-\infty}^{s_{1}}\alpha(\bar{s})d\bar{s}}e^{\int_{s_{1}}^{s_{*}}\alpha(\bar{s})d\bar{s}}\leq -e^{\int_{-\infty}^{s_{1}}\alpha(\bar{s})d\bar{s}}e^{-A(s_{*}-s_{1})}<0,
\end{equation*}
which contradicts the assumption $\hat z_{*} = 0$. 

Now, it remains to show that $\theta_{1}$ is uniformly bounded. Using the equations \eqref{eq: dzlambda} and \eqref{eq: dzphi1 equation}, we have that \begin{equation}
    \hat z\frac{d}{d\hat z}\left(\frac{\mr{r}}{\mr{\hat\lambda}}\mr{\hat\phi}_{1}^{\prime}\right) = -k_{1}+\frac{\mr{\hat\phi}_{1}^{\prime}}{\mr{\hat\lambda}}\left(\frac{\frac{1}{4}\mr{\hat\Omega}^{2}+\mr{\hat\lambda}\mr{\hat\nu}}{\mr{\hat\lambda}}+(-\hat z)\mr{\hat\lambda}\right).\label{eq: equation for theta1 in proving the boundedness}
\end{equation}
Assuming that $\theta_{1} = \frac{\mr{r}}{\mr{\hat\lambda}}\mr{\hat\phi}_{1}^{\prime}$ is unbounded, there exists a sequence $\bar{s}_{n}\rightarrow s_{*}$ such that $\theta_{1}(\bar{s}_{n})\rightarrow +\infty$. In particular, we can choose $\bar{s}_{n}$ such that $\frac{d\theta_{1}}{ds}(\bar{s}_{n})>0$. Then by \eqref{eq: equation for theta1 in proving the boundedness}, we have \begin{equation*}
    0<\theta_{1}<\frac{k_{1}\mr{r}}{\frac{\frac{1}{4}\mr{\hat\Omega}^{2}+\mr{\hat\lambda}\mr{\hat\nu}}{\mr{\hat\lambda}}+(-\hat z)\mr{\hat\lambda}}.
\end{equation*}
Recall that the Hawking mass is defined to be \begin{equation*}
    1-\mr{\mu} = -\frac{4\mr{\hat\lambda}\mr{\hat\nu}}{\mr{\hat\Omega}^{2}}.
\end{equation*}
Hence, we have \begin{equation*}
   0<\theta_{1}(\bar{s}_{n})<\frac{k_{1}\mr{r}}{\left(1+\frac{\mr{\mu}}{1-\mr{\mu}}\right)(-\hat z)\mr{\hat\lambda}+\frac{\mr{\mu}}{1-\mr{\mu}}\mr{r}}<\frac{k_{1}}{\mr{\mu}}(\bar{s}_{n}) = \frac{k_{1}\mr{r}}{2\mr{m}}(\bar{s}_{n}).
\end{equation*}
It remains to establish a lower bound for $\mr{\mu}$. Using the equation for $\mr{m}$, we have \begin{equation*}
    \frac{d\mr{m}}{d\hat z}=\frac{1}{2}\mr{\hat\lambda}(\theta_{1}^{2}+\theta_{2}^{2})\left(1-\mr{\mu}\right)>0,
\end{equation*}
where the last inequality is due to the fact that in the $s$-region $(-\infty,s_{*})$, the corresponding spacetime is free of trapped surfaces. Therefore, we have that $\mr{m}>0$ is bounded away from $0$. Recall that $\beta = \frac{1}{(-\alpha)}$. Using \eqref{eq: alpha equation for the autonomous system}, we have \begin{equation*}
    \hat z\frac{d\beta}{d\hat z}>(1-k_{1}^{2}-k_{2}^{2})\left(\beta^{2}+\beta\right)>(1-k_{1}^{2}-k_{2}^{2})\beta.
\end{equation*}
Then we have $\alpha(z)\lesssim-(-\hat z)^{-(1-k_{1}^{2}-k_{2}^{2})}$ for $\hat{z}$ sufficiently close to $0$. Using \begin{equation*}
    \frac{d\log\mr{r}}{d\hat z} = \frac{1}{\hat z\alpha}
\end{equation*}
again, we can conclude that $\mr{r}$ is uniformly bounded on $(-\infty,s_{*})$. Hence, $\theta_{1}$ is uniformly bounded. This concludes the proof.
\end{proof}
Now, we analyze for which values of $k_{1}$ and $k_{2}$ the corresponding maximal interval of existence $s_{*}$ is finite. We have the following proposition.
\begin{proposition}
    The maximal domain of existence $s_{*}$ for the equations \eqref{eq: alpha equation for the autonomous system}-\eqref{eq: theta equation for the autonomous system} is finite $s_{*}<\infty$ if $k_{1}^{2}+k_{2}^{2}<1$.
\end{proposition}
\begin{proof}
Recall that $\beta = \frac{1}{(-\alpha)}$. For $k_{1}^{2}+k_{2}^{2}<1$, by Proposition \ref{prop: monotonicity of z}, we have that $\beta>0$ and is decreasing on $(-\infty,s_{*})$. We can rewrite the equation \eqref{eq: alpha equation for the autonomous system} in terms of $\beta$ as \begin{equation}
    \frac{d\beta}{ds} = -\beta\left[\left(1+\left(\frac{k_{2}}{k_{1}}\right)^{2}\right)(\theta_{1}^{2}+2k_{1}\theta_{1})+1\right]-(1-k_{1}^{2}-k_{2}^{2})<-(1-k_{1}^{2}-k_{2}^{2}).\label{eq: beta equation in determining the maximal domain of existence}
\end{equation}
Assuming that $s_{*} = \infty$ and integrating \eqref{eq: beta equation in determining the maximal domain of existence}, we have \begin{equation*}
    \beta(s)<\beta(-100)-(1-k_{1}^{2}-k_{2}^{2})(s+100),
\end{equation*} 
which contradicts the fact that $\beta$ is positive on $(-\infty,s_{*})$ by Proposition \ref{prop: monotonicity of z} when $s$ is sufficiently large. This concludes the proof.
\end{proof}
The next proposition studies the limit of $\theta_{1}$ when $s\rightarrow s_{*}$ for $0<k_{1}^{2}+k_{2}^{2}<1$.
\begin{proposition}
    For $0<k_{1}^{2}+k_{2}^{2}<1$, we have that \begin{equation}
        \theta_{1}\rightarrow \frac{k_{1}}{k_{1}^{2}+k_{2}^{2}},\quad \theta_{2}\rightarrow \frac{k_{2}}{k_{1}^{2}+k_{2}^{2}},\quad s\rightarrow s_{*}.\label{eq: the limit of theta}
    \end{equation}
\end{proposition}
\begin{proof}
    The second limit in \eqref{eq: the limit of theta} follows from the first limit and the relation $\frac{\theta_{1}}{k_{1}} = \frac{\theta_{2}}{k_{2}}$. Therefore, it suffices to prove the first limit in \eqref{eq: the limit of theta}. Let $\widetilde{\theta}_{1} = \theta_{1}-\frac{k_{1}}{k_{1}^{2}+k_{2}^{2}}$. Then by \eqref{eq: theta equation for the autonomous system}, we can write down the equation for $\widetilde{\theta}_{1}$: \begin{equation*}
        \frac{d\widetilde{\theta}_{1}}{ds} = (k_{1}^{2}+k_{2}^{2})\widetilde{\theta}_{1}\alpha+\underbrace{\left(\widetilde{\theta}_{1}+\frac{k_{1}}{k_{1}^{2}+k_{2}^{2}}\right)\left[\left(1+\left(\frac{k_{2}}{k_{1}}\right)^{2}\right)\left(\widetilde{\theta}_{1}+\frac{k_{1}}{k_{1}^{2}+k_{2}^{2}}+k_{1}\right)\right]}_{\eta},
    \end{equation*}
    where $\eta$ is uniformly bounded by Proposition \ref{prop: boundedness of theta}. Using the method of integrating factors, we have \begin{equation*}
        \widetilde{\theta}_{1}(s) = e^{-\int_{s_{*}-\epsilon}^{s}(k_{1}^{2}+k_{2}^{2})\frac{1}{\beta}d\bar{s}}\widetilde{\theta}_{1}(s_{*}-\epsilon)+\int_{s_{*}-\epsilon}^{s}e^{\int_{s}^{\bar{s}}(k_{1}^{2}+k_{2}^{2})\frac{1}{\beta}d\widetilde{s}}\eta d\bar{s},\quad \epsilon>0.
    \end{equation*}
    By \eqref{eq: beta equation in determining the maximal domain of existence}, for $\epsilon$ sufficiently small, we have that $\beta\rightarrow (1-k_{1}^{2}-k_{2}^{2})(s_{*}-s)$. Hence, we have $\widetilde{\theta}_{1}\rightarrow 0$. This concludes the proof of this proposition.
\end{proof}
Finally, we conclude the limiting behavior of $\mr{r}$, $\mr{\hat\lambda}$, $\mr{\hat\phi}_{1}^{\prime}$, and $\mr{\hat\phi}_{2}^{\prime}$ when $s\rightarrow s_{*}$ for $0<k_{1}^{2}+k_{2}^{2}<1$.
\begin{proposition}
\label{prop: blowup rate of geometric quantities}
    For $0<k_{1}^{2}+k_{2}^{2}<1$, when $s\rightarrow s_{*}$, the following bounds hold \begin{equation}
        \mr{r}\approx 1,\quad \mr{\hat\lambda}\approx (-\hat z)^{-k_{1}^{2}-k_{2}^{2}},\quad \mr{\hat\phi}_{1}^{\prime}\approx\mr{\hat\phi}_{2}^{\prime}\approx(-\hat z)^{-k_{1}^{2}-k_{2}^{2}}.
    \end{equation}
\end{proposition}
\begin{proof}
    Since $s_{*}<\infty$, by the relation $s = \log\mr{r}$, we have $0<\lim_{s\rightarrow s_{*}}\mr{r}<\infty$. By \eqref{eq: beta equation in determining the maximal domain of existence}, we have that \begin{equation*}
        \alpha\rightarrow -(1-k_{1}^{2}-k_{2}^{2})(s_{*}-s)^{-1},\quad s\rightarrow s_{*}.
    \end{equation*}
    Then, by the equation $\frac{d\hat z}{ds} = \hat z\alpha$, we have \begin{equation*}
        (-\hat z)\approx (s_{*}-s)^{1-k_{1}^{2}-k_{2}^{2}}, \quad s\rightarrow s_{*}.
    \end{equation*}
    Using $\mr{\lambda} = \frac{\mr{r}}{\hat z\alpha}$, we have \begin{equation*}
        \mr{\hat\lambda}\approx (-\hat z)^{-k_{1}^{2}-k_{2}^{2}},\quad s\rightarrow s_{*}.
    \end{equation*}
    Using the boundedness of $\theta_{1}$ and $\theta_{2}$, we have \begin{equation*}
        \mr{\hat\phi}_{1}^{\prime}\approx \mr{\hat\phi}_{2}^{\prime}\approx (-\hat z)^{-k_{1}^{2}-k_{2}^{2}},\quad s\rightarrow s_{*}.
    \end{equation*}
    This concludes the proof of this proposition.
\end{proof}
\subsection{A perturbative construction of the interior solution for the Einstein--Maxwell-charged scalar field equations}
\label{sec: interior construction}
In this section, we construct the corresponding interior regions of naked singularity spacetimes for the Einstein--Maxwell-charged scalar field equations by perturbing the interior solution we have constructed in Section~\ref{sec: construction of uncharged solution}.
\subsubsection{Regular coordinates}
In view of Proposition \ref{prop: blowup rate of geometric quantities}, for $0<k_{1}^{2}+k_{2}^{2}<1$, the geometric quantity $\mr{\hat\lambda} = \partial_{\hat v}r$ and the derivatives of the scalar field  $\mr{\hat\phi}_{1}^{\prime}$ and $\mr{\hat\phi}_{2}^{\prime}$ blow up when approaching the light cone $\{\hat{z} = 0\}$. However, this blowup turns out to be a coordinate singularity. Let $u = \hat{u}$ and $(-v) = (-\hat{v})^{1-k_{1}^{2}-k_{2}^{2}}$. Then the double-null coordinates become \begin{equation*}
    g = -\frac{1}{2}\hat{\Omega}^{2}d\hat u\otimes d\hat v-\frac{1}{2}\hat\Omega^{2}d\hat v\otimes d\hat u = -\frac{1}{2}\Omega^{2}du\otimes dv-\frac{1}{2}\Omega^{2}dv\otimes du,
\end{equation*}
where $\Omega^{2}: = \frac{1}{1-k^{2}}(-\hat{v})^{k^{2}}\hat\Omega^{2}$ and $k^{2}: = k_{1}^{2}+k_{2}^{2}$. Under the change of variables $(\hat u,\hat v)\rightarrow (u,v)$, we have that \begin{align*}
   & r(u,v) = (-u)\mr{r}(z),\quad \lambda: = \partial_{v}r = \frac{(-\hat{v})^{k^{2}}}{1-k^{2}}\partial_{\hat{v}}r =: (-u)^{k^{2}}\mr{\lambda}(z),\quad \nu: =\partial_{u}r = \partial_{\hat{u}}r = \mr{\nu}(z),\\&\Omega^{2}(u,v)=:(-u)^{k^{2}}\mr{\Omega}^{2}(z),\quad \phi_{1}(u,v) = \mr{\phi}_{1}(z)-k_{1}\log(-u),\quad \phi_{2}(u,v) = \mr{\phi}_{2}(z)-k_{2}\log(-u),\\&
  \mr{\phi}_{1}^{\prime}: = \frac{d}{dz}\mr{\phi}_{1},\quad \mr{\phi}_{2}^{\prime}: = \frac{d}{dz}\mr{\phi}_{2},
\end{align*}
where $z: = \frac{v}{(-u)^{1-k^{2}}}$ is the renormalized self-similar coordinate. The following proposition is an immediate consequence of this change of variables.
\begin{proposition}
    If $0<k_{1}^{2}+k_{2}^{2}<1$, then the solution to the Einstein--Maxwell-uncharged scalar field equations can be regularly extended to the region $-1\leq z\leq0$, where $\{z = -1\}$ corresponds to the center of the spacetime, and $\{z = 0\}$ corresponds to the ingoing light cone emanating from the singularity. In particular, $(r,\lambda,\nu,\mr{\phi}_{1}^{\prime},\mr{\phi}_{2}^{\prime})$ is uniformly bounded.
\end{proposition}
We denote the region $\{-1\leq u<0,\ -(-u)^{q_{k}}\leq v\leq0\}$ as $\mathcal{Q}^{(in)}$. To distinguish the exact $(k_{1},k_{2})$-self-similar spacetime we constructed in Section~\ref{sec: construction of uncharged solution} from the spacetime constructed in this section, we denote by $(\widetilde{g}^{k_{1},k_{2}},\widetilde{\phi}^{k_{1},k_{2}})$ the $(k_{1},k_{2})$-self-similar solution to the Einstein--Maxwell-uncharged scalar field equations, where \begin{align*}
    &\widetilde{g}^{k_{1},k_{2}} = -\frac{1}{2}\widetilde{\Omega}_{k_{1},k_{2}}^{2}du\otimes dv-\frac{1}{2}\widetilde{\Omega}_{k_{1},k_{2}}^{2}dv\otimes du+\widetilde{r}_{k_{1},k_{2}}^{2}d\sigma^{2},\\&
    \widetilde{\phi}^{k_{1},k_{2}} =\widetilde{\phi}_{1}^{k_{1},k_{2}}+i\widetilde\phi_{2}^{k_{1},k_{2}},\\&
    \widetilde{\phi}_{1}^{k_{1},k_{2}} = \mr{\widetilde{\phi}}_{1}^{k_{1},k_{2}}(z)-k_{1}\log(-u),\quad \widetilde{\phi}_{2}^{k_{1},k_{2}} = \mr{\widetilde{\phi}}_{2}^{k_{1},k_{2}}(z)-k_{2}\log(-u)
\end{align*}

We further fix the gauge freedom of $(\widetilde{g}^{k_{1},k_{2}},\widetilde{\phi}^{k_{1},k_{2}})$ by letting \begin{equation*}
    \widetilde{\phi}^{k_{1},k_{2}}(u,0) = \widetilde{\phi}^{k_{1},k_{2}}_{1}(u,0)+i\widetilde{\phi}^{k_{1},k_{2}}_{2}(u,0) = -k_{1}\log(-u)-ik_{2}\log(-u).
\end{equation*}
For simplicity of notation, we often suppress the index $(k_{1},k_{2})$ and denote the $(k_{1},k_{2})$-self-similar solution by $(\widetilde{g},\widetilde{\phi},\widetilde{Q}\equiv 0)$.

We collect the bounds for the $(k_{1},k_{2})$-self-similar spacetimes $(\widetilde{g}^{k_{1},k_{2}},\widetilde{\phi}^{k_{1},k_{2}})$ in the following proposition for $0<k_{1}^{2}+k_{2}^{2}<1$.
\begin{proposition}
\label{prop: estimate on the background spacetime}
    For $0<k_{1}^{2}+k_{2}^{2}<1$, the $(k_{1},k_{2})$-self-similar solution $(\widetilde g,\widetilde{\phi})$ to the Einstein-Maxwell-uncharged scalar field equations constructed in Section \ref{sec: construction of uncharged solution} satisfies the bounds \begin{align}
        &\widetilde r\approx (-u)(z+1),\quad (-\widetilde\nu)\approx 1,\quad \widetilde\lambda\approx (-u)^{k^{2}},\quad \widetilde{\Omega}^{2}\approx(-u)^{k^{2}},\quad \frac{\widetilde{\mu}}{\mr{\widetilde r}}\approx k^{2},\\&
        \left\vert \mr{\widetilde\phi}_{1}\right\vert\lesssim k_{1},\quad \left\vert\mr{\widetilde{\phi}}_{2}\right\vert\lesssim k_{2},\quad \left\vert\partial_{u}\mr{\widetilde{\phi}}_{1}\right\vert\lesssim \frac{k_{1}}{(-u)},\quad \left\vert\partial_{u}\mr{\widetilde{\phi}}_{2}\right\vert\lesssim \frac{k_{2}}{(-u)},\\& \left\vert\partial_{v}\mr{\widetilde{\phi}}_{1}\right\vert\lesssim \frac{k_{1}}{\eta}\frac{1}{(-u)^{q_{k}}}\left\vert z\right\vert^{-\eta},\quad \left\vert\partial_{v}\mr{\widetilde{\phi}}_{2}\right\vert\lesssim\frac{k_{2}}{\eta}\frac{1}{(-u)^{q_{k}}}\left\vert z\right\vert^{-\eta}.\label{eq: estimate on dvphitilde}
    \end{align}
    Moreover, on the ingoing cone $\{v = 0\}$, we have \begin{align}
        &\widetilde{\mu}(u,0) = \frac{k^{2}}{1+k^{2}},\quad \widetilde{r}(u,0)= \mr{\widetilde{r}}(0)(-u),\quad \widetilde{\Omega}^{2}(u,0) = \mr{\widetilde{\Omega}}^{2}(0)(-u)^{k^{2}},\\&\widetilde{\lambda}(u,0) = \frac{1}{4(1+k^{2})}\mr{\widetilde{\Omega}}^{2}(0)(-u)^{k^{2}}.
    \end{align}
\end{proposition}
\begin{proof}
The bounds for $\widetilde{r}$, $\widetilde{\nu}$, $\widetilde{\lambda}$, and $\widetilde{\mu}$ follow directly from our construction in Section \ref{sec: construction of uncharged solution} and the coordinate transformation. Recall that for the $(k_{1},k_{2})$-self-similar solution, we have the following equation for $\mr{\widetilde{\phi}}_{1}$ \begin{equation*}
    z\mr{\widetilde{r}}\frac{d\mr{\widetilde{\phi}}_{1}^{\prime}}{dz}+2z\frac{d\mr{\widetilde{r}}}{dz}\mr{\widetilde{\phi}}_{1}^{\prime}-\frac{k_{1}^{2}}{q_{k}}\mr{\widetilde{r}}\mr{\widetilde{\phi}}_{1}^{\prime}+\frac{k_{1}}{q_{k}}\frac{d\mr{\widetilde{r}}}{dz} = 0,
\end{equation*}
which can be rewritten as \begin{equation*}
   \frac{d}{dz}\left((-z)^{-\frac{k^{2}}{q_{k}}}\mr{\widetilde{r}}^{2}\mr{\widetilde{\phi}}_{1}^{\prime}\right) = \frac{k_{1}}{q_{k}}(-z)^{-1-\frac{k^{2}}{q_{k}}}\mr{\widetilde{r}}\frac{d\mr{\widetilde{r}}}{dz}.
\end{equation*}
Hence, we have \begin{equation*}
    \mr{\widetilde{\phi}}_{1}^{\prime} = \frac{k_{1}}{q_{k}}(-z)^{\frac{k^{2}}{q_{k}}}\frac{1}{\mr{\widetilde{r}}^{2}(z)}\int_{-1}^{z}(-s)^{-1-\frac{k^{2}}{q_{k}}}\mr{\widetilde{r}}\frac{d\mr{\widetilde{r}}}{ds}ds.
\end{equation*}
Then for $z\in[-1,-\frac{1}{2}]$, we have \begin{equation*}
    \left\vert \mr{\widetilde{\phi}}_{1}^{\prime}\right\vert\lesssim \frac{1}{\mr{\widetilde r}^{2}(z)}\int_{-1}^{-\frac{1}{2}}\frac{k_{1}}{(-s)}\mr{\widetilde{r}}ds\lesssim k_{1}.
\end{equation*}
For $z\in[-\frac{1}{2},0]$, for any $\eta\gg k$ we have \begin{equation*}
    \left\vert z \right\vert^{\eta}\left\vert\mr{\widetilde{\phi}}_{1}^{\prime}\right\vert(z)\lesssim k_{1}(-z)^{\frac{k^{2}}{q_{k}}}\int_{-1}^{z}(-s)^{-1-\frac{k^{2}}{q_{k}}+\eta} ds\lesssim \frac{k_{1}}{\eta}.
\end{equation*}
Since $\theta_{1}\rightarrow \frac{k_{1}}{k^{2}}$ when $z\rightarrow 0$, we have $\mr{\widetilde \phi}_{1}^{\prime}\approx \frac{1}{k_{1}}$ for $z$ sufficiently close to $0$.

By the explicit form of $\widetilde{\Omega}$, $\widetilde{r}$, $\widetilde{\lambda}$, and $\widetilde{\nu}$, the stated bounds follow from the $(k_{1},k_{2})$-self-similarity. Recall on the ingoing cone $\{v = 0\}$, the transport equation for $\widetilde \mu$ is reduced to \begin{equation*}
    \partial_{u}\widetilde{\mu}+\frac{k^{2}+1}{u}\widetilde{\mu} = \frac{k^{2}}{u}.
\end{equation*}
Hence, we have \begin{equation*}
    \widetilde\mu(u,0) = \frac{k^{2}}{1+k^{2}}.
\end{equation*}
\end{proof}

\subsubsection{Setup of the perturbations}
\label{sec: setup of the perturbations}
To construct the interior solutions to the original Einstein--Maxwell-charged scalar field equations with a nonzero coupling constant $q_{0}$, we adopt the perturbative approach. Let $\mathcal{Q}_{\epsilon}^{(in)}: = \mathcal{Q}^{(in)}\cap \{u\geq -\epsilon\}$. The approach here is twofold. First, note that under some admissible conditions (see Proposition \ref{prop: admissible condition for uncharge to charge}), the $(k_{1},k_{2})$-self-similar spacetime $(\widetilde{g},\widetilde\phi)$ is also a solution to the Einstein--Maxwell-charged scalar field equations with zero spacetime charge $Q$. Taking $(\widetilde g,\widetilde{\phi}, \widetilde Q\equiv 0)$ as a solution to the Einstein--Maxwell-charged scalar field equations in the region $\mathcal{Q}_{\epsilon}^{(in)}$ and imposing well-designed initial data on the characteristic initial hypersurface $\mathcal{H}_{\epsilon}: = \{u = -\epsilon,\ -\epsilon^{q_{k}}\leq v\leq 0\}\cup \{v = 0,-1\leq u\leq -\epsilon\}$, we will establish the existence and quantitative estimates for a solution $(g,\phi, Q)$ to the Einstein--Maxwell-charged scalar field equations in the region $\{-1\leq u\leq -\epsilon,\ -(-u)^{q_{k}}\leq v\leq 0\}$. The solutions in these two regions $\mathcal{Q}_{\epsilon}^{(in)}$ and $\mathcal{Q}^{(in)}\backslash\mathcal{Q}^{(in)}_{\epsilon}$ can be glued together to form a solution $(g^{\epsilon},\phi^{\epsilon}, Q^{\epsilon})$ to the Einstein--Maxwell-charged scalar field equations in the whole interior region $\mathcal{Q}^{(in)}$. Second, letting $\epsilon$ go to zero, a limiting argument will yield a naked singularity interior for the Einstein--Maxwell-charged scalar field equations. The following picture (Figure~\ref{fig:perturbative construction}) depicts the construction. 
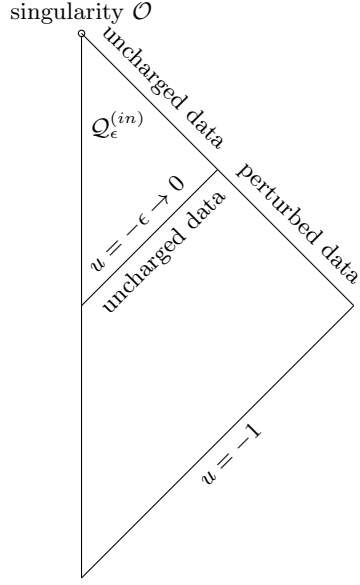
\begin{figure}[htbp]
    \centering
    \begin{tikzpicture}[ scale=1.2, >=Latex, every node/.style={font=\small} ]
\coordinate (O) at (0,0);
\coordinate (A) at (3,-3);
\coordinate (B) at (0,-6);
\draw (A)--(B) node[midway, below, sloped] {$u = -1$};
\draw (O)--(B);
 \filldraw[
        draw=black,
        fill=white
    ] (O) circle (1pt);
\coordinate (C) at (1.5,-1.5);
\coordinate (D) at (0,-3);
\draw (O)--(C) node [midway, above, sloped]{uncharged data};
\draw (C)--(A) node [midway, above, sloped] {perturbed data};
\draw (C)--(D) node[midway, above,sloped] {$u = -\epsilon\rightarrow 0$}
node [midway, below, sloped] {uncharged data};
\node at (0.4,-1) {$\mathcal{Q}_{\epsilon}^{(in)}$};
\node[above] at (O) {singularity $\mathcal{O}$};

\end{tikzpicture}
    \caption{Perturbative construction for the interior region}
    \label{fig:perturbative construction}
\end{figure}

One distinction between solutions to the Einstein–Maxwell system with an uncharged scalar field and those with a charged scalar field is that the former enjoy a translation symmetry: shifting the scalar field by a constant still produces a solution to the same equations. Therefore, in the construction of the $(k_{1},k_{2})$-self-similar solutions to the Einstein--Maxwell-uncharged scalar field equations in Section \ref{sec: construction of uncharged solution}, only the derivatives of the scalar field enter into the construction, and we did not fix the gauge freedom associated with this translation invariance. However, to make the $(k_{1},k_{2})$-self-similar solution $(\widetilde{g},\widetilde{\phi},\widetilde{Q}\equiv 0)$ a solution to the Einstein--Maxwell-charged scalar field equations with zero spacetime charge $Q$, we have to fix this gauge freedom in the following proposition.
\begin{proposition}
\label{prop: admissible condition for uncharge to charge}
    The $(k_{1},k_{2})$-self-similar solution with $0<k_{1}^{2}+k_{2}^{2}<1$ constructed in Section \ref{sec: construction of uncharged solution} is also a solution to the Einstein--Maxwell-charged scalar field equations if \begin{equation}
        \frac{\mr{\widetilde{\phi}}_{1}(0)}{k_{1}} = \frac{\mr{\widetilde{\phi}_{2}}(0)}{k_{2}}.\label{eq: admissible condition}
    \end{equation}
\end{proposition}
\begin{proof}
    By the local well-posedness of the Einstein--Maxwell-charged scalar field equations, we only need to prove that on the outgoing cone $\{u = -1,\ -1\leq v\leq 0\}$, the solution $$(\widetilde{g},\widetilde{\phi}, \widetilde Q\equiv 0)|_{u = -1}$$ satisfies the constraint equations \eqref{eq:dv-Ray equation} and \eqref{eq:v-transport equation for Q}. Since $\mr{\widetilde{\phi}}_{1}(0) = \mr{\widetilde{\phi}}_{2}(0) = 0$ and $$\frac{1}{k_{1}}\frac{d}{dz}\mr{\widetilde{\phi}}_{1} = \frac{1}{k_{2}}\frac{d}{dz}\mr{\widetilde{\phi}}_{2},$$we have that \begin{equation*}
        \frac{\widetilde{\phi}_{1}}{k_{1}}(-1,v) = \frac{\widetilde{\phi}_{2}}{k_{2}}(-1,v).
    \end{equation*}
    Then we have \begin{equation*}
        \partial_{v}Q(-1,v) = q_{0}^{2}r^{2}\Im\left(\widetilde{\phi}\overline{\partial_{v}\widetilde{\phi}}\right) = q_{0}^{2}r^{2}\left(\widetilde{\phi}_{2}\partial_{v}\widetilde{\phi}_{1}-\widetilde{\phi}_{1}\partial_{v}\widetilde{\phi}_{2}\right) = 0.
    \end{equation*}
    This concludes the proof.
\end{proof}
Therefore, we take the $(k_{1},k_{2})$-self-similar solution $(\widetilde{g},\widetilde{\phi},\widetilde Q\equiv 0)$ satisfying the admissibility condition \eqref{eq: admissible condition} as our solution to the Einstein--Maxwell-charged scalar field equations in the region $\mathcal{Q}_{\epsilon}^{(in)}$. We then construct a solution to the Einstein--Maxwell-charged scalar field equations in the complementary region $\mathcal{Q}^{(in)}\backslash \mathcal{Q}_{\epsilon}^{(in)}$, which can be glued to the above solution to yield a solution in the entire interior region $\mathcal{Q}^{(in)}$. 

Now, we impose the initial data on the initial characteristic hypersurface $\mathcal{H}_{\epsilon}$ and consider the propagation of this given data in the region $\mathcal{Q}^{(in)}\backslash \mathcal{Q}_{\epsilon}^{(in)}$. Let $f_{1}$ and $f_{2}$ be functions on $[-1,0]$ with the following estimates
\begin{equation}
    \left\vert\partial_{u}^{i}(f_{1}-1)\right\vert\leq (-u)^{\gamma-i},\quad \left\vert\partial_{u}^{i}f_{2}\right\vert\leq (-u)^{\gamma-i},\qquad i = 0,1,\label{eq: condition on f}
\end{equation}
for some constant $\gamma>0$.

The perturbations of $\widetilde\phi_{1}$ and $\widetilde\phi_{2}$ on the initial ingoing cone $\{v = 0,\ -1\leq u\leq-\epsilon\}$ should be thought of as $f_{1}\chi_{\epsilon}$ and $f_{2}\chi_{\epsilon}$ respectively, where $\chi_{\epsilon}$ is a suitable smooth cut-off function with support on $[-1,-\epsilon]$. A natural choice of the cut-off function $\chi_{\epsilon}$ can be $\chi_{\epsilon}\equiv 0$ on $[-2\epsilon,0]$, $\chi_{\epsilon}\equiv 1$ on $[-1,-3\epsilon]$, and $\chi_{\epsilon}$ is decreasing on $[-3\epsilon,-2\epsilon]$.  Let $\chi(u)$ be a fixed smooth cut-off function with \begin{equation*}
\chi(u)=
\begin{cases}
0, & u\geq 0,\\[4pt]
\chi^{\prime}(u)\leq 0, & u\in[-1,0],\\[6pt]
1, & u\leq -1.
\end{cases}
\end{equation*}
Let \begin{equation}
    \chi_{\epsilon} =\chi\left(\frac{u+2\epsilon}{2\epsilon}\right).
\end{equation}
Then we have the following lemma, whose proof is a straightforward computation.
\begin{lemma}
    The cut-off function $\chi_{\epsilon}$ is a decreasing function on $[-1,0]$ with \begin{equation*}
\chi_{\epsilon}(u)=
\begin{cases}
0, & u\in[-2\epsilon,0],\\[4pt]
\chi_{\epsilon}^{\prime}(u)\leq 0, & u\in[-3\epsilon,-2\epsilon],\\[6pt]
1, & u\in[-1,-3\epsilon].
\end{cases}
\end{equation*}
We have the following estimate for $\partial_{u}\chi_{\epsilon}$ \begin{align}
    &\partial_{u}\chi_{\epsilon}\approx \frac{1}{u},\quad u\in[-3\epsilon,-2\epsilon].\label{eq: estimate on the cutoff}
\end{align}
\end{lemma}

{By the local well-posedness result in Proposition~\ref{prop: mixed lwp}}, the free data consist of a gauge choice of the center $\Gamma$, the data for $(r,\phi)$ on the outgoing cone $\{u = -\epsilon\}$, and the data for $(r,\phi)$ on the ingoing cone $\{v = 0\}$. We shall always fix our gauge of the center $\Gamma$ to be $\Gamma = \{z = -1\}$. The initial data on the characteristic initial hypersurface $\mathcal{H}_{\epsilon}$ are given as follows: \begin{itemize}
    \item On the ingoing cone $\{v = 0,\ -1\leq u\leq -\epsilon\}$, we have \begin{align*}
        &r_{\epsilon}(u,0) = \widetilde{r}(u,0),\\&\phi_{\epsilon}(u,0) = \phi_{1,\epsilon}(u,0)+i\phi_{2,\epsilon}(u,0) = \widetilde{\phi}_{1}(u,0)+\delta f_{1,\epsilon}(u)+i\left(\widetilde{\phi}_{2}(u,0)+\delta f_{2,\epsilon}(u)\right);
    \end{align*}
    \item On the outgoing cone $\{u = -\epsilon\}$, we have \begin{align*}
        &\partial_{v}r_{\epsilon}(-\epsilon,v) = \partial_{v}\widetilde{r}(-\epsilon,v),\\& \phi_{\epsilon}(-\epsilon,v) = \widetilde{\phi}_{1}(-\epsilon,v)+i\widetilde{\phi}_{2}(-\epsilon,v).
    \end{align*}
\end{itemize}

The solution to the Einstein--Maxwell-charged scalar field equations in the region $\mathcal{Q}^{(in)}\backslash\mathcal{Q}_{\epsilon}^{(in)}$ arising from the above initial data is denoted by $(r_{\epsilon},\Omega_{\epsilon},\phi_{\epsilon},Q_{\epsilon},A_{\epsilon})$, where $\phi_{\epsilon} = \phi_{1,\epsilon}+i\phi_{2,\epsilon}$ and $A_{\epsilon} = A_{\epsilon,u}du+A_{\epsilon,v}dv$. As a reminder, in view of the gauge invariance \eqref{eq: gauge invariance of the EMcsf} of the equations, we always fix the gauge freedom of the electromagnetic potential $A = A_{u}du+A_{v}dv$ by 
\begin{equation*}
    A_{v}\equiv 0,\quad A_{u}(u,0) = 0.
\end{equation*}

\subsubsection{Consequences of the initial perturbation on the initial hypersurface}
In this section, we study the consequences of the initial perturbations defined in Section \ref{sec: setup of the perturbations}. In particular, we study the behavior of $\partial_{u}r_{\epsilon}, Q_{\epsilon}, \partial_{u}\phi_{\epsilon}, A_{\epsilon,u}$ on the initial outgoing cone $\{u = -\epsilon\}$ and $\partial_{v}r_{\epsilon}, Q_{\epsilon},\partial_{v}\phi_{\epsilon}$ on the initial ingoing cone $\{v = 0\}$. Let $\Psi_{p}$ be the difference between $\Psi_{\epsilon}$ and $\widetilde{\Psi}$ for $\Psi\in\{r,\nu,\lambda,\mu,\phi\}$.

In view of the $u$-transport equation \eqref{eq:u-transport equation for Q} for the spacetime charge $Q$, we can prove the following proposition.
\begin{proposition}
\label{prop: initial estimate on Q}
    On the ingoing cone $\{v = 0,\ -1\leq u<0\}$, the spacetime charge $Q_{\epsilon}$ satisfies the following estimates 
    \begin{align}
    Q_{\epsilon}(u,0)&\equiv 0,\quad -2\epsilon\leq u<0,\label{estimate on Qepsilon for u close to singularity}\\
    \left\vert Q_{\epsilon}\right\vert(u,0)&\lesssim\vert q_{0}\vert\delta k u^{2}\left\vert\log(-u)\right\vert,\quad -3\epsilon\leq u\leq -2\epsilon,\label{estimate on Qepsilon for u in the intermedia region}\\
    \left\vert Q_{\epsilon}(u,0)-Q_{0}(u)\right\vert&\lesssim \vert q_{0}\vert\delta k\epsilon^{2}\left\vert\log(-\epsilon)\right\vert,\quad -1\leq u\leq-3\epsilon,\label{estimate for Qepsilon for u in the faraway region}\\
\left\vert Q_{0}(u)+\frac{1}{2}q_{0}\delta k_{2}\mr{\widetilde{r}}^{2}(0)u^{2}\right\vert&\lesssim \vert q_{0}\vert\delta k(-u)^{\gamma+2}\left\vert\log(-u)\right\vert, \quad -1\leq u<0,
    \end{align}
    where $Q_{0}$ is defined in \eqref{eq: definition of Q0}.
\end{proposition}
\begin{proof}

By the choice of initial data, we can rewrite \eqref{eq:u-transport equation for Q} as 
\begin{equation}
\label{eq: u transport equation in estimating Q}
    \begin{aligned}
        \partial_{u}Q_{\epsilon}(u,0) =& q_{0}r_{\epsilon}^{2}\left(\phi_{1,\epsilon}\partial_{u}\phi_{2,\epsilon}-\phi_{2,\epsilon}\partial_{u}\phi_{1,\epsilon}\right)\\ =&q_{0}r_{\epsilon}^{2}\left((\widetilde{\phi}_{1}+\delta f_{1}\chi_{\epsilon})(\partial_{u}\widetilde{\phi}_{2}+\delta\partial_{u}(f_{2}\chi_{\epsilon}))-(\widetilde{\phi}_{2}+\delta f_{2}\chi_{\epsilon})(\partial_{u}\widetilde{\phi}_{1}+\delta\partial_{u}(f_{1}\chi_{\epsilon}))\right)\\=&\underbrace{
        -q_{0}\delta r_{\epsilon}^{2}\widetilde{\phi}_{2}\partial_{u}\chi_{\epsilon}}_{\text{leading order term}}+q_{0}\delta r_{\epsilon}^{2}\left(\chi_{\epsilon}\partial_{u}\widetilde{\phi}_{2}\right)+q_{0}\delta r_{\epsilon}^{2}\left((f_{1}-1)\chi_{\epsilon}\partial_{u}\widetilde{\phi}_{2}-\widetilde{\phi}_{2}\partial_{u}((f_{1}-1)\chi_{\epsilon})\right)\\&+q_{0}\delta r_{\epsilon}^{2}\left(\widetilde{\phi}_{1}\partial_{u}(f_{2}\chi_{\epsilon})-(f_{2}\chi_{\epsilon})\partial_{u}\widetilde{\phi}_{1}\right)+q_{0}\delta^{2}r_{\epsilon}^{2}\left(f_{1}\chi_{\epsilon}\partial_{u}(f_{2}\chi_{\epsilon})-(f_{2}\chi_{\epsilon})\partial_{u}(f_{1}\chi_{\epsilon})\right).
    \end{aligned}
\end{equation}
Since on $[-2\epsilon,0)$ the cut-off function $\chi_{\epsilon}\equiv 0$, we have $Q_{\epsilon}\equiv0$ on $[-2\epsilon,0)$. On $[-3\epsilon,-2\epsilon]$, we have $\partial_{u}\chi_{\epsilon}\approx \frac{1}{u}$. Then by~\eqref{eq: condition on f}, we have \begin{equation*}
    \vert Q_{\epsilon}\vert\lesssim \vert q_{0}\vert\delta k u^{2}\vert\log(-u)\vert.
\end{equation*}
On the segment $[-1,-3\epsilon]$, we have $\chi_{\epsilon}\equiv 1$. Therefore, we have 
\begin{equation}
\label{eq: definition of Q0}
\begin{aligned}
    \partial_{u}Q_{\epsilon} =&\underbrace{q_{0}r_{\epsilon}^{2}\left(\delta f_{1}\partial_{u}\widetilde{\phi}_{2}+\delta(\partial_{u}f_{2})\widetilde{\phi}_{1}+\delta^{2}f_{1}\partial_{u}f_{2}-\delta f_{2}\partial_{u}\widetilde{\phi}_{1}-\delta(\partial_{u}f_{1})\widetilde{\phi}_{2}-\delta^{2}f_{2}\partial_{u}f_{1}\right)}_{\partial_{u}Q_{0}(u)}\\=&q_{0}\delta k_{2}\mr{\widetilde{r}}^{2}(0)(-u)+O(q_{0}\delta k (-u)^{\gamma+1}\vert\log(-u)\vert),
\end{aligned}
\end{equation}
where $Q_{0}(0) = 0$. Therefore, integrating the above equation, we have \begin{align*}
    &\left\vert Q_{\epsilon}(u,0)+\frac{1}{2}q_{0}\delta k_{2}\mr{\widetilde{r}}^{2}(0)u^{2}\right\vert\lesssim \vert q_{0}\vert\delta k (-u)^{\gamma+2}\vert \log(-u)\vert+\vert q_{0}\vert\delta k\epsilon^{2}\vert \log(-\epsilon)\vert,\\&
    \left\vert Q_{\epsilon}(u,0)-Q_{0}(u)\right\vert\lesssim \vert q_{0}\vert\delta k\epsilon^{2}\left\vert\log(-\epsilon)\right\vert.
\end{align*}
 This concludes the proof.
\end{proof}
\begin{remark}
    One can observe from Proposition~\ref{prop: initial estimate on Q} that for each point $(u,0)$ on the ingoing cone $\{v= 0\}$, the spacetime charge $Q$ converges to $Q_{0}\approx u^{2}$ as $\epsilon\rightarrow 0$ in view of the estimate~\eqref{estimate for Qepsilon for u in the faraway region}. However, since $f_{1}\rightarrow1$ as $u\rightarrow 0$, the limit of the truncated perturbation $(f_{1}\chi_{\epsilon})^{\prime}$ as $\epsilon\rightarrow 0$ will be a Dirac function. Therefore, on the segment $u\in[-3\epsilon,-2\epsilon]$, one can only get the bound~\eqref{estimate on Qepsilon for u in the intermedia region}, which is more singular than $Q_{0}\approx u^{2}$ when $u$ is sufficiently small.
\end{remark}

Next, we derive the quantitative estimates for $\lambda_{p}$ and $\mu_{p}$ on the ingoing cone $\{v = 0\}$. 

\begin{proposition}
\label{prop: initial estimate on mu}
On the ingoing cone $\{v = 0\}$, the quantity $\mu_{p}$ satisfies the following estimate
\begin{equation}
\label{eq: initial estimate on mu}\begin{aligned}
    &\mu_{p}(u,0)\equiv 0,\quad u\in[-2\epsilon,0),\\&
    \left\vert\mu_{p}(u,0)\right\vert\lesssim k\delta \left(1-\left(\frac{2\epsilon}{(-u)}\right)^{1+k^{2}}\right)+\delta^{2} k^{2}q_{0}^{2}\epsilon^{2},\quad u\in[-3\epsilon,-2\epsilon],\\&
    \left\vert\mu_{p}(u,0)\right\vert\lesssim k\delta,\quad u\in[-1,-3\epsilon].
\end{aligned} 
\end{equation}
\end{proposition}
\begin{proof}
    Recall that on the ingoing cone $\{v = 0\}$, the equation for $\mu$ takes the form \begin{equation*}
        \partial_{u}\mu_{\epsilon}+\left(\frac{\nu_{\epsilon}}{r_{\epsilon}}+\frac{r_{\epsilon}}{\nu_{\epsilon}}\left\vert\partial_{u}\phi_{\epsilon}\right\vert^{2}\right)\mu_{\epsilon} = \frac{r_{\epsilon}}{\nu_{\epsilon}}\left\vert\partial_{u}\phi_{\epsilon}\right\vert^{2}+\frac{\nu_{\epsilon}}{r_{\epsilon}^{3}}Q_{\epsilon}^{2}.
    \end{equation*}
We decompose $\mu_{\epsilon}$ into two parts $\mu_{\epsilon} = \mu_{\epsilon}^{(1)}+\mu_{\epsilon}^{(2)} $:\begin{align*}
&\partial_{u}\mu_{\epsilon}^{(1)}+\left(\frac{1}{u}+u\left\vert\partial_{u}\phi_{\epsilon}\right\vert^{2}\right)\mu_{\epsilon}^{(1)}=u\left\vert\partial_{u}\phi_{\epsilon}\right\vert^{2},\\&
    \partial_{u}\mu_{\epsilon}^{(2)}+\left(\frac{1}{u}+u\left\vert\partial_{u}\phi_{\epsilon}\right\vert^{2}\right)\mu_{\epsilon}^{(2)} = \frac{Q_{\epsilon}^{2}}{u^{3}},
    \end{align*}
    where we have used the fact that $\frac{\nu_{\epsilon}}{r_{\epsilon}}(u,0) = \frac{1}{u}$.
    Using the integrating factor, we have \begin{align*}
    \partial_{u}\left((-u)^{k^{2}+1}e^{\int_{0}^{u}\mathcal{P}_{\epsilon}(s)ds}\mu_{\epsilon}^{(1)}\right) =& -k^{2}(-u)^{k^{2}}e^{\int_{0}^{u}\mathcal{P}_{\epsilon}(s)ds}+(-u)^{k^{2}+1}\mathcal{P}_{\epsilon}(u)e^{\int_{0}^{u}\mathcal{P}_{\epsilon}(s)ds}\\=&
    \partial_{u}\left(\frac{k^{2}}{1+k^{2}}(-u)^{k^{2}+1}e^{\int_{0}^{u}\mathcal{P}_{\epsilon}(s)ds}\right)+\frac{1}{1+k^{2}}(-u)^{k^{2}+1}\mathcal{P}_{\epsilon}(u)e^{\int_{0}^{u}\mathcal{P}_{\epsilon}(s)ds},
\end{align*}
where \begin{equation}
\begin{aligned}
    \mathcal{P}_{\epsilon}(u) :=& -2k_{1}\delta \partial_{u}(f_{1}\chi_{\epsilon})-2k_{2}\delta \partial_{u}(f_{2}\chi_{\epsilon})+\delta^{2}u(\partial_{u}(f_{1}\chi_{\epsilon}))^{2}+\delta^{2}u(\partial_{u}(f_{2}\chi_{\epsilon}))^{2}\\=&
    -2k_{1}\delta \partial_{u}\chi_{\epsilon}+\delta^{2}u(\partial_{u}\chi_{\epsilon})^{2}-2k_{1}\delta\partial_{u}((f_{1}-1)\chi_{\epsilon})-2k_{2}\delta \partial_{u}(f_{2}\chi_{\epsilon})\\&+\delta^{2}u(\partial_{u}((f_{1}-1)\chi_{\epsilon}))^{2}+\delta^{2}u(\partial_{u}(f_{2}\chi_{\epsilon}))^{2}\\:=&
    -2k_{1}\delta\partial_{u}\chi_{\epsilon}+\delta^{2}u(\partial_{u}\chi_{\epsilon})^{2}+\mathcal{P}_{1,\epsilon}(u).
    \end{aligned}
    \label{eq: definition of Pepsilon}
\end{equation}
By the assumption~\eqref{eq: condition on f} on $f_{1}$ and $f_{2}$, we have \begin{equation}
    \left\vert\mathcal{P}_{1,\epsilon}\right\vert\lesssim \delta k (-u)^{\gamma-1}.
\end{equation}
Hence, we have \begin{equation*}
   \mu_{\epsilon}^{(1)}-\frac{k^{2}}{1+k^{2}}= \frac{1}{1+k^{2}}(-u)^{-(1+k^{2})}\int_{0}^{u}(-s)^{1+k^{2}}\mathcal{P}_{\epsilon}(s) e^{\int_{u}^{s}\mathcal{P}_{\epsilon}(t)dt}ds.
\end{equation*}
For $u\in[-2\epsilon,0)$, we have $\mathcal{P}_{\epsilon} = 0$. It follows that $\mu_{\epsilon}^{(1)} = \frac{k^{2}}{1+k^{2}}$ for $u\in[-2\epsilon,0)$. For $u\in[-3\epsilon,-2\epsilon]$, we have \begin{equation*}
    \left\vert\mu_{\epsilon}^{(1)}-\frac{k^{2}}{1+k^{2}}\right\vert\lesssim k\delta (-u)^{-(1+k^{2})}\int_{u}^{-2\epsilon}(-s)^{k^{2}}ds\lesssim k\delta \left(1-\left(\frac{2\epsilon}{(-u)}\right)^{1+k^{2}}\right).
\end{equation*}
For $u\in[-1,-3\epsilon]$, we have \begin{align*}
    \left\vert\mu_{\epsilon}^{(1)}-\frac{k^{2}}{1+k^{2}}\right\vert\lesssim& k\delta (-u)^{-(1+k^{2})}\int_{-3\epsilon}^{-2\epsilon}(-s)^{k^{2}}ds+k\delta (-u)^{-(1+k^{2})}\int_{u}^{-3\epsilon}(-s)^{1+k^{2}}(-s)^{\gamma-1} ds\\\lesssim &k\delta \left(\frac{\epsilon}{-u}\right)^{1+k^{2}}+k\delta \left((-u)^{\gamma}-\frac{(3\epsilon)^{\gamma+k^{2}+1}}{(-u)^{1+k^{2}}}\right)\lesssim k\delta.
\end{align*}
For $\mu_{\epsilon}^{(2)}$, similarly we have \begin{equation*}
    \mu_{\epsilon}^{(2)} = (-u)^{-(1+k^{2})}\int_{0}^{u}\frac{Q_{\epsilon}^{2}}{s^{3}}(-s)^{1+k^{2}}e^{\int_{u}^{s}\mathcal{P}_{\epsilon}(t)dt}ds.
\end{equation*}

Using \eqref{estimate on Qepsilon for u close to singularity}--\eqref{estimate for Qepsilon for u in the faraway region}, we have \begin{align*}
    \mu_{\epsilon}^{(2)}(u,0)&\equiv 0,\quad -2\epsilon\leq u<0,\\
    \left\vert\mu_{\epsilon}^{(2)}\right\vert(u,0)&\lesssim q_{0}^{2}k^{2}\delta^{2}u^{2}(\log(-u))^{2}\left(1-\left(\frac{2\epsilon}{(-u)}\right)^{1+k^{2}}\right),\quad -3\epsilon\leq u\leq-2\epsilon,\\
    \left\vert\mu_{\epsilon}^{(2)}\right\vert(u,0)&\lesssim k\delta,\quad -1\leq u\leq-3\epsilon.
\end{align*}
Putting everything together concludes the proof.
\end{proof}
\begin{proposition}
\label{prop: initial estimate on Omegap}
    On the ingoing cone $\{v = 0\}$, we have that \begin{align}
    \label{eq: initial estimate on Omegap}    &\left\vert(\Omega^{2})_{p}\right\vert\lesssim k\delta (-u)^{k^{2}},\\&
    \left\vert\lambda_{p}\right\vert\lesssim k\delta (-u)^{k^{2}}.\label{eq: initial estimate on lambdap}
    \end{align}
\end{proposition}
\begin{proof}
    We consider the Raychaudhuri equation \eqref{eq:spherical symmetric equaion1} on the ingoing cone $\{v = 0\}$: \begin{equation*}
        \partial_{u}\log\Omega_{\epsilon}^{2} = \frac{k^{2}}{u}+\mathcal{P}_{\epsilon}(u).
    \end{equation*}
    Solving the equation, we have \begin{equation*}
        \Omega_{\epsilon}^{2}(u,0) = \widetilde{\Omega}^{2}(u,0)e^{\int_{0}^{u}\mathcal{P}_{\epsilon}(s)ds}.
    \end{equation*}
    For $u\leq-3\epsilon$, we have \begin{equation*}
        \int_{0}^{u}\mathcal{P}_{\epsilon}(s)ds = \int_{-2\epsilon}^{-3\epsilon}-2k_{1}\delta\partial_{u}\chi_{\epsilon}+\delta^{2}u(\partial_{u}\chi_{\epsilon})^{2} du+\int_{-2\epsilon}^{u}\mathcal{P}_{1,\epsilon}(s)ds.
    \end{equation*}
    By the choice of the cut-off functions, we have that \begin{equation*}
        \int_{-2\epsilon}^{-3\epsilon}-2k_{1}\delta\partial_{u}\chi_{\epsilon}+\delta^{2}u(\partial_{u}\chi_{\epsilon})^{2} du: = L_{1},
    \end{equation*}
    for some constant $L_{1}$ independent of $\epsilon$. Moreover, we have \begin{equation*}
        \left\vert L_{1}\right\vert\lesssim k\delta.
    \end{equation*} 
    Hence, we have \begin{equation*}
        \left\vert\Omega^{2}_{\epsilon}-\widetilde{\Omega}^{2}\right\vert(u,0)\lesssim k\delta (-u)^{k^{2}}.
    \end{equation*}
    Since $\lambda_{\epsilon} = \frac{\Omega_{\epsilon}^{2}}{4\nu_{\epsilon}}(1-\mu_{\epsilon})$, using \eqref{eq: initial estimate on mu}, we obtain the estimate~\eqref{eq: initial estimate on lambdap}. This concludes the proof.
\end{proof}

\begin{proposition}
\label{prop: initial estimate on dvphi}
    On the ingoing cone $\{v = 0,\ -1\leq u\leq -\epsilon\}$, the quantities $\left(\partial_{v}(r\phi_{1})\right)_{p}$ and $\left(\partial_{v}(r\phi_{2})\right)_{p}$ satisfy the following estimates \begin{align}
        (-u)^{-k^{2}}\left\vert\partial_{v}(r_{\epsilon}\phi_{1,p})\right\vert\lesssim \delta, \quad
        (-u)^{-k^{2}}\left\vert\partial_{v}(r_{\epsilon}\phi_{2,p})\right\vert\lesssim\delta.
    \end{align}
\end{proposition}
    \begin{proof}
        On the ingoing cone $\{v = 0\}$, we consider the difference between the equations for $\widetilde{\phi}$ and $\phi_{\epsilon}$ \eqref{eq: uncharged phi1 wave equation}-\eqref{eq: uncharged phi2 wave equation}. We have 
        \begin{align}
            \partial_{u}\partial_{v}(r_{\epsilon}\phi_{1,p})&=\frac{q_{0}}{u}\frac{\lambda_{\epsilon}}{1-\mu_{\epsilon}}Q_{\epsilon}\phi_{2,\epsilon}+\frac{k_{1}}{(-u)}\lambda_{p},\label{eq: reduction of dvphi1 on the ingoing cone}\\
            \partial_{u}\partial_{v}(r_{\epsilon}\phi_{2,p})& = \frac{q_{0}}{(-u)}\frac{\lambda_{\epsilon}}{1-\mu_{\epsilon}}Q_{\epsilon}\phi_{1,\epsilon}+\frac{k_{2}}{(-u)}\lambda_{p}.\label{eq: reduction of dvphi2 on the ingoing cone}
        \end{align}
        On the domain $u\in[-2\epsilon,0)$, we have that $(\partial_{v}(r\phi_{1}))_{p} = (\partial_{v}(r\phi_{2}))_{p}\equiv 0$. Using \eqref{eq: condition on f}, \eqref{estimate on Qepsilon for u close to singularity}--\eqref{estimate for Qepsilon for u in the faraway region}, \eqref{eq: initial estimate on mu}, and \eqref{eq: initial estimate on lambdap}, we have \begin{align*}
            &\left\vert\frac{q_{0}}{u}\frac{\lambda_{\epsilon}}{1-\mu_{\epsilon}}Q_{\epsilon}\phi_{1,\epsilon}\right\vert\lesssim \delta kq_{0}^{2}(-u)^{1+k^{2}}\vert\log(-u)\vert,\quad \left\vert\frac{q_{0}}{u}\frac{\lambda_{\epsilon}}{1-\mu_{\epsilon}}Q_{\epsilon}\phi_{2,\epsilon}\right\vert\lesssim \delta kq_{0}^{2}(-u)^{1+k^{2}}\vert\log(-u)\vert,\\&
            \left\vert\frac{1}{(-u)}\lambda_{p}\right\vert\lesssim \delta k(-u)^{k^{2}-1}.
        \end{align*}
        Hence, integrating \eqref{eq: reduction of dvphi1 on the ingoing cone}-\eqref{eq: reduction of dvphi2 on the ingoing cone} concludes the proof.
    \end{proof}
\subsubsection{Bootstrap assumption in the truncated region}
In this section, we set up our bootstrap assumption. Motivated by the bounds in Proposition \ref{prop: initial estimate on Q}, Proposition \ref{prop: initial estimate on mu}, Proposition \ref{prop: initial estimate on Omegap}, and Proposition \ref{prop: initial estimate on dvphi}, for a region $\mathcal{D}\subset\mathcal{Q}^{(in)}\backslash\mathcal{Q}^{(in)}_{\epsilon}$, we define \begin{align*}
    &\mathcal{R}_{Q}^{\epsilon}(\mathcal{D}) := \sup_{\mathcal{D}}\left\vert\frac{Q_{\epsilon}}{u^{2}\log(-u)}\right\vert,\quad \mathcal{R}_{\nu_{p}}^{\epsilon}(\mathcal{D}) := \sup_{\mathcal{D}}\left\vert\nu_{p}\right\vert,\quad \mathcal{R}_{\lambda_{p}}^{\epsilon}(\mathcal{D}): = \sup_{\mathcal{D}}\left\vert\frac{\lambda_{p}}{(-u)^{k^{2}}}\right\vert,\\&
    \mathcal{R}_{\partial_{v}\phi_{p}}^{\epsilon}(\mathcal{D}): = \sup_{\mathcal{D}}(-u)^{q_{k}}\left\vert\partial_{v}\phi_{1,p}\right\vert+\sup_{\mathcal{D}}(-u)^{q_{k}}\left\vert\partial_{v}\phi_{2,p}\right\vert,\\&\mathcal{R}_{\partial_{u}\phi_{p}}^{\epsilon}(\mathcal{D}): = \sup_{\mathcal{D}}(-u)\left\vert\partial_{u}\phi_{1,p}\right\vert+\sup_{\mathcal{D}}(-u)\left\vert\partial_{u}\phi_{2,p}\right\vert,\\&
    \mathcal{R}^{\epsilon} (\mathcal{D}):=\mathcal{R}_{Q}^{\epsilon}(\mathcal{D})+\mathcal{R}_{\nu_{p}} ^{\epsilon}(\mathcal{D})+\mathcal{R}_{\lambda_{p}}^{\epsilon}(\mathcal{D})+\mathcal{R}_{\partial_{v}(\phi_{p})}^{\epsilon}(\mathcal{D})+\mathcal{R}_{\partial_{u}(\phi_{p})}^{\epsilon}(\mathcal{D}).
\end{align*}
{By the local well-posedness result in Proposition~\ref{prop: mixed lwp}}, for given initial data in Section \ref{sec: setup of the perturbations}, there exists $\widetilde{\epsilon}>0$ such that the solution to the Einstein--Maxwell-charged scalar field equations arising from this initial data exists in the region $\mathcal{Q}^{(in)}_{\epsilon+\widetilde{\epsilon}}\backslash\mathcal{Q}_{\epsilon}^{(in)}$. Now we assume that there exists a universal constant $C_{0}$ such that in the region $\mathcal{Q}_{\epsilon+\widetilde{\epsilon}}^{(in)}\backslash\mathcal{Q}_{\epsilon}^{(in)}$, we have that \begin{equation}
    \mathcal{R}^{\epsilon}(\mathcal{Q}_{\epsilon+\widetilde{\epsilon}}^{(in)}\backslash\mathcal{Q}_{\epsilon}^{(in)})\leq 2C_{0}\delta.\label{eq: bootstrap assumption}
\end{equation}
Our goal in the next few sections is to improve the bound \begin{equation*}
    \mathcal{R}^{\epsilon}(\mathcal{Q}_{\epsilon+\widetilde{\epsilon}}^{(in)}\backslash\mathcal{Q}_{\epsilon}^{(in)})\leq C_{0}\delta.
\end{equation*}
Then the standard local well-posedness and extension principle will give the existence of the solution in the whole region $\mathcal{Q}^{(in)}\backslash\mathcal{Q}^{(in)}_{\epsilon}$.

To facilitate our limiting argument, we also extend the definition of $\mathcal{R}^{\epsilon}$ to measure the difference between $(r_{\epsilon_{1}},\mu_{\epsilon_{1}},\phi_{\epsilon_{1}},Q_{\epsilon_{1}})$ and $(r_{\epsilon_{2}},\mu_{\epsilon_{2}},\phi_{\epsilon_{2}},Q_{\epsilon_{2}})$. \begin{align*}
    &\mathcal{R}_{Q}^{\epsilon_{1},\epsilon_{2}}(\mathcal{D}) := \sup_{\mathcal{D}}\left\vert\frac{Q_{\epsilon_{1}}-Q_{\epsilon_{2}}}{u^{2}\log(-u)}\right\vert,\quad \mathcal{R}_{\nu}^{\epsilon_{1},\epsilon_{2}}(\mathcal{D}) := \sup_{\mathcal{D}}\left\vert\nu_{\epsilon_{1}}-\nu_{\epsilon_{2}}\right\vert,\quad \mathcal{R}_{\lambda}^{\epsilon_{1},\epsilon_{2}}(\mathcal{D}): = \sup_{\mathcal{D}}\left\vert\frac{\lambda_{\epsilon_{1}}-\lambda_{\epsilon_{2}}}{(-u)^{k^{2}}}\right\vert,\\&
    \mathcal{R}_{\partial_{v}\phi}^{\epsilon_{1},\epsilon_{2}}(\mathcal{D}): = \sup_{\mathcal{D}}(-u)^{q_{k}}\left\vert\partial_{v}\phi_{1,\epsilon_{1}}-\partial_{v}\phi_{1,\epsilon_{2}}\right\vert+\sup_{\mathcal{D}}(-u)^{q_{k}}\left\vert\partial_{v}\phi_{2,\epsilon_{1}}-\partial_{v}\phi_{2,\epsilon_{2}}\right\vert,\\&\mathcal{R}_{\partial_{u}\phi}^{\epsilon_{1},\epsilon_{2}}(\mathcal{D}): = \sup_{\mathcal{D}}(-u)\left\vert\partial_{u}\phi_{1,\epsilon_{1}}-\partial_{u}\phi_{1,\epsilon_{2}}\right\vert+\sup_{\mathcal{D}}(-u)\left\vert\partial_{u}\phi_{2,\epsilon_{1}}-\partial_{u}\phi_{2,\epsilon_{2}}\right\vert,\\&
    \mathcal{R}^{\epsilon_{1},\epsilon_{2}}(\mathcal{D}) :=\mathcal{R}_{Q}^{\epsilon_{1},\epsilon_{2}}(\mathcal{D})+\mathcal{R}_{\nu} ^{\epsilon_{1},\epsilon_{2}}(\mathcal{D})+\mathcal{R}_{\lambda}^{\epsilon_{1},\epsilon_{2}}(\mathcal{D})+\mathcal{R}_{\partial_{v}\phi}^{\epsilon_{1},\epsilon_{2}}(\mathcal{D})+\mathcal{R}_{\partial_{u}\phi}^{\epsilon_{1},\epsilon_{2}}(\mathcal{D}).
\end{align*}
In the remainder of this section, we prove some direct consequences of the bootstrap assumption. 

We first derive the following equations for the perturbed quantities $\Psi_{p} = \Psi_{\epsilon}-\widetilde{\Psi}$, where $\Psi\in\{r,\lambda,\nu,\phi,\partial_{u}\phi,\partial_{v}\phi,Q,A\}$.
\begin{align}
\partial_{u}\lambda_{p} = &\left(\frac{\mu}{1-\mu}\frac{\lambda\nu}{r}\right)_{p}-\frac{\lambda_{\epsilon}\nu_{\epsilon}}{(1-\mu_{\epsilon})r_{\epsilon}}\frac{Q_{\epsilon}^{2}}{r_{\epsilon}^{2}},\label{eq: difference equation for lambdap in bootstrap argument}\allowdisplaybreaks\\
\partial_{v}\nu_{p}=&\left(\frac{\mu}{1-\mu}\frac{\lambda\nu}{r}\right)_{p}-\frac{\lambda_{\epsilon}\nu_{\epsilon}}{(1-\mu_{\epsilon})r_{\epsilon}}\frac{Q_{\epsilon}^{2}}{r_{\epsilon}^{2}},\allowdisplaybreaks\\
    \partial_{u}\mu_{p}+\left(\frac{\nu_{\epsilon}}{r_{\epsilon}}+\frac{r_{\epsilon}}{\nu_{\epsilon}}\left\vert D_{u}\phi_{\epsilon}\right\vert^{2}\right)\mu_{p} = &-\left[\left(\frac{\nu}{r}\right)_{p}+\left(\frac{r_{\epsilon}}{\nu_{\epsilon}}\left\vert D_{u}\phi_{\epsilon}\right\vert^{2}-\frac{\widetilde{r}}{\widetilde{\nu}}\left\vert\partial_{u}\widetilde{\phi}\right\vert^{2}\right)\right]\widetilde\mu\nonumber\\&+\frac{r_{\epsilon}}{\nu_{\epsilon}}\left\vert D_{u}\phi_{\epsilon}\right\vert^{2}-\frac{\widetilde{r}}{\widetilde{\nu}}\left\vert\partial_{u}\widetilde{\phi}\right\vert^{2}+\frac{\nu_{\epsilon}}{r_{\epsilon}^{3}}Q_{\epsilon}^{2},\allowdisplaybreaks\\\partial_{v}\mu_{p}+\left(\frac{\lambda_{\epsilon}}{r_{\epsilon}}+\frac{r_{\epsilon}}{\lambda_{\epsilon}}\left\vert\partial_{v}\phi_{\epsilon}\right\vert^{2}\right)\mu_{p}=&-\left[\left(\frac{\lambda}{r}\right)_{p}+\left(\frac{r_{\epsilon}}{\lambda_{\epsilon}}\left\vert\partial_{v}\phi_{\epsilon}\right\vert^{2}-\frac{\widetilde{r}}{\widetilde{\lambda}}\left\vert\partial_{v}\widetilde{\phi}\right\vert^{2}\right)\right]\widetilde{\mu}\nonumber\\&+\frac{r_{\epsilon}}{\lambda_{\epsilon}}\left\vert\partial_{v}\phi_{\epsilon}\right\vert^{2}-\frac{\widetilde{r}}{\widetilde{\lambda}}\left\vert\partial_{v}\widetilde{\phi}\right\vert^{2}+\frac{\lambda_{\epsilon}}{r_{\epsilon}^{3}}Q_{\epsilon}^{2},\label{eq: difference equations for dvmu}\allowdisplaybreaks\\
    \partial_{u}\partial_{v}(r_{\epsilon}\phi_{1,p})=& q_{0}\left(A_{u,\epsilon}\partial_{v}(r\phi_{2})_{\epsilon}+\frac{\lambda_{\epsilon}\nu_{\epsilon}}{1-\mu_{\epsilon}}\frac{Q_{\epsilon}}{r_{\epsilon}^{2}}(r\phi_{2})_{\epsilon}\right)-\nu_{p}\partial_{v}\widetilde{\phi}_{1}-\lambda_{p}\partial_{u}\widetilde{\phi_{1}}\nonumber\\&-r_{p}\partial_{u}\partial_{v}\widetilde{\phi}_{1}-\phi_{1,p}\partial_{u}\partial_{v}r_{\epsilon},\label{eq: difference equation for phi1 in the bootstrap argument}\allowdisplaybreaks\\
    \partial_{u}\partial_{v}(r_{\epsilon}\phi_{2,p})=& -q_{0}\left(A_{u,\epsilon}\partial_{v}(r\phi_{1})_{\epsilon}+\frac{\lambda_{\epsilon}\nu_{\epsilon}}{1-\mu_{\epsilon}}\frac{Q_{\epsilon}}{r_{\epsilon}^{2}}(r\phi_{1})_{\epsilon}\right)-\nu_{p}\partial_{v}\widetilde{\phi}_{2}-\lambda_{p}\partial_{u}\widetilde{\phi_{2}}\nonumber\\&-r_{p}\partial_{u}\partial_{v}\widetilde{\phi}_{2}-\phi_{2,p}\partial_{u}\partial_{v}r_{\epsilon},\label{eq: difference equation for phi2 in the bootstrap argument}
\end{align}
\begin{proposition}
    Under the bootstrap assumption \eqref{eq: bootstrap assumption}, for $\delta$ sufficiently small, in the region $\mathcal{Q}_{\epsilon+\widetilde{\epsilon}}^{(in)}\backslash\mathcal{Q}_{\epsilon}^{(in)}$, we have \begin{align*}
      & \lambda_{\epsilon}\approx (-u)^{k^{2}},\quad r_{\epsilon}\approx (-u)(z+1),\quad -\nu_{\epsilon}\approx 1,\quad \left\vert r_{p}\right\vert\lesssim C_{0}\delta (-u)(z+1)\\&
      \left\vert\partial_{u}(\phi_{1,\epsilon})\right\vert\lesssim \frac{k}{(-u)},\quad \left\vert\partial_{u}\phi_{2,\epsilon}\right\vert\lesssim \frac{k}{(-u)},\quad \left\vert\partial_{v}\phi_{1,\epsilon}\right\vert\lesssim k(-u)^{k^{2}-1},\quad \left\vert\partial_{v}\phi_{2,\epsilon}\right\vert\lesssim k(-u)^{k^{2}-1},\\&
      \left\vert Q_{\epsilon}\right\vert\lesssim C_{0}\delta u^{2}\left\vert\log(-u)\right\vert,\quad \left\vert A_{u,\epsilon}\right\vert\lesssim C_{0}\delta (-u)\left\vert\log(-u)\right\vert\left\vert z\right\vert,\quad 
      \left\vert\mu_{p}\right\vert\lesssim C_{0}k\delta,\quad \left\vert \mu_{\epsilon}\right\vert\lesssim k^{2}.
    \end{align*}
    \label{prop: consequence of the bootstrap in the existence argument}
\end{proposition}
\begin{proof}
    It suffices to prove the estimates for $r_{p}$, $A_{u,\epsilon}$, $\mu_{p}$, and $\mu_{\epsilon}$. All the other estimates follow from the bootstrap assumption \eqref{eq: bootstrap assumption} and Proposition \ref{prop: estimate on the background spacetime}. Since \begin{equation*}
        r_{p}(u,v)  = \int_{-(-u)^{q_{k}}}^{v}\lambda_{p}d\widetilde{v},
    \end{equation*}
    we have that \begin{equation*}
        \left\vert r_{p}\right\vert\leq\int_{-(-u)^{q_{k}}}^{v}\left\vert\lambda_{p}\right\vert\lesssim \delta(-u)^{k^{2}}\left(v+(-u)^{q_{k}}\right)\lesssim \delta (-u)(z+1).
    \end{equation*}
 Using the equation \eqref{eq:spherical symmetric equation last}, we have \begin{equation}
    \begin{aligned}
        \left\vert A_{u,\epsilon}\right\vert(u,v)\leq\int_{v}^{0}\left\vert\frac{2Q_{\epsilon}}{r_{\epsilon}^{2}}\frac{\lambda_{\epsilon}\nu_{\epsilon}}{1-\mu_{\epsilon}}\right\vert(u,\widetilde v)d\widetilde v\lesssim C_{0}\delta(-u)^{k^{2}}\left\vert\log(-u)\right\vert(-v)\lesssim C_{0}\delta (-u)\left\vert\log(-u)\right\vert \left\vert z\right\vert.
        \end{aligned}
    \end{equation}
    To estimate $\mu_{p}$, we can rewrite the equation \eqref{eq: difference equations for dvmu} as \begin{equation}
        \partial_{v}(e^{\int_{0}^{v}\frac{r_{\epsilon}}{\lambda_{\epsilon}}\left\vert\partial_{v}\phi_{\epsilon}\right\vert^{2}}r_{\epsilon}\mu_{p}) = r_{\epsilon}e^{\int_{0}^{v}\frac{r_{\epsilon}}{\lambda_{\epsilon}}\left\vert\partial_{v}\phi_{\epsilon}\right\vert^{2}}(\text{RHS of} \eqref{eq: difference equations for dvmu}).\label{eq: for estimate the bootstrap of mu}
    \end{equation}
    Using the estimates on $\lambda_{\epsilon}$, $r_{\epsilon}$, and $\partial_{v}\phi_{\epsilon}$, we have that \begin{equation*}
        e^{\int_{-(-u)^{q_{k}}}^{v}\frac{r_{\epsilon}}{\lambda_{\epsilon}}\left\vert\partial_{v}\phi_{\epsilon}\right\vert^{2}}
    \end{equation*}
    is uniformly bounded. Integrating the equation~\eqref{eq: for estimate the bootstrap of mu}, we have \begin{equation}
        \label{eq: final version of the muepsilon bootstrap estimet}\vert\mu_{p}\vert\lesssim k\delta+\underbrace{\int_{v}^{0}\left\vert\left[\left(\frac{\lambda}{r}\right)_{p}+\left(\frac{r_{\epsilon}}{\lambda_{\epsilon}}\vert\partial_{v}\phi_{\epsilon}\vert^{2}-\frac{\widetilde{r}}{\widetilde{\lambda}}\left\vert\partial_{v}\widetilde{\phi}\right\vert^{2}\right)\right]\widetilde{\mu}\right\vert}_{I}+\underbrace{\int_{v}^{0}\left\vert\frac{r_{\epsilon}}{\lambda_{\epsilon}}\vert\partial_{v}\phi_{\epsilon}\vert^{2}-\frac{\widetilde{r}}{\widetilde{\lambda}}\vert\partial_{v}\widetilde{\phi}\vert^{2}\right\vert}_{II}+\underbrace{\int_{v}^{0}\left\vert\frac{\lambda_{\epsilon}}{r_{\epsilon}^{3}}\right\vert Q_{\epsilon}^{2}}_{III}.
    \end{equation}
    For $I$ on the right-hand side of~\eqref{eq: final version of the muepsilon bootstrap estimet}, due to the presence of $\widetilde{\mu}$, we have \begin{equation*}
        \vert I\vert\lesssim C_{0}k\delta.
    \end{equation*}
    For $II$ on the right-hand side of~\eqref{eq: final version of the muepsilon bootstrap estimet}, we have \begin{align*}
        II\lesssim \int_{v}^{0}\frac{r_{\epsilon}}{\lambda_{\epsilon}}\left\vert\vert\partial_{v}\phi_{\epsilon}\vert^{2}-\vert\partial_{v}\widetilde{\phi}\vert^{2}
        \right\vert+\left\vert\left(\frac{r}{\lambda}\right)_{p}\right\vert\vert\partial_{v}\widetilde{\phi}\vert^{2}.
    \end{align*}
    Due to the estimate~\eqref{eq: estimate on dvphitilde} for $\partial_{v}\widetilde{\phi}$, we have \begin{equation*}
        II\lesssim C_{0}k\delta.
    \end{equation*}
    For the term $III$ on the right-hand side of~\eqref{eq: final version of the muepsilon bootstrap estimet}, using the bootstrap assumption, we have \begin{equation*}
        III\lesssim C_{0}^{2}\delta^{2}.
    \end{equation*}
    Hence, we have \begin{equation}
        \left\vert\mu_{p}\right\vert\lesssim C_{0}(k\delta+\delta^{2}).
    \end{equation}
    Taking $\delta\ll k$, we can conclude the proof.
\end{proof}

\subsubsection{Closing the bootstrap argument in the truncated region}
\label{sec: closing the bootstrap}
In this section, we close the bootstrap argument for a given scalar field charge $q_{0}$ and $k$ sufficiently small.
\begin{proposition}
\label{prop: closing the bootstrap}
    Assuming \eqref{eq: bootstrap assumption}, for any given scalar field charge $q_{0}$, there exists $k_{0}$ sufficiently small, depending on $q_{0}$ and a universal constant $C_{0}$, such that for any $0<k<k_{0}$, there exists $\delta_{0}$ sufficiently small depending on $k$, such that $\mathcal{R}^{\epsilon}(\mathcal{Q}_{\epsilon+\widetilde{\epsilon}}^{(in)}\backslash\mathcal{Q}_{\epsilon}^{(in)})\leq C_{0}\delta$.
\end{proposition}

\begin{proof}
For the spacetime charge $Q_{\epsilon}$, we consider the equation \eqref{eq:v-transport equation for Q} \begin{align*}
    \partial_{v}Q_{\epsilon} =& q_{0}r_{\epsilon}^{2}(\phi_{2,\epsilon}\partial_{v}\phi_{1,\epsilon}-\phi_{1,\epsilon}\partial_{v}\phi_{2,\epsilon})\\=&
    q_{0}r_{\epsilon}^{2}(\phi_{2,\epsilon}-\widetilde{\phi}_{2}+\widetilde{\phi}_{2})(\partial_{v}\phi_{1,\epsilon}-\partial_{v}\widetilde{\phi}_{1}+\partial_{v}\widetilde{\phi}_{1})\\&-q_{0}r_{\epsilon}^{2}(\phi_{1,\epsilon}-\widetilde{\phi}_{1}+\widetilde{\phi}_{1})(\partial_{v}\phi_{2,\epsilon}-\partial_{v}\widetilde{\phi}_{2}+\partial_{v}\widetilde{\phi}_{2}).
\end{align*}
Since \begin{equation*}
    \widetilde{\phi}_{2}\partial_{v}\widetilde{\phi}_{1}-\widetilde{\phi}_{1}\partial_{v}\widetilde{\phi}_{2} = 0,
\end{equation*}
using the bootstrap assumption~\eqref{eq: bootstrap assumption}, we have that \begin{equation*}
     \left\vert \partial_{v}Q_{\epsilon}\right\vert\lesssim \vert q_{0}\vert k\delta u^{2-q_k}\left\vert \log(-u)\right\vert.
 \end{equation*}
 Integrating from $\{v=  0\}$ and using \eqref{estimate on Qepsilon for u close to singularity}--\eqref{estimate for Qepsilon for u in the faraway region}, we can conclude that \begin{equation}
     \left\vert Q_{\epsilon}\right\vert\lesssim k\delta \vert q_{0}\vert(-u)^{2}\left\vert \frac{v}{(-u)^{q_{k}}}\right\vert \left\vert\log(-u)\right\vert+k\delta \vert q_{0}\vert u^{2}\left\vert\log(-u)\right\vert\lesssim k\delta \vert q_{0}\vert u^{2}\left\vert\log(-u)\right\vert.
 \end{equation}
 Taking $k$ to be sufficiently small, we have that \begin{equation*}
     \mathcal{R}_{Q}^{\epsilon}(\mathcal{Q}_{\epsilon+\widetilde{\epsilon}}^{(in)}\backslash\mathcal{Q}_{\epsilon}^{(in)})\leq\frac{1}{100}C_{0}\delta.
 \end{equation*}
 Next, we consider the quantity $\lambda_{p}$. We first estimate the right-hand side of \eqref{eq: difference equation for lambdap in bootstrap argument} \begin{align*}
     \left\vert\text{RHS of \eqref{eq: difference equation for lambdap in bootstrap argument}}\right\vert\lesssim& \frac{\mu_{\epsilon}}{1-\mu_{\epsilon}}\left\vert\left(\frac{\lambda\nu}{r}\right)_{p}\right\vert+\frac{\widetilde{\lambda}\widetilde{\nu}}{\widetilde{r}}\left\vert\left(\frac{\mu}{1-\mu}\right)_{p}\right\vert+\left\vert\frac{\lambda_{\epsilon}\nu_{\epsilon}}{(1-\mu_{\epsilon})r_{\epsilon}}\right\vert\frac{Q_{\epsilon}^{2}}{r_{\epsilon}^{2}}\\\lesssim&C_{0}
     k^{2}\delta (-u)^{k^{2}-1}+C_{0}k\delta (-u)^{k^{2}-1}+C_{0}^{2}\delta^{2}(-u)^{1+k^{2}}\left(\log(-u)\right)^{2}.
 \end{align*}
 We consider two cases here. For $v\in[-\epsilon^{q_{k}},0]$, integrating from the initial outgoing cone, we have \begin{equation*}
     \left\vert\lambda_{p}\right\vert(u,v)\lesssim k C_{0}^{2}\delta (-u)^{k^{2}}.
 \end{equation*}
 For $v\in[-(-u)^{q_{k}},-\epsilon^{q_{k}}]$, integrating from the center $\Gamma$, we have \begin{equation*}
     \left\vert\lambda_{p}\right\vert(u,v)\lesssim kC_{0}^{2}\delta (-u)^{k^{2}}+\left\vert\lambda_{p}\right\vert(-(-v)^{\frac{1}{q_{k}}},v).
 \end{equation*}
  On the center $\Gamma$, we have that \begin{equation*}
     \nu_{p}|_{\Gamma}+q_{k}(-u)^{-k^{2}}\lambda_{p}|_{\Gamma} = 0.
 \end{equation*}
 Therefore, taking $k$ sufficiently small, we can conclude that \begin{equation*}
     \mathcal{R}_{\lambda_{p}}^{\epsilon}(\mathcal{Q}_{\epsilon+\widetilde{\epsilon}}^{(in)}\backslash\mathcal{Q}_{\epsilon}^{(in)})\leq\frac{1}{100}C_{0}\delta.
 \end{equation*}
 Similarly, we can show that \begin{equation*}
     \mathcal{R}_{\nu_{p}}^{\epsilon}(\mathcal{Q}_{\epsilon+\widetilde{\epsilon}}^{(in)}\backslash\mathcal{Q}_{\epsilon}^{(in)})\leq\frac{1}{100}C_{0}\delta.
 \end{equation*}
 Next, we consider $\partial_{u}(r_{\epsilon}\phi_{p})$. For the right-hand side of \eqref{eq: difference equation for phi1 in the bootstrap argument}, we have \begin{equation*}
     \left\vert\text{RHS of \eqref{eq: difference equation for phi1 in the bootstrap argument}}\right\vert\lesssim C_{0}^{2}\vert q_{0}\vert\delta k(-u)^{1+k^{2}}\left\vert\log(-u)\right\vert^{2}+\frac{\delta k}{\eta} (-u)^{-q_{k}}\left\vert z\right\vert^{-\eta}.
 \end{equation*}
 Similarly, for the right-hand side of \eqref{eq: difference equation for phi2 in the bootstrap argument}, we have \begin{equation*}
     \left\vert\text{RHS of \eqref{eq: difference equation for phi2 in the bootstrap argument}}\right\vert\lesssim \vert q_{0}\vert\delta k(-u)^{1+k^{2}}\left\vert\log(-u)\right\vert^{2}+\frac{\delta k}{\eta} (-u)^{-q_{k}}\left\vert z\right\vert^{-
     \eta}.
 \end{equation*}
 Hence, integrating from $\{v= 0\}$ and taking $k$ sufficiently small and $\eta\gg k$, we have \begin{equation*}
     \sup_{\mathcal{Q}_{\epsilon+\widetilde{\epsilon}}^{(in)}\backslash\mathcal{Q}_{\epsilon}^{(in)}}\left\vert\partial_{u}(r_{\epsilon}\phi_{1,p})\right\vert\leq \frac{1}{100}C_{0}\delta,\quad \sup_{\mathcal{Q}_{\epsilon+\widetilde\epsilon}^{(in)}\backslash\mathcal{Q}_{\epsilon}^{(in)}}\left\vert\partial_{u}(r_{\epsilon}\phi_{2,p})\right\vert\leq \frac{1}{100}C_{0}\delta.
 \end{equation*}
 Then we have \begin{equation*}
     \sup_{\mathcal{Q}_{\epsilon+\widetilde{\epsilon}}^{(in)}\backslash\mathcal{Q}_{\epsilon}^{(in)}}\left\vert\phi_{1,p}\right\vert\lesssim \frac{1}{100}C_{0}\delta\left(1+\frac{\epsilon}{u}\right),\quad \sup_{\mathcal{Q}_{\epsilon+\widetilde{\epsilon}}^{(in)}\backslash\mathcal{Q}_{\epsilon}^{(in)}}\left\vert\phi_{2,p}\right\vert\lesssim \frac{1}{100}C_{0}\delta\left(1+\frac{\epsilon}{u}\right).
 \end{equation*}
 Hence, we have \begin{align*}
     \mathcal{R}^{\epsilon}_{\partial_{u}\phi_{p}}(\mathcal{Q}_{\epsilon+\widetilde{\epsilon}}^{(in)}\backslash\mathcal{Q}_{\epsilon}^{(in)})\lesssim& \sup_{\mathcal{Q}_{\epsilon+\widetilde{\epsilon}}^{(in)}\backslash\mathcal{Q}_{\epsilon}^{(in)}}\left\vert\partial_{u}(r_{\epsilon}\phi_{1,p})\right\vert+\sup_{\mathcal{Q}_{\epsilon+\widetilde{\epsilon}}^{(in)}\backslash\mathcal{Q}_{\epsilon}^{(in)}}\left\vert\partial_{u}(r_{\epsilon}\phi_{2,p})\right\vert\\&+\sup_{\mathcal{Q}_{\epsilon+\widetilde{\epsilon}}^{(in)}\backslash\mathcal{Q}_{\epsilon}^{(in)}}\left\vert\phi_{1,p}\right\vert+\sup_{\mathcal{Q}_{\epsilon+\widetilde{\epsilon}}^{(in)}\backslash\mathcal{Q}_{\epsilon}^{(in)}}\left\vert\phi_{2,p}\right\vert\\\lesssim& \frac{1}{10}C_{0}\delta.
 \end{align*}
 Similarly, for $\mathcal{R}_{\partial_{v}\phi_{p}}^{\epsilon}$, we have that \begin{equation*}
     \mathcal{R}_{\partial_{v}\phi_{p}}^{\epsilon}(\mathcal{Q}_{\epsilon+\widetilde{\epsilon}}^{(in)}\backslash\mathcal{Q}_{\epsilon}^{(in)})\leq\frac{1}{10}C_{0}\delta.
 \end{equation*}
 Hence, putting everything together, we have \begin{equation*}
     \mathcal{R}^{\epsilon}(\mathcal{Q}_{\epsilon+\widetilde{\epsilon}}^{(in)}\backslash\mathcal{Q}_{\epsilon}^{(in)})\leq C_{0}\delta.
 \end{equation*}
 This concludes the proof.
\end{proof}
A direct consequence of the above improved bootstrap estimates is the existence of a solution arising from the initial data given in Section \ref{sec: setup of the perturbations} in the whole region $\mathcal{Q}^{(in)}\backslash\mathcal{Q}_{\epsilon}^{(in)}$.
\begin{corollary}
    For any given scalar field charge $q_{0}\neq0$, any $(k_{1},k_{2})$ with $0<k_{1}^{2}+k^{2}_{2}\leq k_{0}^{2}$ and any $\delta$ sufficiently small, there exists a family of solutions $(g_{\epsilon},\phi_{\epsilon},Q_{\epsilon})$ to the Einstein--Maxwell-charged scalar field equations in the whole region $\mathcal{Q}^{(in)}$, which coincide with the $(k_{1},k_{2})$-self-similar solution $(\widetilde{g}^{k_{1},k_{2}},\widetilde{\phi}^{k_{1},k_{2}},Q\equiv 0)$ constructed in Section \ref{sec: construction of uncharged solution} in the region $\mathcal{Q}_{\epsilon}^{(in)}$, and satisfy the initial condition imposed on $\mathcal{H}_{\epsilon}$ defined in Section \ref{sec: setup of the perturbations} in the region $\mathcal{Q}^{(in)}\backslash\mathcal{Q}_{\epsilon}^{(in)}$.
\end{corollary}

\subsubsection{Limiting argument}
In this section, we show that the solution $(g_{\epsilon},\phi_{\epsilon},Q_{\epsilon})$ constructed in Section \ref{sec: closing the bootstrap} in the whole region $\mathcal{Q}^{(in)}$ has a nontrivial limit solving the Einstein--Maxwell-charged scalar field equations when $\epsilon\rightarrow0$. First, we define the following norm \begin{equation*}
\left\Vert(r,\mu,\phi,Q)\right\Vert_{X(\Omega)}: = \left\Vert r\right\Vert_{C^{1}(\Omega)}+\left\Vert\mu\right\Vert_{C^{0}(\Omega)}+\left\Vert \phi\right\Vert_{C^{1}(\Omega)}+\left\Vert Q\right\Vert_{C^{0}(\Omega)},
\end{equation*}
and the function space $X(\Omega)$:\begin{equation*}
    X(\Omega):=\left\{(r,\mu,\phi,Q)|\left\Vert(r,\mu,\phi,Q)\right\Vert_{X(\Omega)}<\infty\right\}.
\end{equation*}

The goal in this section is to prove \begin{equation*}
\left\Vert(r_{\epsilon_{1}}-r_{\epsilon_{2}},\mu_{\epsilon_{1}}-\mu_{\epsilon_{2}},\phi_{\epsilon
_{1}}-\phi_{\epsilon_{2}},Q_{\epsilon_{1}}-Q_{\epsilon_{2}})\right\Vert_{X(\mathcal{Q}^{(in)})}\lesssim \left\vert\epsilon_{1}-\epsilon_{2}\right\vert.
\end{equation*}
Then by the standard Arzelà--Ascoli lemma, we can prove the existence of the limit $(r_{0},\mu_{0},\phi_{0}, Q_{0})$ solving the Einstein--Maxwell-charged scalar field equations.

Without loss of generality, we may assume that $0<\epsilon_{1}<\epsilon_{2}$. In the region $\mathcal{Q}_{2\epsilon_{1}}^{(in)}$, trivially we have \begin{equation*}
    \left\Vert(r_{\epsilon_{1}}-r_{\epsilon_{2}},\mu_{\epsilon_{1}}-\mu_{\epsilon_{2}},\phi_{\epsilon_{1}}-\phi_{\epsilon_{2}},Q_{\epsilon_{1}}-Q_{\epsilon_{2}})\right\Vert_{X(\mathcal{Q}_{2\epsilon_{1}}^{(in)})} = 0.
\end{equation*}
We first estimate the difference on the ingoing cone $\{v = 0\}$.
\begin{proposition}
\label{prop: difference estimate for Q}
On the ingoing cone $\{v = 0\}$, for the spacetime charge, we have \begin{align}
    \left\vert Q_{\epsilon_{1}}-Q_{\epsilon_{2}}\right\vert(u,0)\lesssim& \delta k\vert q_{0}\vert\epsilon_{1}^{2}(-\log\epsilon_{1})+\delta k\vert q_{0}\vert\epsilon_{2}^{2}(-\log\epsilon_{2}),\quad u\in[-1,-3\epsilon_{2}]. 
\end{align}
\end{proposition}

\begin{proof}
Considering the equation \eqref{eq:u-transport equation for Q} and taking the difference, we have \begin{equation}
    \partial_{u}(Q_{\epsilon_{1}}-Q_{\epsilon_{2}}) = q_{0}r^{2}(\phi_{1,\epsilon_{1}}\partial_{u}\phi_{2,\epsilon_{1}}-\phi_{2,\epsilon_{1}}\partial_{u}\phi_{1,\epsilon_{1}})-q_{0}r^{2}(\phi_{1,\epsilon_{2}}\partial_{u}\phi_{2,\epsilon_{2}}-\phi_{2,\epsilon_{2}}\partial_{u}\phi_{1,\epsilon_{2}}).
\end{equation}
By the choice of our cut-off functions, for $u\in[-1,-3\epsilon_{2}]$, we have that $Q_{\epsilon_{1}}-Q_{\epsilon_{2}}\equiv const$. For $u\in[-2\epsilon_{1},0)$, we have $Q_{\epsilon_{1}} = Q_{\epsilon_{2}}$.

Moreover, by \eqref{estimate on Qepsilon for u close to singularity}--\eqref{estimate for Qepsilon for u in the faraway region}, in the region $[-3\epsilon_{2},0]$, we have that \begin{equation*}
    \left\vert Q_{\epsilon_{1}}-Q_{\epsilon_{2}}\right\vert\lesssim \delta k \vert q_{0}\vert\epsilon_{1}^{2}(-\log\epsilon_{1})+\delta k \vert q_{0}\vert\epsilon_{2}^{2}(-\log\epsilon_{2}).
\end{equation*}
This concludes the proof.
\end{proof}
Next, we study the difference of $\mu_{\epsilon}$ and $\lambda_{\epsilon}$ on the ingoing cone.
\begin{proposition}
\label{prop: estimate on the difference of Hawking mass}
    On the ingoing cone $\{v = 0\}$, the difference of the mass ratios $\mu_{\epsilon_{1}}-\mu_{\epsilon_{2}}$ satisfies \begin{equation}
        \left\vert\mu_{\epsilon_{1}}-\mu_{\epsilon_{2}}\right\vert(u)\lesssim k\delta \epsilon_{2}^{\frac{1}{2}},\qquad u\leq -\epsilon_{2}^{\frac{1}{2}}.
    \end{equation}
\end{proposition}
\begin{proof}
Given the initial data on the ingoing cone $\{v = 0\}$, the equation for $\mu$ can be viewed as a linear ODE. Recall that in the proof of Proposition \ref{prop: initial estimate on mu}, we split this ODE into the following two parts: \begin{align*}
    &\partial_{u}\mu_{\epsilon}^{(1)}+\left(\frac{1}{u}+u\left\vert\partial_{u}\phi_{\epsilon}\right\vert^{2}\right)\mu_{\epsilon}^{(1)}=u\left\vert\partial_{u}\phi_{\epsilon}\right\vert^{2},\\&
    \partial_{u}\mu_{\epsilon}^{(2)}+\left(\frac{1}{u}+u\left\vert\partial_{u}\phi_{\epsilon}\right\vert^{2}\right)\mu_{\epsilon}^{(2)} = \frac{Q_{\epsilon}^{2}}{u^{3}}.
\end{align*}
Then we have \begin{equation*}
    \left\vert\mu_{\epsilon_{1}}-\mu_{\epsilon_{2}}\right\vert\leq\sum_{i = 1}^{2} \left\vert\mu_{\epsilon_{1}}^{(i)}-\mu_{\epsilon_{2}}^{(i)}\right\vert.
\end{equation*}
We first estimate $\mu_{\epsilon_{1}}^{(2)}-\mu_{\epsilon_{2}}^{(2)}$. We have \begin{equation*}
    \mu_{\epsilon}^{(2)} = (-u)^{-(1+k^{2})}\int_{0}^{u}\frac{Q_{\epsilon}^{2}(s)}{s^{3}}(-s)^{1+k^{2}}e^{\int_{u}^{s}\mathcal{P}_{\epsilon}(t)dt}ds,
\end{equation*}
where $\mathcal{P}_{\epsilon}$ is defined in~\eqref{eq: definition of Pepsilon}. Therefore, we have 
\begin{equation}
\label{eq: equation to estimate the difference of mu}\begin{aligned}
\left\vert\mu_{\epsilon_{1}}^{(2)}-\mu_{\epsilon_{2}}^{(2)}\right\vert\leq& (-u)^{-(1+k^{2})}\int_{u}^{0}(-s)^{-2+k^{2}}\left\vert Q_{\epsilon_{1}}^{2}e^{\int_{u}^{s}\mathcal{P}_{\epsilon
    _{1}}(t)dt}-Q_{\epsilon_{2}}^{2}e^{\int_{u}^{s}\mathcal{P}_{\epsilon_{2}}(t)dt}\right\vert ds\\\leq&(-u)^{-(1+k^{2})}\int_{u}^{-2\epsilon_{1}}(-s)^{-2+k^{2}}\left\vert Q_{\epsilon_{1}}^{2}-Q_{\epsilon_{2}}^{2}\right\vert e^{\int_{u}^{s}\mathcal{P}_{\epsilon_{1}}(t)dt}ds\\&+(-u)^{-(1+k^{2})}\int_{u}^{0}(-s)^{-2+k^{2}}Q_{\epsilon_{2}}^{2}\left\vert e^{\int_{u}^{s}\mathcal{P}_{\epsilon_{1}}(t)dt}-e^{\int_{u}^{s}\mathcal{P}_{\epsilon_{2}}(t)dt}\right\vert ds
\end{aligned}
\end{equation}
For the first term on the right-hand side of~\eqref{eq: equation to estimate the difference of mu}, we have \begin{align*}
    (-u)^{-(1+k^{2})}\int_{u}^{-2\epsilon_{1}}(-s)^{-2+k^{2}}\vert Q_{\epsilon_{1}}^{2}-Q_{\epsilon_{2}}^{2}\vert e^{\int_{u}^{s}\mathcal{P}_{\epsilon_{1}}(t)dt}ds&\lesssim q_{0}^{2}\delta^{2}k^{2}u^{2}(\log(-u))^{2},\quad -3\epsilon_{2}\leq u\leq-2\epsilon_{1},\\(-u)^{-(1+k^{2})}\int_{u}^{-2\epsilon_{1}}(-s)^{-2+k^{2}}\vert Q_{\epsilon_{1}}^{2}-Q_{\epsilon_{2}}^{2}\vert e^{\int_{u}^{s}\mathcal{P}_{\epsilon_{1}}(t)dt}ds&\lesssim q_{0}^{2}\delta^{2}k^{2}\epsilon_{2}^{2}(\log(-\epsilon_{2}))^{2},\quad -1\leq u\leq-3\epsilon_{2}.
\end{align*} 
Next, we estimate $\mu_{\epsilon}^{(1)}$. Recall that we have \begin{equation*}
   \mu_{\epsilon}^{(1)}-\frac{k^{2}}{1+k^{2}}= \frac{1}{1+k^{2}}(-u)^{-(1+k^{2})}\int_{0}^{u}(-s)^{1+k^{2}}\mathcal{P}_{\epsilon}(s) e^{\int_{u}^{s}\mathcal{P}_{\epsilon}(t)dt}ds.
\end{equation*}
For $u\in[-2\epsilon,0)$, we have that $\mu_{\epsilon}^{(1)}\equiv \frac{k^{2}}{1+k^{2}}$. For $u\in[-1,-\epsilon_{2}^{\frac{1}{2}}]$, we have \begin{equation}
\begin{aligned}
    \left\vert\mu_{\epsilon_{1}}^{(1)}-\mu_{\epsilon_{2}}^{(1)}\right\vert\lesssim& (-u)^{-(1+k^{2})}\int_{u}^{0}(-s)^{1+k^{2}}\left\vert\mathcal{P}_{\epsilon_{1}}(s)e^{\int_{u}^{s}\mathcal{P}_{\epsilon_{1}}(t)dt}-\mathcal{P}_{\epsilon_{2}}(s)e^{\int_{u}^{s}\mathcal{P}_{\epsilon_{2}}(t)dt}\right\vert ds\\\lesssim&
    (-u)^{-(1+k^{2})}\int_{-3\epsilon_{2}}^{0}(-s)^{1+k^{2}}\left\vert\mathcal{P}_{\epsilon_{1}}(s)e^{\int_{u}^{s}\mathcal{P}_{\epsilon_{1}}(t)dt}-\mathcal{P}_{\epsilon_{2}}(s)e^{\int_{u}^{s}\mathcal{P}_{\epsilon_{2}}(t)dt}\right\vert ds \\\lesssim& k\delta (-u)^{-(1+k^{2})}\epsilon_{2}^{1+k^{2}}\lesssim k\delta \epsilon_{2}^{\frac{1+k^{2}}{2}}.
    \end{aligned}
    \label{eq: estimate on mu1}
\end{equation} 
Putting everything together, we have \begin{equation*}
    \left\vert\mu_{\epsilon_{1}}-\mu_{\epsilon_{2}}\right\vert(u)\lesssim k\delta \epsilon_{2}^{\frac{1}{2}},\qquad u\leq -\epsilon_{2}^{\frac{1}{2}}.
\end{equation*}
\end{proof}
\begin{proposition}
\label{prop: difference estimate for Omega and lambda}
On the ingoing cone $\{v = 0\}$, we have the following estimates on $\Omega_{\epsilon_{1}}^{2}-\Omega_{\epsilon_{2}}^{2}$ and $\lambda_{\epsilon_{1}}-\lambda_{\epsilon_{2}}$:  \begin{align}
    &\left\vert\Omega_{\epsilon_{1}}^{2}-\Omega_{\epsilon_{2}}^{2}\right\vert(u,0)\lesssim k\delta(-u)^{k^{2}}\epsilon_{2}^{\gamma},\quad u\in[-1,-3\epsilon_{2}],\\&
    \left\vert\lambda_{\epsilon_{1}}-\lambda_{\epsilon_{2}}\right\vert(u,0)\lesssim k\delta\epsilon_{2}^{\min\{\gamma,\frac{1}{2}\}}(-u)^{k^{2}},\quad u\in[-1,-\epsilon_{2}^{\frac{1}{2}}].
\end{align}
\end{proposition}
\begin{proof}
Recall that on the ingoing cone $\{v = 0\}$, we have 
   \begin{equation*}
        \Omega_{\epsilon}^{2}(u,0) = \widetilde{\Omega}^{2}(u,0)e^{\int_{0}^{u}\mathcal{P}_{\epsilon}(s)ds}.
    \end{equation*}
    Then for $u\leq-3\epsilon_{2}$, we have \begin{equation*}
        \left\vert\Omega_{\epsilon_{1}}^{2}-\Omega_{\epsilon_{2}}^{2}\right\vert\lesssim \widetilde{\Omega}^{2}(u,0)e ^{L_{1}}\left\vert e^{\int_{0}^{-3\epsilon_{2}}\mathcal{P}_{1,\epsilon_{1}}(s)ds}-e^{\int_{0}^{-3\epsilon_{2}}\mathcal{P}_{1,\epsilon_{2}}(s)ds}\right\vert\lesssim (-u)^{k^{2}}k\delta \epsilon_{2}^{\gamma},\quad u\in[-1,-3\epsilon_{2}].
    \end{equation*}
    Since $1-\mu = \frac{4\lambda\nu}{\Omega^{2}}$, we have \begin{equation*}
        \left\vert\lambda_{\epsilon_{1}}-\lambda_{\epsilon_{2}}\right\vert(u,0)\lesssim k\delta \epsilon_{2}^{\min\{\gamma,\frac{1}{2}\}}(-u)^{k^{2}},\quad u\in[-1,-\epsilon_{2}^{\frac{1}{2}}].
    \end{equation*}
\end{proof}
For $\partial_{u}\phi_{\epsilon}$ on the ingoing cone, trivially we have \begin{equation*}
    \left\vert\partial_{u}\phi_{1,\epsilon_{1}}-\partial_{u}\phi_{1,\epsilon_{2}}\right\vert = 0,\quad \left\vert\partial_{u}\phi_{2,\epsilon_{1}}-\partial_{u}\phi_{2,\epsilon_{2}}\right\vert = 0,\quad u\leq -\epsilon_{2}^{\frac{1}{2}}.
\end{equation*}
Next, we consider the difference of $\partial_{v}\phi_{\epsilon}$.
\begin{proposition}
\label{prop: difference estimate for dvphi}
    On the ingoing cone $\{v = 0\}$, we have the following estimate on $\partial_{v}\phi_{\epsilon_{1}}-\partial_{v}\phi_{\epsilon_{2}}$: \begin{align}
       & (-u)^{q_{k}}\left\vert\partial_{v}\phi_{1,\epsilon_{1}}-\partial_{v}\phi_{1,\epsilon
        _{2}}\right\vert\lesssim \delta \epsilon_{2}^{\min\{\gamma,\frac{1}{2}\}},\quad u\in[-1,-\epsilon_{2}^{\frac{1}{2}}],\label{eq: estimate on the difference of dvphi on the ingoing cone}\\&
        (-u)^{q_{k}}\left\vert\partial_{v}\phi_{2,\epsilon_{1}}-\partial_{v}\phi_{2,\epsilon_{2}}\right\vert\lesssim \delta \epsilon_{2}^{\min\{\gamma,\frac{1}{2}\}},\quad u\in[-1,-\epsilon_{2}^{\frac{1}{2}}].\label{eq: estimate on the difference of dvphi2 on the ingoing cone}
    \end{align}
\end{proposition}
\begin{proof}
    The equation for $\partial_{v}\phi_{1}$ and $\partial_{v}\phi_{2}$ on the ingoing cone $\{v = 0\}$ can be reduced to\begin{align*}
        &\partial_{u}\left((-u)\partial_{v}\phi_{1,\epsilon_{1}}-(-u)\partial_{v}\phi_{1,\epsilon_{2}}\right) = -\left(\frac{Q_{\epsilon_{1}}\Omega_{\epsilon_{1}}^{2}}{4r_{\epsilon_{1}}}\phi_{2,\epsilon_{1}}-\frac{Q_{\epsilon_{2}}\Omega_{\epsilon_{2}}^{2}}{4r_{\epsilon_{2}}}\phi_{2,\epsilon_{2}}\right)-\left(\lambda_{\epsilon_{1}}\partial_{u}\phi_{1,\epsilon_{1}}-\lambda_{\epsilon_{2}}\partial_{u}\phi_{1,\epsilon_{2}}\right),\\&\partial_{u}\left((-u)\partial_{v}\phi_{2,\epsilon_{1}}-(-u)\partial_{v}\phi_{2,\epsilon_{2}}\right) = \left(\frac{Q_{\epsilon_{1}}\Omega_{\epsilon_{1}}^{2}}{4r_{\epsilon_{1}}}\phi_{1,\epsilon_{1}}-\frac{Q_{\epsilon_{2}}\Omega_{\epsilon_{2}}^{2}}{4r_{\epsilon_{2}}}\phi_{1,\epsilon_{2}}\right)-\left(\lambda_{\epsilon_{1}}\partial_{u}\phi_{2,\epsilon_{1}}-\lambda_{\epsilon_{2}}\partial_{u}\phi_{2,\epsilon_{2}}\right).
    \end{align*}
Directly integrating the equation for $\partial_{v}\phi_{1,\epsilon_{1}}-\partial_{v}\phi_{1,\epsilon_{2}}$, for $u\leq -\epsilon_{2}^{\frac{1}{2}}$, we have \begin{equation}
\label{eq: step one in estimating the difference of dvphi}
\begin{aligned}
    (-u)\left\vert \partial_{v}\phi_{1,\epsilon_{1}}-\partial_{v}\phi_{1,\epsilon_{2}}\right\vert(u,0)\leq& \int_{u}^{0}\left\vert\frac{Q_{\epsilon_{1}}\Omega_{\epsilon_{1}}^{2}}{4r_{\epsilon_{1}}}\phi_{2,\epsilon_{1}}-\frac{Q_{\epsilon_{2}}\Omega_{\epsilon_{2}}^{2}}{4r_{\epsilon_{2}}}\phi_{2,\epsilon_{2}}\right\vert+\left\vert\lambda_{\epsilon_{1}}\partial_{u}\phi_{1,\epsilon_{1}}-\lambda_{\epsilon_{2}}\partial_{u}\phi_{1,\epsilon_{2}}\right\vert ds\\\lesssim &\int_{-\epsilon_{2}^{\frac{1}{2}}}^{0}\left\vert\frac{Q_{\epsilon_{1}}\Omega_{\epsilon_{1}}^{2}}{4r_{\epsilon_{1}}}\phi_{2,\epsilon_{1}}-\frac{Q_{\epsilon_{2}}\Omega_{\epsilon_{2}}^{2}}{4r_{\epsilon_{2}}}\phi_{2,\epsilon_{2}}\right\vert+\left\vert\lambda_{\epsilon_{1}}\partial_{u}\phi_{1,\epsilon_{1}}-\lambda_{\epsilon_{2}}\partial_{u}\phi_{1,\epsilon_{2}}\right\vert ds\\&+\int_{u}^{-\epsilon_{2}^{\frac{1}{2}}}\left\vert\frac{Q_{\epsilon_{1}}\Omega_{\epsilon_{1}}^{2}}{4r_{\epsilon_{1}}}\phi_{2,\epsilon_{1}}-\frac{Q_{\epsilon_{2}}\Omega_{\epsilon_{2}}^{2}}{4r_{\epsilon_{2}}}\phi_{2,\epsilon_{2}}\right\vert+\left\vert\lambda_{\epsilon_{1}}\partial_{u}\phi_{1,\epsilon_{1}}-\lambda_{\epsilon_{2}}\partial_{u}\phi_{1,\epsilon_{2}}\right\vert ds
    \end{aligned}
\end{equation}
For the first term on the right-hand side of \eqref{eq: step one in estimating the difference of dvphi}, we have \begin{align*}
    \text{First term on the RHS of \eqref{eq: step one in estimating the difference of dvphi}}\lesssim \delta \epsilon_{2}^{\frac{k^{2}}{2}}.
\end{align*}
For the second term on the right-hand side of \eqref{eq: step one in estimating the difference of dvphi}, we have \begin{equation*}
    \text{Second term on the RHS of \eqref{eq: step one in estimating the difference of dvphi}}\lesssim \delta \epsilon_{2}^{\min\{\gamma,\frac{1}{2}\}}\left((-u)^{k^{2}}-\epsilon_{2}^{\frac{k^{2}}{2}}\right).
\end{equation*}
Hence, we have \begin{equation*}
    (-u)^{q_{k}}\left\vert\partial_{v}\phi_{1,\epsilon_{1}}-\partial_{v}\phi_{1,\epsilon_{2}}\right\vert\lesssim \delta \epsilon_{2}^{\min\{\gamma,\frac{1}{2}\}}.
\end{equation*}
This concludes the proof of \eqref{eq: estimate on the difference of dvphi on the ingoing cone}. The estimate~\eqref{eq: estimate on the difference of dvphi2 on the ingoing cone} follows similarly.
\end{proof}
For $\zeta>0$, let $$\Theta_{\zeta,\epsilon}: = \{-1\leq u\leq-\epsilon,\ -\min\{(-u)^{q_{k}},\zeta\}\leq v\leq0\}.$$
Since we have established the estimates for $\mu_{\epsilon_{1}}-\mu_{\epsilon_{2}}$, $\Omega_{\epsilon_{1}}^{2}-\Omega_{\epsilon_{2}}^{2}$, $\lambda_{\epsilon_{1}}-\lambda_{\epsilon_{2}}$, $\partial_{u}(\phi_{\epsilon_{1}}-\phi_{\epsilon_{2}})$, $\partial_{v}(\phi_{\epsilon_{1}}-\phi_{\epsilon_{2}})$, and $Q_{\epsilon_{1}}-Q_{\epsilon_{2}}$ in Proposition \ref{prop: difference estimate for Q}, Proposition \ref{prop: estimate on the difference of Hawking mass}, Proposition \ref{prop: difference estimate for Omega and lambda}, and Proposition \ref{prop: difference estimate for dvphi}, we assume that for some $\zeta>0$, there exists a universal constant $B$ such that \begin{equation}
    \mathcal{R}^{\epsilon_{1},\epsilon_{2}}\left(\Theta_{\zeta,\epsilon^{\frac{1}{2}}_{2}}\right)\leq 2B\delta\epsilon_{2}^{\min\{\gamma,\frac{1}{2}\}}.\label{eq: bootstrap assumption for difference}
\end{equation} 
We will prove that this bound can be improved to \begin{equation*}
    \mathcal{R}^{\epsilon_{1},\epsilon_{2}}\left(\Theta_{\zeta,\epsilon^{\frac{1}{2}}_{2}}\right)\leq B\delta\epsilon_{2}^{\min\{\gamma,\frac{1}{2}\}},
\end{equation*}  
thus establishing the global $\epsilon_{2}$-smallness of $\mathcal{R}^{\epsilon_{1},\epsilon_{2}}$ in the region $\mathcal{Q}\backslash\mathcal{Q}^{\epsilon_{2}^{\frac{1}{2}}}$.

For any $\Psi$ that is a combination of the quantities in $\{r,\nu,\lambda,\mu,\phi,\partial_{u}\phi,\partial_{v}\phi,Q,A\}$, we define $$\Psi_{\epsilon_{1},\epsilon_{2}}: = \Psi_{\epsilon_{1}}-\Psi_{\epsilon_{2}}.$$

Analogously to the derivation of~\eqref{eq: difference equation for lambdap in bootstrap argument}-\eqref{eq: difference equation for phi2 in the bootstrap argument}, we obtain the difference equations between  $(g_{\epsilon_{1}},\phi_{\epsilon_{1}},Q_{\epsilon_{1}})$ and $(g_{\epsilon_{2}},\phi_{\epsilon_{2}},Q_{\epsilon_{2}})$. \begin{align}
    \partial_{u}\lambda_{\epsilon_{1},\epsilon_{2}} =& \left(\frac{\mu}{1-\mu}\frac{\lambda\nu}{r}\right)_{\epsilon_{1},\epsilon_{2}}-\left(\frac{\lambda\nu}{(1-\mu)r}\frac{Q^{2}}{r^{2}}\right)_{\epsilon_{1},\epsilon_{2}},\allowdisplaybreaks\label{eq: epsilon1 and epsilon2 difference equation first}\allowdisplaybreaks\\
    \partial_{v}\nu_{\epsilon_{1},\epsilon_{2}} = &\left(\frac{\mu}{1-\mu}\frac{\lambda\nu}{r}\right)_{\epsilon_{1},\epsilon_{2}}-\left(\frac{\lambda\nu}{(1-\mu)r}\frac{Q^{2}}{r^{2}}\right)_{\epsilon_{1},\epsilon_{2}},\allowdisplaybreaks\\
    \partial_{u}\partial_{v}\left(r_{\epsilon_{1}}(\phi_{1})_{\epsilon_{1},\epsilon_{2}}\right) =& -\nu_{\epsilon_{1},\epsilon_{2}}\partial_{v}\phi_{1,\epsilon_{2}}-\lambda_{\epsilon_{1},\epsilon_{2}}\partial_{u}\phi_{1,\epsilon_{2}}-(\phi_{1})_{\epsilon_{1},\epsilon_{2}}\partial_{u}\partial_{v}r_{\epsilon_{1}}\nonumber\\&+q_{0}\left(A_{u}\partial_{v}(r\phi_{2})+\frac{Q}{r}\frac{\lambda\nu}{1-\mu}\phi_{2}\right)_{\epsilon_{1},\epsilon_{2}},\allowdisplaybreaks\\
    \partial_{u}\partial_{v}\left(r_{\epsilon_{1}}(\phi_{2})_{\epsilon_{1},\epsilon_{2}}\right) =& -\nu_{\epsilon_{1},\epsilon_{2}}\partial_{v}\phi_{2,\epsilon_{2}}-\lambda_{\epsilon_{1},\epsilon_{2}}\partial_{u}\phi_{2,\epsilon_{2}}-(\phi_{2})_{\epsilon_{1},\epsilon_{2}}\partial_{u}\partial_{v}r_{\epsilon_{1}}\nonumber\\&-q_{0}\left(A_{u}\partial_{v}(r\phi_{1})+\frac{Q}{r}\frac{\lambda\nu}{1-\mu}\phi_{1}\right)_{\epsilon_{1},\epsilon_{2}},\allowdisplaybreaks\\
\partial_{v}\mu_{\epsilon_{1},\epsilon_{2}}+\left(\frac{\lambda_{\epsilon_{1}}}{r_{\epsilon_{1}}}+\frac{r_{\epsilon_{1}}}{\lambda_{\epsilon_{1}}}\left\vert \partial_{v}\phi_{\epsilon_{1}}\right\vert^{2} \right)\mu_{\epsilon_{1},\epsilon_{2}}=&-\left(\frac{\lambda}{r}+\frac{r}{\lambda}\left\vert \partial_{v}\phi\right\vert^{2}\right)_{\epsilon_{1},\epsilon_{2}}\mu_{\epsilon_{2}}+\left(\frac{r}{\lambda}\left\vert \partial_{v}\phi\right\vert^{2}\right)_{\epsilon_{1},\epsilon_{2}}\nonumber\\&+\left(\frac{\lambda Q^{2}}{r^{3}}\right)_{\epsilon_{1},\epsilon_{2}},\allowdisplaybreaks\label{eq: dvmu epsilon1 epsilon2 difference equation}\\\partial_{u}\mu_{\epsilon_{1},\epsilon_{2}}+\left(\frac{\nu_{\epsilon_{1}}}{r_{\epsilon_{1}}}+\frac{r_{\epsilon_{1}}}{\nu_{\epsilon_{1}}}\left\vert D_{u}\phi_{\epsilon_{1}}\right\vert^{2} \right)\mu_{\epsilon_{1},\epsilon_{2}}=&-\left(\frac{\nu}{r}+\frac{r}{\nu}\left\vert D_{u}\phi\right\vert^{2}\right)_{\epsilon_{1},\epsilon_{2}}\mu_{\epsilon_{2}}+\left(\frac{r}{\nu}\left\vert D_{u}\phi\right\vert^{2}\right)_{\epsilon_{1},\epsilon_{2}}\nonumber\\&+\left(\frac{\nu Q^{2}}{r^{3}}\right)_{\epsilon_{1},\epsilon_{2}},\allowdisplaybreaks\\
\partial_{v}Q_{\epsilon_{1},\epsilon_{2}} = &q_{0}\left(r^{2}(\phi_{2}\partial_{v}\phi_{1}-\phi_{1}\partial_{v}\phi_{2})\right)_{\epsilon_{1},\epsilon_{2}},\allowdisplaybreaks\\
\partial_{v}(A_{u})_{\epsilon_{1},\epsilon_{2}}=&2\left(\frac{Q\lambda\nu}{(1-\mu)r^{2}}\right)_{\epsilon_{1},\epsilon_{2}}.
\label{eq: epsilon 1 and epsilon 2 equation last}
\end{align}
 The following proposition is a direct consequence of the bootstrap assumption \eqref{eq: bootstrap assumption for difference}.
\begin{proposition}
    Under the bootstrap assumption \eqref{eq: bootstrap assumption for difference}, for $\delta$ and $\epsilon_{2}$ sufficiently small, in the region $\Theta_{\zeta,\epsilon_{2}^{\frac{1}{2}}}$, we have \begin{align*}
        &\left\vert r_{\epsilon_{1},\epsilon_{2}}\right\vert(u,v)\lesssim B\delta \epsilon_{2}^{\min\{\gamma,\frac{1}{2}\}}(-u),\quad \left\vert(\phi_{1})_{\epsilon_{1},\epsilon_{2}}\right\vert(u,v)\lesssim B\delta\epsilon_{2}^{\min\{\gamma,\frac{1}{2}\}},\quad \left\vert(\phi_{2})_{\epsilon_{1},\epsilon_{2}}\right\vert\lesssim B\delta\epsilon_{2}^{\min\{\gamma,\frac{1}{2}\}},\\&
        \left\vert\mu_{\epsilon_{1},\epsilon_{2}}\right\vert\lesssim Bk\delta \epsilon_{2}^{\min\{\gamma,\frac{1}{2}\}},\quad \left\vert (A_{u})_{\epsilon_{1},\epsilon_{2}}\right\vert\lesssim B\delta k\epsilon_{2}^{\min\{\gamma,\frac{1}{2}\}}. 
    \end{align*}
\end{proposition}
\begin{proof}
    For $r_{\epsilon_{1},\epsilon_{2}}$, we have \begin{equation*}
        \left\vert r_{\epsilon_{1},\epsilon_{2}}\right\vert(u,v)\lesssim \int_{v}^{0}\left\vert\lambda_{\epsilon_{1},\epsilon_{2}}\right\vert(u,s)ds\lesssim \delta \epsilon_{2}^{\min\{\gamma,\frac{1}{2}\}}(-u).
    \end{equation*}
    For $(\phi_{1})_{\epsilon_{1},\epsilon_{2}}$, we have \begin{equation*}
\left\vert(\phi_{1})_{\epsilon_{1},\epsilon_{2}}\right\vert(u,v)\lesssim\int_{v}^{0}\left\vert\partial_{v}\phi_{\epsilon_{1},\epsilon_{2}}\right\vert(u,s)ds\lesssim B\delta \epsilon_{2}^{\min\{\gamma,\frac{1}{2}\}}\int_{v}^{0}(-u)^{-q_{k}}ds\lesssim B\delta \epsilon_{2}^{\min\{\gamma,\frac{1}{2}\}}.
    \end{equation*}
    We can derive the estimate for $(\phi_{2})_{\epsilon_{1},\epsilon_{2}}$ similarly. To show the estimate for $\mu_{\epsilon_{1},\epsilon_{2}}$, using \eqref{eq: dvmu epsilon1 epsilon2 difference equation}, we have \begin{equation}
        \partial_{v}\left(r_{\epsilon_{1}}e^{\int_{0}^{v}\frac{r_{\epsilon_{1}}}{\lambda_{\epsilon_{1}}}\left\vert\partial_{v}\phi_{\epsilon_{1}}\right\vert^{2}}\mu_{\epsilon_{1},\epsilon_{2}}\right) = r_{\epsilon_{1}}e^{\int_{0}^{v}\frac{r_{\epsilon_{1}}}{\lambda_{\epsilon_{1}}}\left\vert\partial_{v}\phi_{\epsilon_{1}}\right\vert^{2}}(\text{RHS of \eqref{eq: dvmu epsilon1 epsilon2 difference equation}}).
    \end{equation}
    Directly integrating the above equation and estimating the right-hand side similar as in the estimate for $\mu_{p}$ in Proposition~\ref{prop: consequence of the bootstrap in the existence argument}, we have \begin{equation*}
        \left\vert\mu_{\epsilon_{1},\epsilon_{2}}\right\vert\lesssim Bk\delta \epsilon_{2}^{\min\{\gamma,\frac{1}{2}\}}.
    \end{equation*}
    Integrating the equation \eqref{eq: epsilon 1 and epsilon 2 equation last} and using Proposition \ref{prop: difference estimate for Q}, we have \begin{equation*}
        \left\vert(A_{u})_{\epsilon_{1},\epsilon_{2}}\right\vert\lesssim B\delta k \epsilon_{2}^{2}\left\vert\log\epsilon_{2}\right\vert\int_{v}^{0} (-u)^{-2+k^{2}}ds\lesssim B\delta k\epsilon_{2}^{\min\{\gamma,\frac{1}{2}\}}.
    \end{equation*}
    This concludes the proof.
\end{proof}
Now we are ready to close the bootstrap argument. We have the following proposition.
\begin{proposition}
\label{prop: closing the bootstrap for epsilon}
    Assuming \eqref{eq: bootstrap assumption for difference} and $0<\epsilon_{1}<\epsilon_{2}$, for any given scalar field charge $q_{0}$, there exists $k_{0}$ sufficiently small, depending on $q_{0}$ and the universal constant $B$, such that $\mathcal{R}^{\epsilon_{1},\epsilon_{2}}\left(\Theta_{\zeta,\epsilon_{2}^{\frac{1}{2}}}\right)\leq B\delta \epsilon_{2}^{\min\{\gamma,\frac{1}{2}\}}$ for any $0<k<k_{0}$. Moreover, we have \begin{equation*}
        \mathcal{R}^{\epsilon_{1},\epsilon_{2}}(\mathcal{Q}^{(in)}\backslash\mathcal{Q}^{(in)}_{\sqrt{\epsilon_{2}}})\lesssim k\delta \epsilon_{2}^{\min\{\gamma,\frac{1}{2}\}}.
    \end{equation*}
\end{proposition}
\begin{proof}
    The proof is similar to that of Proposition~\ref{prop: closing the bootstrap}. 
\end{proof}
\begin{corollary}
    Let $N$ be a positive integer. Then $(g_{\epsilon},\phi_{\epsilon},Q_{\epsilon})$ is a Cauchy sequence in $X(\mathcal{Q}^{(in)}\backslash \mathcal{Q}_{\frac{1}{N}}^{(in)})$ for $0<\epsilon<\frac{1}{N^{2}}$. Moreover, there exists $(g,\phi,Q)$ independent of $N$ such that $(g_{\epsilon},\phi_{\epsilon},Q_{\epsilon})\rightarrow (g,\phi,Q)$ in $X\left(\mathcal{Q}^{(in)}\backslash\mathcal{Q}^{(in)}_{\frac{1}{N}}\right)$   when $\epsilon\rightarrow 0 $ for any $N>0$. The limit $(g,\phi,Q)$ is a solution to the Einstein--Maxwell-charged scalar field equations.
\end{corollary}
\begin{proof}
    According to Proposition \ref{prop: closing the bootstrap for epsilon}, we have that $$(r_{\epsilon},\lambda_{\epsilon},\nu_{\epsilon},\mu_{\epsilon},\phi_{\epsilon},\partial_{u}\phi_{\epsilon},\partial_{v}\phi_{\epsilon},Q_{\epsilon},A_{u,\epsilon})$$ is a Cauchy sequence in $X\left(\mathcal{Q}^{(in)}\backslash\mathcal{Q}_{\frac{1}{N}}^{(in)}\right)$ for any $N\geq1$. Using the equations \eqref{eq:spherical symmetric equaion1}-\eqref{eq:spherical symmetric equation last}, we have $$(\partial_{u}\lambda_{\epsilon},\partial_{v}\nu_{\epsilon},\partial_{u}\mu_{\epsilon},\partial_{v}\mu_{\epsilon},\partial_{u}\partial_{v}\phi_{\epsilon},\partial_{u}Q_{\epsilon},\partial_{v}Q_{\epsilon},\partial_{u}A_{u,\epsilon})$$ is also a Cauchy sequence in $X$. Hence, we have \begin{align}
       & r_{\epsilon}\rightarrow r,\quad \lambda_{\epsilon}\rightarrow\lambda = \partial_{v}r,\quad \nu_{\epsilon}\rightarrow\nu = \partial_{u}r,\quad \mu_{\epsilon}\rightarrow \mu,\\&\phi_{\epsilon}\rightarrow \phi,\quad \partial_{u}\phi_{\epsilon
        }\rightarrow \partial_{u}\phi,\quad \partial_{v} \phi_{\epsilon}\rightarrow \partial_{v}\phi,\\&
        Q_{\epsilon}\rightarrow Q,\quad A_{u,\epsilon}\rightarrow A_{u}.
    \end{align}
    Moreover, we have that $(r,\lambda,\nu,\mu,\phi,Q,A_{u})$ solves the Einstein--Maxwell-charged scalar field equations.
\end{proof}
To conclude this section, we prove the following proposition giving quantitative estimates for the limit $(g,\phi,Q)$.
\begin{proposition}
    The following estimates hold for $(g,\phi,Q)$: \begin{align}
        &r\approx (-u)(z+1),\quad \lambda\approx (-u)^{k^{2}},\quad -\nu\approx 1,\quad \mu\approx k^{2},\\&
       \left\vert \partial_{u}\phi\right\vert\lesssim \frac{1}{(-u)},\quad\left\vert\partial_{v}\phi\right\vert\lesssim \frac{1}{(-u)^{q_{k}}},\quad \left\vert Q\right\vert\lesssim u^{2}.
    \end{align}
    Moreover, on the ingoing cone $\{v = 0\}$, we have \begin{equation}
        \left\vert Q\right\vert\approx u^{2}.
    \end{equation}
\end{proposition}
\begin{proof}
    The proof follows directly from the limiting process.
\end{proof}
\section{Overview of the exterior construction}
\label{sec: overview of the exterior construction}
\subsection{Decomposition into four regions}
For the construction of the exterior region, we divide it into four regions: 
\begin{figure}[htbp]
    \centering
\begin{tikzpicture}[ scale=0.8, >=Latex, every node/.style={font=\small} ]
\coordinate (O) at (0,0);
\coordinate (A) at (4,-4);
\coordinate (B) at (6,6);
\coordinate (C) at (10,2);
\draw (O)--(A);
\draw[dashed] (O)--(B);
\draw (A)--(C);
\draw[dashed] (B) -- (C)
    node[midway, above, sloped]
    {\rotatebox{90}{$\mathscr{I}$}};
\node[left] at (O) { $\mathcal{O}$};
\coordinate (A1) at (6,-2);
\coordinate (A2) at (8,0);
\draw[dashed] (O)--(A1);
\draw[dashed] (O)--(A2);
\coordinate (A3) at (3,0);
\coordinate (A4) at (1.5,1.5);
\draw[dashed] (A3)--(A4);
\coordinate (L1) at (2.5,-1.5);
\node at (L1){I};
\coordinate (L2) at (3.5,-0.5);
\node at (L2) {II};
\coordinate (L3) at (1.5,0.5);
\node at (L3) {III};
\coordinate (L4) at (5,3);
\node at (L4){IV};
\end{tikzpicture}
    \caption{Decomposition of regions}
    \label{fig:decomposition of regions}
\end{figure}
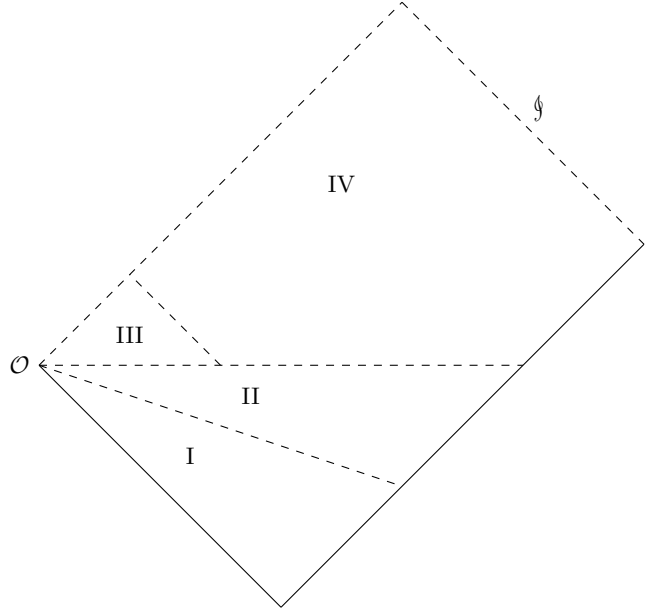

\subsection{A brief review of the construction of Christodoulou's naked singularity exterior in the region I \cite{christodoulou1994examples}}
\label{sub:a_quick_propagational_construction_of_christodoulou_s_naked_singularity_exterior_in_the_region_rmnum_1}
In this section, we revisit the construction of naked singularities for the Einstein-scalar field equations. In the seminal work \cite{christodoulou1994examples}, assuming spherical symmetry and exact $k$-self-similarity, Christodoulou reduced the Einstein-scalar field equations to a $2\times 2$ autonomous system of ODEs. Then, a careful phase portrait analysis in the exterior region shows the existence of naked singularities for $0<k^{2}<\frac{1}{3}$.

Although Christodoulou's construction provides the first mathematical insight into the existence of naked singularities, his approach is not immediately generalizable to Einstein equations coupled with other models. In particular, for the Einstein vacuum equations, naked singularities cannot exist in the regime of spherical symmetry due to the famous Birkhoff theorem. In the celebrated work of Rodnianski and Shlapentokh--Rothman \cite{rodnianski2023naked}, the authors developed a robust hyperbolic approach to construct the exterior region of a naked singularity spacetime. Now we briefly review this approach in the setting of spherically symmetric Einstein-scalar field equations in the region $\Rmnum{1}$, which is the most difficult region in the whole construction.

Under spherical symmetry, the Einstein-scalar field equations reduce to~\eqref{eq: uncharged u Ray equation}-\eqref{eq: uncharged dvm equation} with a real scalar field $\phi = \phi_{1},\ \phi_{2} = 0$ \begin{align}
    \partial_{u}\left(\frac{\partial_{u}r}{\Omega^{2}}\right) =& -\frac{r}{\Omega^{2}}(\partial_{u}\phi)^{2},\allowdisplaybreaks\\
    \partial_{v}\left(\frac{\partial_{v}r}{\Omega^{2}}\right) = &-\frac{r}{\Omega^{2}}(\partial_{v}\phi)^{2},\allowdisplaybreaks\\
    r\partial_{u}\partial_{v}r+\partial_{u}r\partial_{v}r=&-\frac{1}{4}\Omega^{2},\label{eq: wave equation for the ESF}\allowdisplaybreaks\\
    r\partial_{u}\partial_{v}\phi+\partial_{u}r\partial_{v}\phi+\partial_{v}r\partial_{u}\phi =& 0,\allowdisplaybreaks\\
    \partial_{u}m+\frac{r}{\partial_{u}r}(\partial_{u}\phi)^{2}m=&\frac{1}{2}\frac{r^{2}}{\partial_{u}r}(\partial_{u}\phi)^{2},\allowdisplaybreaks\\
    \partial_{v}m+\frac{r}{\partial_{v}r}(\partial_{v}\phi)^{2}m=&\frac{1}{2}\frac{r^{2}}{\partial_{v}r}(\partial_{v}\phi)^{2}.
\end{align}
We prescribe suitably-designed initial data on the characteristic hypersurface $\Sigma_{-1}\cup\underline{\Sigma}_{0}$, where $\underline{\Sigma}_{0}$ is the past light cone emanating from the singularity $\mathcal{O}$. We then establish global existence in the whole exterior region $\{-1\leq u<0,\ v\geq 0\}$, which is free of trapped surfaces.

On the ingoing cone $\underline{\Sigma}_{0}$, we prescribe the following initial data and gauge choice, which coincide with those in \cite{christodoulou1994examples} up to scaling.\begin{equation*}
r(u,0) = -u,\quad \partial_{u}\phi(u,0) = \frac{k}{(-u)},\quad \Omega^{2}(-1,0) = \frac{4}{1+k^{2}},
\end{equation*}
where $k$ is the parameter for the twisted self-similarity.

As a consequence of the above initial data, using the wave equations \eqref{eq:wave equation for r}-\eqref{eq:wave equation for phi}, on the ingoing cone, we have\begin{equation}
r_{v}(u,0) = (-u)^{k^{2}},\quad \Omega^{2}(u,0) = \frac{4}{1+k^{2}}(-u)^{k^{2}},\quad \partial_{v}\phi(u,0) = \frac{1}{(-u)}\left(\partial_{v}\phi(-1,0)-\frac{1}{k}\right)+\frac{1}{k}(-u)^{k^{2}-1}.
\end{equation}
To avoid the well-known blue-shift instability in \cite{christodoulou1999instability}, the choice of $\partial_{v}\phi$ on the bifurcating sphere $S_{-1,0}$ is $$\partial_{v}\phi(-1,0) = \frac{1}{k}.$$

To establish the global existence, we will choose $k$ sufficiently small as our smallness assumption on the initial data. However, one can immediately see that $\partial_{v}\phi$ on the initial ingoing cone $\partial_{v}\phi(u,0) = \frac{1}{k}(-u)^{k^{2}-1}$ is large pointwise when $k$ is small. 

To overcome this largeness, we adopt the idea introduced in \cite{rodnianski2023naked}. Instead of considering the $L^{\infty}$-norm of $\partial_{v}\phi$, using the exact $k$-self-similarity, we construct an auxiliary function $\left(\partial_{v}\phi\right)_{S}(u,v)$ in the region $\Rmnum{1}$ such that $\left(\partial_{v}\phi\right)_{S} = \partial_{v}\phi$ on $\underline{\Sigma}_{0}$ and $\Vert \left(\partial_{v}\phi\right)_{S}(u,\cdot)\Vert_{L_{v}^{p}([0,\underline{v}(-u)^{q_{k}}])}\lesssim k$. Then one can consider the renormalized quantity $\partial_{v}\phi-\left(\partial_{v}\phi\right)_{S}$ in the region $\Rmnum{1}$ which is exactly zero on the initial ingoing cone $\underline{\Sigma}_{0}$ and small in the $L^{p}$-norm on the initial outgoing cone $\{u = -1,\ 0\leq v\leq \underline{v}\}$. Then one can recover smallness for the renormalized quantity $\partial_{v}\phi-\left(\partial_{v}\phi\right)_{S}$. A well-designed bootstrap argument can close the construction in the region $\Rmnum{1}$.

To construct the auxiliary function $\left(\partial_{v}\phi\right)_{S}$, we consider the wave equation for $\phi$:\begin{equation}
\partial_{u}\partial_{v}\phi+\frac{r_{u}}{r}\partial_{v}\phi+\frac{r_{v}}{r}\partial_{u}\phi = 0.\label{eq:wave for phi in the Cri argument}
\end{equation}
Heuristically using the $k$-self-similarity, we assume $r$, $r_{u}$, $r_{v}$, and $\partial_{u}\phi$ take their values on the initial ingoing cone $\underline{\Sigma}_{0}$ and $\left(\partial_{v}\phi\right)_{S}$ takes the form of $\left(\partial_{v}\phi\right)_{S} = \frac{1}{(-u)^{q_{k}}}F\left(\frac{v}{(-u)^{q_{k}}}\right)$. Then we can derive the following equation for $F$ from \eqref{eq:wave for phi in the Cri argument}:\begin{equation}
q_{k}\frac{v}{(-u)^{q_{k}}}F^{\prime}\left(\frac{v}{(-u)^{q_{k}}}\right)-k^{2}F\left(\frac{v}{(-u)^{q_{k}}}\right)+k = 0.
\end{equation}
Solving the above equation directly, we obtain\begin{equation}
F(z) = z^{\frac{k^{2}}{1-k^{2}}}\left(\underline{v}\right)^{-\frac{k^{2}}{1-k^{2}}}F(\underline{v})+\frac{1}{k}\left(1-z^{\frac{k^{2}}{1-k^{2}}}(\underline{v})^{-\frac{k^{2}}{1-k^{2}}}\right),\ 0\leq z\leq\underline{v}.
\end{equation}
Let $F(\underline{v} ) = 0$. We can see that $\left(\partial_{v}\phi\right)_{S}$ agrees with $\partial_{v}\phi$ on $\underline{\Sigma}_{0}$. Moreover, one can directly check that $\Vert\left(\partial_{v}\phi\right)_{S}(u,\cdot)\Vert_{L_{v}^{p}([0,\underline{v}(-u)^{q_{k}}])}\lesssim k$. 
\begin{remark}
The regularity of the auxiliary function $\left(\partial_{v}\phi\right)_{S}(u,v)$ is $C_{v}^{0,\frac{k^{2}}{1-k^{2}}}$ in the $v$-direction, which is consistent with the regularity of the scalar field in Christodoulou's original construction \cite{christodoulou1994examples}. Intuitively, the largeness on the initial ingoing cone $\underline{\Sigma}_{0}$ and the smallness requirement of the $L^{p}$-norm suggest that $\left(\partial_{v}\phi\right)_{S}$ should decay fast near $\underline{\Sigma}_{0}$, hence the regularity of $\left(\partial_{v}\phi\right)_{S}$ should be limited.
\end{remark}
\subsection{A new perspective on the construction of an auxiliary function}
One can observe that the construction of $(\partial_{v}\phi)_{S}$ essentially uses the $k$-self-similar ansatz, in the sense that we assume the forms of $r$ and $(\partial_{v}\phi)_{S}$ in a small neighborhood of $\underline{\Sigma}_{0}$ are exactly determined by the $k$-self-similarity. In particular, we assume that $\left(\partial_{v}\phi\right)_{S}$  takes the form \begin{equation}
\frac{1}{(-u)^{q_{k}}}F\left(\frac{v}{(-u)^{q_{k}}}\right).\label{eq: speicial form of auxiliary function}
\end{equation}
The role of the $k$-self-similarity here is that, although the wave equation \eqref{eq: wave equation for the ESF} can be viewed as a $u$-transport equation for $\partial_{v}\phi$, the $k$-self-similarity allows one to propagate $\partial_{v}\phi$ in the $v$-direction, hence solving $(\partial_{v}\phi)_{S}$ in the whole region $\Rmnum{1}$ with given data on $\underline{\Sigma}_{0}$.

However, the specific $k$-self-similar form we assumed for $(\partial_{v}\phi)_{S},r,\partial_{u}\phi$ might be too specific to be generalized to the setting of this paper. We give a new perspective on constructing the auxiliary function $\left(\partial_{v}\phi\right)_{S}$. Instead of assuming the special form \eqref{eq: speicial form of auxiliary function} of $\left(\partial_{v}\phi\right)_{S}$, we can directly consider the $u$-transport equation of $\partial_{v}\phi$:\begin{equation}
\partial_{u}(r\partial_{v}\phi) = -r_{v}\partial_{u}\phi.\label{eq: heuristic ode}
\end{equation}
Let $r$, $r_{v}$ and $\partial_{u}\phi$ take their values on $\underline{C}$ and set $\partial_{v}\phi\left(-\left(\frac{v}{\underline{v}}\right)^{\frac{1}{q_{k}}},v\right) = 0$. For each fixed $0<v<\underline{v}$, solving the ODE \eqref{eq: heuristic ode} in the $u$-direction, we obtain the same auxiliary function:\begin{equation}
\left(\partial_{v}\phi\right)_{S} = \frac{1}{k}\frac{1}{(-u)^{q_{k}}}\left(1-\left(\frac{v}{\underline{v}}\right)^{\frac{k^{2}}{1-k^{2}}}(-u)^{-k^{2}}\right).
\end{equation}
In the above construction, we avoid guessing the special form of $\left(\partial_{v}\phi\right)_{S}$. The $k$-self-similar ansatz was used in the sense that $\left(\partial_{v}\phi\right)_{S}$ is zero on the cone $\{(u,v):\ \frac{v}{(-u)^{q_{k}}} = \underline{v}\}$, which corresponds to $\{z = \underline{v}\}$ under the $k$-self-similar coordinate $z: = \frac{v}{(-u)^{q_{k}}}$.
\section{Initial data for the exterior construction}
In this section, we construct the exterior region of the naked singularity spacetime solving the Einstein--Maxwell-charged scalar field equations. We prescribe initial data on the characteristic hypersurface $\underline{\Sigma}_{0} \cup \Sigma_{-e^{-1}}$ and establish global existence of the solution in the region $\{-e^{-1} \leq u < 0,\ v \geq 0\}$. Moreover, we show that no trapped surfaces form in this region. On the initial ingoing cone $\underline{\Sigma}_{0}$, instead of using the data from the interior construction, we consider more general data and aim to provide a general framework for the exterior construction. The free initial data consist of \begin{itemize}
    \item On the ingoing cone $\underline{\Sigma}_{0}$, we shall fix the gauge choice of $r$ and give initial data for $\phi$.
    \item On the outgoing cone $\Sigma_{-e^{-1}}$, we shall give the initial data for $\lambda$ and $\partial_{v}\phi$.
    \item On the intersecting sphere $(u,v) = (-e^{-1},0)$, we shall give the values of $\Omega^{2}(-e^{-1},0) = \Delta$ and $\phi(-e^{-1},0) = 0$.
\end{itemize}

\subsection{Setup of the initial data on the initial ingoing cone $\underline{\Sigma}_{0}$}
On the ingoing cone $\underline{\Sigma}_{0}$, we choose $r(u,0) = (-u)$. Then we have $\nu = \partial_{u}r = -1$. We consider the polar form of the scalar field \begin{equation*}
    \phi = \rho e^{i\theta} = \rho\cos\theta+i\rho\sin\theta.
\end{equation*}
Let $k>0$ and $\kappa\geq0$ be two constants defined as follows \begin{equation*}
    k: = \lim_{u\rightarrow 0^{-}}(-u)\partial_{u}\rho(u,0),\quad \kappa: = \lim_{u\rightarrow 0^{-}}(-u)\rho(u,0)\partial_{u}\theta(u,0),\quad q^{2}: = k^{2}+\kappa^{2}.
\end{equation*}
Therefore, we can assume that \begin{equation*}
    \partial_{u}\rho(u,0) = \frac{k}{(-u)}-\rho_{1}(u),\quad \rho\partial_{u}\theta (u,0)= \frac{\kappa}{(-u)}-\theta_{1}(u),
\end{equation*}
where $\rho_{1}$ and $\theta_{1}$ are integrable on $[-e^{-1},0]$ and satisfy \begin{equation}
    \left\vert\rho_{1}(u)\right\vert\lesssim q\frac{1}{\left\vert u\right\vert \left\vert\log (-u)\right\vert^{2}},\quad \left\vert\theta_{1}(u)\right\vert\lesssim q^{2}\frac{1}{\left\vert u\right\vert\left\vert\log(-u)\right\vert}.
\end{equation} 
Then the initial data for $\rho$ and $\theta$ on the initial ingoing cone $\underline{\Sigma}_{0}$ can be written as \begin{align*}
    \rho(u,0) &= -k\log(-u)+\rho_{0}+\int_{u}^{0}\rho_{1}(s)ds,\\
    \theta(u,0)& = \int_{-1}^{u}\frac{\kappa}{(-s)\rho(s,0)}ds+\theta_{0}+\int_{u}^{0}\frac{\theta_{1}(s)}{\rho(s,0)}ds.
\end{align*}
Then, we can derive the following estimate for the spacetime charge $Q$ on the ingoing cone $\underline{\Sigma}_{0}$:
\begin{equation}
\begin{aligned}
   \left\vert Q(u,0) \right\vert&= \left\vert\int_{0}^{u}q_{0}r^{2}(s,0)\rho^{2}(s,0)\partial_{u}\theta(s,0)\right\vert\\&\approx \vert q_{0}\kappa k\vert u^{2}\vert \log(-u)\vert.
   \end{aligned}
   \label{eq: upper bound for Q in the exterior construction}
\end{equation}
\begin{remark}
    If $\kappa\neq0$ and $k\neq 0$, then along the ingoing cone $\{v = 0\}$, the spacetime charge $Q(u,0)\sim q_{0}k\kappa u^{2}\log(-u)$ decays more slowly as $u\rightarrow 0$ than the $u^{2}$ behavior of the interior solution constructed in Section~\ref{sec: interior construction}. Thus, the exterior construction is more general, allowing a slower decay rate for $Q$ along the ingoing cone $\{v = 0\}$. By taking $\kappa=0$, the exterior construction can be glued to the interior solution constructed in Section~\ref{sec: interior construction}.
\end{remark}

\begin{example}
For the spacetime constructed in Section \ref{sec: interior construction}, taking $\gamma = 0$ and $\delta = 1$, we have \begin{equation*}
    \phi_{1}(u,0) = -k_{1}\log(-u)+1,\quad \phi_{2}(u,0) = -k_{2}\log(-u).
\end{equation*}
Then, we have \begin{equation*}
    \rho^{2} = k^{2}\log^{2}(-u)-2k_{1}\log(-u)+1,\quad \cos\theta = \frac{-k_{1}\log(-u)+1}{\rho},\quad \sin\theta = \frac{-k_{2}\log(-u)}{\rho}.
\end{equation*}
Taking the $u$-derivative, we have \begin{align*}
    -\sin\theta\partial_{u}\theta = \partial_{u}\left(\frac{-k_{1}\log(-u)+1}{\rho}\right).
\end{align*}
Then we can compute $\partial_{u}\theta(u,0)$: \begin{equation*}
    \rho\partial_{u}\theta = \frac{k_{2}}{(-u)\rho\log(-u)}.
\end{equation*}
Then, we have \begin{equation*}
    \lim_{u\rightarrow 0^{-}}(-u)\rho\partial_{u}\theta = 0\Rightarrow \kappa = 0.
\end{equation*}
Hence, we have \begin{equation*}
   \left\vert \partial_{u}\rho-\frac{k}{(-u)}\right\vert\approx \frac{k}{(-u)\log^{2}(-u)},\quad \left\vert \rho\partial_{u}\theta\right\vert\approx \frac{k_{2}}{u\log(-u)}.
\end{equation*}
\end{example}
\begin{example}
    Taking $\rho = -k\log(-u)$ and $\theta = \frac{\kappa}{k}\log(-\log(-u))$, we have \begin{equation*}
        \partial_{u}\rho = \frac{k}{(-u)},\quad \rho\partial_{u}\theta = \frac{\kappa}{(-u)}.
    \end{equation*}
    In this case, we have \begin{align*}
        \partial_{u}Q = q_{0}r^{2}\rho^{2}\partial_{u}\theta = q_{0}k\kappa u\log(-u).
    \end{align*}
    Therefore, we obtain\begin{equation*}
Q(u,0) = \frac{1}{2}q_{0}k\kappa u^{2}\left(\log(-u)-\frac{1}{2}\right).
    \end{equation*}
\end{example}
The $u$-transport equation \eqref{eq:u-transport equation for Q} for $Q$ can be reduced to \begin{equation*}
    \partial_{u}Q(u,0) = q_{0}u^{2}\left(\phi_{1}\partial_{u}\phi_{2}-\phi_{2}\partial_{u}\phi_{1}\right) = q_{0}u^{2}\rho^{2}\partial_{u}\theta.
\end{equation*}
Integrating the above equation, we have \begin{equation}
    \left\vert Q\right\vert(u,0)\leq \vert q_{0}\vert k\kappa u^{2}\left\vert\log(-u)\right\vert.\label{eq: exterior construction initial estimate for Q}
\end{equation}
\begin{remark}
The estimate \eqref{eq: exterior construction initial estimate for Q} for the spacetime charge $Q$ on the ingoing cone $\underline{\Sigma}_{0}$ is sharp, in the sense that the upper bound is attained when one takes $\rho(u,0) = -k\log(-u)$ and $\theta = \frac{\kappa}{k}\log(-\log(-u))$. However, in view of the interior construction in Section \ref{sec: interior construction}, for the choice of $\phi$ arising from the naked singularity interior, one obtains at most a $u^{2}$ decay rate for the spacetime charge $Q$ on $\underline{\Sigma}_{0}$. Since the interior construction is a perturbation of the $(k_{1},k_{2})$-self-similar solution, which is in the same spirit as Christodoulou's original construction \cite{christodoulou1994examples}, this $u^{2}$ decay may be regarded as characteristic of the self-similar framework. Consequently, our exterior construction applies to a strictly larger class of initial data than that considered in \cite{christodoulou1994examples}, as it accommodates the logarithmically weaker $u^{2}\log(-u)$ decay for $Q$ on $\underline{\Sigma}_{0}$, rather than being restricted to the $u^{2}$ rate dictated by the interior model.
\end{remark}

Let $\Omega^{2}(-e^{-1},0) = \Delta$. Using \eqref{eq:spherical symmetric equaion1}, we can determine $\Omega^{2}(u,0)$: \begin{equation}
    \Omega^{2}(u,0) = \Delta e^{\int_{-e^{-1}}^{u}s\left\vert\partial_{u}\phi\right\vert^{2}(s,0)ds} = \Delta e^{\int_{-e^{-1}}^{u}s\left((\partial_{u}\rho)^{2}+\rho^{2}(\partial_{u}\theta)^{2}\right)(s,0)ds}.
\end{equation}
The leading order behavior for $\Omega^{2}(u,0)$ is \begin{equation}
    \Omega^{2}(u,0)\approx (-u)^{k^{2}+\kappa^{2}},\quad u\rightarrow0^{-}.\label{eq: estimate on omega on the ingoing cone for the exterior construction}
\end{equation}
Using \eqref{eq:wave equation for r}, we have \begin{equation*}
    (-u)\partial_{u}\lambda(u,0)+(-1)\lambda(u,0) = -\frac{\Omega^{2}}{4}\left(1-\frac{Q^{2}}{r^{2}}\right)(u,0).
\end{equation*}
Directly integrating the above equation, we have \begin{equation*}
    \lambda(u,0) = \frac{1}{(-u)}\int_{u}^{0}\frac{1}{4}\Omega^{2}(s,0)\left(1-\frac{Q^{2}(s,0)}{s^{2}}\right)ds.
\end{equation*}
Hence, the leading order behavior for $\lambda(u,0)$ is \begin{equation}
    \lambda(u,0)\approx (-u)^{k^{2}+\kappa^{2}},\quad u\rightarrow 0^{-}.\label{eq: estimate for lambda on the ingoing cone for the exterior construction}
\end{equation}
Using the relation $1-\mu = \frac{-4\lambda\nu}{\Omega^{2}}$, we have \begin{equation*}
    \mu (u,0)= 1-\frac{4\lambda}{\Omega^{2}}(u,0)\rightarrow\frac{q^{2}}{1+q^{2}},\quad u\rightarrow 0^{-}.
\end{equation*}
Next, we study the behavior of $\partial_{v}\phi_{1}(u,0)$ and $\partial_{v}\phi_{2}(u,0)$. We can write the wave equation \eqref{eq:wave equation for phi} for $\phi$ as \begin{align*}
    \partial_{u}\left((-u)\partial_{v}\phi_{1}(u,0)\right)& = -\lambda(u,0)\partial_{u}\phi_{1}(u,0)-q_{0}\frac{Q(u,0)}{(-u)}\frac{\Omega^{2}(u,0)}{4}\phi_{2}(u,0),\\
    \partial_{u}\left((-u)\partial_{v}\phi_{2}(u,0)\right)& = -\lambda(u,0)\partial_{u}\phi_{2}(u,0)+q_{0}\frac{Q(u,0)}{(-u)}\frac{\Omega^{2}(u,0)}{4}\phi_{1}(u,0).
\end{align*}
Integrating the above equations, we have \begin{align*}
    (-u)\partial_{v}\phi_{1}(u,0) =& \frac{1}{e}\partial_{v}\phi_{1}(-e^{-1},0)+\underbrace{\int_{-e^{-1}}^{0}-\lambda(u,0)\partial_{u}\phi_{1}(u,0)-q_{0}\frac{Q(u,0)}{(-u)}\frac{\Omega^{2}(u,0)}{4}\phi_{2}(u,0)du}_{I_{1}}\\&-\int_{u}^{0}-\lambda(s,0)\partial_{u}\phi_{1}(s,0)-q_{0}\frac{Q(s,0)}{(-s)}\frac{\Omega^{2}(s,0)}{4}\phi_{2}(s,0)ds,\\
     (-u)\partial_{v}\phi_{2}(u,0) =& \frac{1}{e}\partial_{v}\phi_{2}(-e^{-1},0)+\underbrace{\int_{-e^{-1}}^{0}-\lambda(u,0)\partial_{u}\phi_{2}(u,0)+q_{0}\frac{Q(u,0)}{(-u)}\frac{\Omega^{2}(u,0)}{4}\phi_{1}(u,0)du}_{I_{2}}\\&-\int_{u}^{0}-\lambda(s,0)\partial_{u}\phi_{2}(s,0)+q_{0}\frac{Q(s,0)}{(-s)}\frac{\Omega^{2}(s,0)}{4}\phi_{1}(s,0)ds.
\end{align*}
Although we will postpone the choice of the initial data on $\Sigma_{-e^{-1}}$ to the next section, we shall fix the values of $\partial_{v}\phi_{1}(-e^{-1},0)$ and $\partial_{v}\phi_{2}(-e^{-1},0)$ to be \begin{equation*}
    \partial_{v}\phi_{1}(-e^{-1},0) = -eI_{1},\quad \partial_{v}\phi_{2}(-e^{-1},0) = -eI_{2}.
\end{equation*}
\begin{example}
    For the initial data given by the interior construction in Section \ref{sec: interior construction}, taking $\phi_{1}(u,0) = -k_{1}\log(-u)+\delta$ and $\phi_{2} = -k_{2}\log(-u)$, we have \begin{align*}
        &Q(u,0) = -\frac{1}{2}k_{2}\delta q_{0}u^{2},\quad \Omega^{2}(u,0) = \Delta e^{k^{2}}(-u)^{k^{2}},\\& \lambda(u,0) = \frac{\Delta e^{k^{2}}}{4(1+k^{2})}(-u)^{k^{2}}-\frac{\Delta e^{k^{2}}k_{2}^{2}\delta^{2}q_{0}^{2}}{16(k^{2}+3)}(-u)^{k^{2}+2}.
    \end{align*}
    Hence, we can compute $I_{1}$ and $I_{2}$: \begin{equation*}
        I_{1}\approx \frac{k_{1}}{k^{2}},\quad I_{2}\approx \frac{k_{2}}{k^{2}},
    \end{equation*}
    which are large quantities when $k$ is sufficiently small.
    \label{example: largeness}
\end{example}
\begin{example}
\label{example: exact loglog behavior}
    For the initial data given by \begin{equation*}
        \rho = -k\log(-u),\quad \theta = \frac{\kappa}{k}\log(-\log(-u)),
    \end{equation*}
    we have \begin{equation*}
        \phi_{1}(u,0) = -k\log(-u)\cos\left(\frac{\kappa}{k}\log(-\log(-u))\right),\quad \phi_{2}(u,0) = -k\log(-u)\sin\left(\frac{\kappa}{k}\log(-\log(-u))\right).
    \end{equation*}
    Then, we have \begin{align*}
        Q(u,0) =& \frac{1}{2}q_{0}k\kappa u^{2}\left(\log(-u)-\frac{1}{2}\right),\quad \Omega^{2}(u,0) = \Delta e^{q^{2}}(-u)^{q^{2}},\\
        \lambda(u,0) =&\frac{\Delta e^{q^{2}}}{4(1+q^{2})}(-u)^{q^{2}}-\frac{\Delta e^{q^{2}}}{16}q_{0}^{2}k^{2}\kappa^{2}(-u)^{q^{2}+2}\frac{1}{q^{2}+3}\left(\log(-u)-\frac{1}{2}\right)^{2}\\&-\frac{\Delta e^{q^{2}}}{16}q_{0}^{2}k^{2}\kappa^{2}(-u)^{q^{2}+2}\left(-\frac{2}{(q^{2}+3)^{2}}\left(\log(-u)-\frac{1}{2}\right)+\frac{2}{(q^{2}+3)^{3}}\right).
    \end{align*}
    Then for $q^{2}$ sufficiently small, we have \begin{align*}
        &-I_{1}\approx \int_{-e^{-1}}^{0}(-u)^{q^{2}}\left(\frac{k}{(-u)}\cos\left(\frac{\kappa}{k}\log(-\log(-u))\right)-\frac{\kappa}{(-u)}\sin\left(\frac{\kappa}{k}\log(-\log(-u))\right)\right)du,\\&
        -I_{2}\approx \int_{-e^{-1}}^{0}(-u)^{q^{2}}\left(\frac{k}{(-u)}\sin\left(\frac{\kappa}{k}\log(-\log(-u))\right)+\frac{\kappa}{(-u)}\cos\left(\frac{\kappa}{k}\log(-\log(-u))\right)\right)du.
    \end{align*}
    Taking $k = \kappa$, we can draw the graph (Figure \ref{fig:placeholder}) on the behavior of \begin{equation*}
        \int_{-e^{-1}}^{0}k(-u)^{q^{2}-1}\cos\left(\log(-\log(-u))\right) du = \frac{\sqrt{2}}{2}q\int_{1}^{\infty}e^{-q^{2}t}\cos(\log(t))dt=:I(q)
        \end{equation*}
    when $q\rightarrow 0$, which shows $I(q)$ oscillates between $\frac{1}{q}$ and $-\frac{1}{q}$ when $q\rightarrow0$.
\begin{figure}
    \centering
    \includegraphics[width=0.5\linewidth]{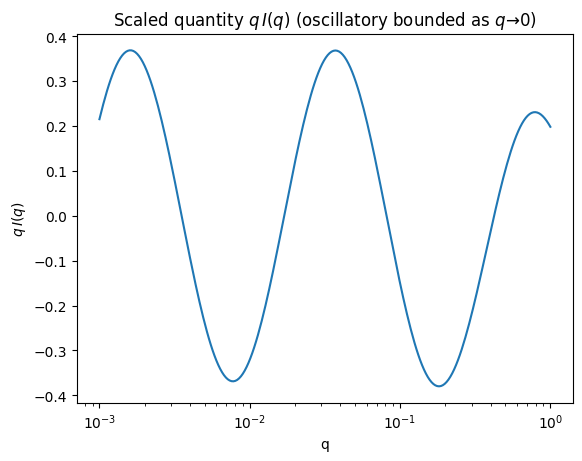}
    \caption{Graph for $qI(q)$}
    \label{fig:placeholder}
\end{figure}
\end{example}
\subsection{Construction of auxiliary functions}
In this section, we construct auxiliary functions $(\partial_{v}\phi_{1})_{S}$ and $\left(\partial_{v}\phi_{2}\right)_{S}$ that are necessary for our exterior construction. First, we define a projection operator: \begin{equation*}
    \mathcal{P}(f(u,v)) = f(u,0),
\end{equation*}
which projects a function to its value on the ingoing cone $\underline{\Sigma}_{0}$. We consider the wave equation \eqref{eq:wave equation for phi} in the region $\left\{0\leq \frac{v}{(-u)^{1-q^{2}}}\leq \underline{v}\right\}$ for $\underline{v}>0$, which will be chosen to be small later. Let $\left(\partial_{v}\phi\right)_{S} = \left(\partial_{v}\phi_{1}\right)_{S}+i\left(\partial_{v}\phi_{2}\right)_{S}$ be a solution to \begin{equation*}
    \mathcal{P}(r)\partial_{u}\left(\partial_{v}\phi\right)_{S}+\mathcal{P}(\partial_{u}r)\left(\partial_{v}\phi\right)_{S}+\mathcal{P}(\partial_{v}r)\mathcal{P}(\partial_{u}\phi) = iq_{0}\mathcal{P}\left(\frac{Q\Omega^{2}}{4r}\phi\right).
\end{equation*}
Since $\mathcal{P}(r) = r(u,0) = (-u)$, we have \begin{align}
    \partial_{u}\left((-u)\left(\partial_{v}\phi_{1}\right)_{S}\right) &= -\lambda(u,0)\partial_{u}\phi_{1}(u,0)+q_{0}\frac{Q(u,0)\Omega^{2}(u,0)}{4(-u)}\phi_{2}(u,0),\label{eq: the equation for constructing dvphi1s}\\
    \partial_{u}\left((-u)\left(\partial_{v}\phi_{2}\right)_{S}\right) &= -\lambda(u,0)\partial_{u}\phi_{2}(u,0)-q_{0}\frac{Q(u,0)\Omega^{2}(u,0)}{4(-u)}\phi_{1}(u,0),\label{eq: the equation for constructing dvphi2s}.
\end{align}
Since the above equations~\eqref{eq: the equation for constructing dvphi1s}--\eqref{eq: the equation for constructing dvphi2s} are transport equations in the $u$ direction, to construct the auxiliary functions $(\partial_{v}\phi_{1})_{S}$ and $(\partial_{v}\phi_{2})_{S}$ in the region $\left\{0\leq\frac{v}{(-u)^{1-q^{2}}}\leq\underline{v}\right\}$, we have to impose data for $(\partial_{v}\phi_{1})_{S}$ and $(\partial_{v}\phi_{2})_{S}$ on the hypersurface $\left\{\frac{v}{(-u)^{1-q^{2}}}=\underline{v}\right\}$. We set \begin{equation*}
    \left(\partial_{v}\phi_{1}\right)_{S}\left(u,(-u)^{1-q^{2}}\underline{v}\right) = 0,\quad \left(\partial_{v}\phi_{2}\right)_{S}\left(u,(-u)^{1-q^{2}}\underline{v}\right) = 0.
\end{equation*}
Then we can solve for functions $\left(\partial_{v}\phi_{1}\right)_{S}$ and $\left(\partial_{v}\phi_{2}\right)_{S}$ \begin{align}
    \left(\partial_{v}\phi_{1}\right)_{S}(u,v) =& \frac{1}{(-u)}\int^{-\left(\frac{v}{\underline{v}}\right)^{\frac{1}{1-q^{2}}}}_{u}-\lambda(s,0)\partial_{u}\phi_{1}(s,0)-q_{0}\frac{Q(s,0)\Omega^{2}(s,0)}{4(-s)}\phi_{2}(s,0)ds,\\
    \left(\partial_{v}\phi_{2}\right)_{S}(u,v)=&\frac{1}{(-u)}\int_{u}^{-\left(\frac{v}{\underline{v}}\right)^{\frac{1}{1-q^{2}}}}-\lambda(s,0)\partial_{u}\phi_{2}(s,0)+q_{0}\frac{Q(s,0)\Omega^{2}(s,0)}{4(-s)}\phi_{1}(s,0)ds.
\end{align}
One can check that when $v = 0$, these auxiliary functions are consistent with $\partial_{v}\phi_{1}(u,0)$ and $\partial_{v}\phi_{2}(u,0)$, respectively. Although the numerical simulation (see Figure \ref{fig:placeholder}) shows that $\partial_{v}\phi_{1}(u,0)$ and $\partial_{v}\phi_{2}(u,0)$ might behave like $\frac{1}{q}$ as $q\rightarrow0$, which will cause severe difficulty when proving the global existence, the following lemma shows that these auxiliary functions are $q$-small in the $L^{1}$-sense.
\begin{lemma}
    For the auxiliary functions $(\partial_{v}\phi_{1})_{S}$ and $\left(\partial_{v}\phi_{2}\right)_{S}$, we have the following estimates 
    \begin{align}
&\left\vert\left(\partial_{v}\phi_{1}\right)_{S}\right\vert\lesssim\frac{1}{q}(-u)^{q^{2}-1}\left[1-\left(\frac{v}{(-u)^{1-q^{2}}}\right)^{\frac{q^{2}}{1-q^{2}}}\left(\frac{1}{\underline{v}}\right)^{\frac{q^{2}}{1-q^{2}}}\right],\label{eq: pointwise estimate for auxiliary function}\\&\left\vert\left(\partial_{v}\phi_{2}\right)_{S}\right\vert\lesssim\frac{1}{q}(-u)^{q^{2}-1}\left[1-\left(\frac{v}{(-u)^{1-q^{2}}}\right)^{\frac{q^{2}}{1-q^{2}}}\left(\frac{1}{\underline{v}}\right)^{\frac{q^{2}}{1-q^{2}}}\right],\\
        &\int_{0}^{(-u)^{1-q^{2}}\underline{v}}\left\vert\left(\partial_{v}\phi_{1}\right)_{S}\right\vert(u,v)dv\lesssim q\underline{v},\quad 
        \int_{0}^{(-u)^{1-q^{2}}\underline{v}}\left\vert\left(\partial_{v}\phi_{2}\right)_{S}\right\vert(u,v)dv\lesssim q\underline{v}.\label{eq: L1 estimate for the auxilary function}
    \end{align}
\end{lemma}
\begin{proof}
    We only prove the estimate for $(\partial_{v}\phi_{1})_{S}$; the estimate for $(\partial_{v}\phi_{2})_{S}$ follows similarly. Using the Fubini theorem, we have \begin{equation}
        \begin{aligned}
            &\frac{1}{(-u)}\int_{0}^{(-u)^{1-q^{2}}\underline{v}}\int_{u}^{-\left(\frac{v}{\underline{v}}\right)^{\frac{1}{1-q^{2}}}}\left\vert\lambda(s,0)\partial_{u}\phi_{1}(s,0)\right\vert+\left\vert q_{0}\frac{Q(s,0)\Omega^{2}(s,0)}{4(-s)}\phi_{2}(s,0)\right\vert dsdv\\=&
            \frac{1}{(-u)}\int_{u}^{0}\int_{0}^{(-s)^{1-q^{2}}\underline{v}}\left\vert\lambda(s,0)\partial_{u}\phi_{1}(s,0)\right\vert+\left\vert q_{0}\frac{Q(s,0)\Omega^{2}(s,0)}{4(-s)}\phi_{2}(s,0)\right\vert dvds\\=&\frac{1}{(-u)}\int_{u}^{0}(-s)^{1-q^{2}}\underline{v}\left(\left\vert\lambda(s,0)\partial_{u}\phi_{1}(s,0)\right\vert+\left\vert q_{0}\frac{Q(s,0)\Omega^{2}(s,0)}{4(-s)}\phi_{2}(s,0)\right\vert\right)ds\\&\lesssim \frac{1}{(-u)}\int_{u}^{0}(-s)^{1-q^{2}}\underline{v}\left((-s)^{q^{2}}\frac{q}{(-s)}+q_{0}^{2}q(-s)^{1+q^{2}}\log(-s)^{2}\right)ds\\&\lesssim \frac{q\underline{v}}{(-u)}\int_{u}^{0}ds\lesssim q\underline{v},
        \end{aligned}
    \end{equation}
    where we have used the estimates \eqref{eq: exterior construction initial estimate for Q}, \eqref{eq: estimate on omega on the ingoing cone for the exterior construction}, and \eqref{eq: estimate for lambda on the ingoing cone for the exterior construction}.
\end{proof}
\begin{remark}
    The construction of the exterior solutions may be viewed as a global existence result for the Einstein--Maxwell-charged scalar field equations in the region $\mathcal{Q}^{(ex)}$, with initial data prescribed on $\underline{\Sigma}_{0}\cup\Sigma_{-e^{-1}}$. In view of the nonlinear hyperbolic nature of the system, we impose suitable smallness assumptions on the initial data. Nevertheless, as illustrated in Example~\ref{example: largeness} and Example~\ref{example: exact loglog behavior}, the quantity $\partial_{v}\phi$ may be large in $L^{\infty}$ near the intersecting sphere $(u,v) = (-e^{-1},0)$. Thus, one cannot in general expect the initial data for $\partial_{v}\phi$ to be small in $L^{\infty}$ on the outgoing initial hypersurface $\Sigma_{-e^{-1}}$. It is, however, still reasonable to impose smallness in $L^{p}$. The auxiliary functions $(\partial_{v}\phi_{1})_{S}$ and $(\partial_{v}\phi_{2})_{S}$ introduced above are designed to impose small initial data in $L^{1}$ on $\Sigma_{-e^{-1}}$. The initial data for $\partial_{v}\phi$ on $\Sigma_{-e^{-1}}$ can be thought of as a suitable perturbation of $(\partial_{v}\phi)_{S}$.
\end{remark}

\subsection{Setup of the initial data on the initial outgoing cone $\Sigma_{-e^{-1}}$}
The initial data on the outgoing cone $\{u = -e^{-1},\ v\geq 0\}$ are stated as in Theorem~\ref{main theorem1}.
\section{Exterior construction}
\subsection{Construction of region I}
In this section, we show that under the given characteristic initial data, the solution exists and is free of trapped surfaces in the region I: $\left\{-e^{-1}\leq u<0,\ 0\leq\frac{v}{(-u)^{1-q^{2}}}\leq \underline{v}\right\}$.\begin{figure}[htbp]
    \centering
    \begin{tikzpicture}[ scale=0.8, >=Latex, every node/.style={font=\small} ]
\coordinate (O) at (0,0);
\coordinate (A) at (4,-4);
\coordinate (B) at (6,6);
\coordinate (C) at (10,2);
\draw (O)--(A);
\draw[dashed] (O)--(B);
\draw (A)--(C);
\draw[dashed] (B) -- (C)
    node[midway, above, sloped]
    {\rotatebox{90}{$\mathscr{I}$}};
\node[left] at (O) { $\mathcal{O}$};
\coordinate (A1) at (6,-2);
\coordinate (A2) at (8,0);
\draw[dashed] (O)--(A1);
\draw[dashed] (O)--(A2);
\coordinate (A3) at (3,0);
\coordinate (A4) at (1.5,1.5);
\draw[dashed] (A3)--(A4);
\coordinate (L1) at (2.5,-1.5);
\coordinate (L2) at (3.5,-0.5);
\node at (L2) {II};
\coordinate (L3) at (1.5,0.5);
\node at (L3) {III};
\coordinate (L4) at (5,3);
\node at (L4){IV};
\fill[gray!30] (O) -- (A1) -- (A) -- cycle;
\node at (L1){I};
\end{tikzpicture}
    \caption{$R_{I}$}
    \label{fig: region 1}
\end{figure}
\subsubsection{Local well-posedness and continuation criteria}
To prove the existence of a solution that arises from our initial data in the region $I$, we will appeal to the standard local well-posedness result and apply the bootstrap argument to extend the solution to the whole region $I$, using the continuation criterion. The standard local well-posedness result is stated as follows.\begin{proposition}
For some $v_{0}>0$, given $BV$ ($C^{k}$) initial data on $\underline{\Sigma}_{0}\cup\Sigma_{-e^{-1}}$, there exists a decreasing function $T(u):[-1,0]\rightarrow\mathbb{R}$ with $T(-1) = v_{0}$ and $T(0) = 0$ such that there exists a unique $BV$ ($C^{k}$) solution to the Einstein--Maxwell-charged scalar field equations arising from the initial data in the region\begin{equation*}
\{(u,v),\ -1\leq u<0,\ 0\leq v<T(u)\}.
\end{equation*}
\label{local well-posedness: label}
\end{proposition}
The standard continuation criterion is given in the following proposition.
\begin{proposition}
Given $BV$ initial data on the characteristic bifurcating hypersurface $\Sigma_{-e^{-1}}\cup\underline{\Sigma}_{0}$, let $R$ be the maximal hyperbolic development of the initial data. Then one of the quantities $\{(\partial_{v}r)^{-1},(-\partial_{u}r),\partial_{u}\phi,\partial_{v}\phi,Q,\Omega^{2}\}$ must be unbounded in the region $R$.
\end{proposition}
\subsubsection{The renormalized quantities and the renormalized equations}
Since the quantity $\partial_{v}\phi(u,0)$ is a pointwise large quantity when $q$ is taken to be sufficiently small, to deal with this largeness, we renormalize the quantities we shall consider and derive the equations for the renormalized quantities in this section. Without ambiguity\footnote{Note that the notation $\widetilde{\Psi}$ has been used in the interior construction to denote the background $(k_{1},k_{2})$-self-similar solution to the Einstein--Maxwell-uncharged scalar field equations. Since our exterior construction is not perturbative, we avoid introducing a new notation and simply use $\widetilde{\Psi}$ here to denote the renormalized quantity.}, we use $\widetilde{\Psi}$ to denote \begin{equation*}
    \widetilde{\Psi}(u,v) = \Psi(u,v)-\Psi(u,0),\quad \Psi\in\left\{\partial_{u}r,\ \partial_{v}r,\ \mu,\ \Omega^{2},\ \partial_{u}\phi,\ Q\right\}.
\end{equation*}
For $\partial_{v}\phi$, we make the following renormalization:\begin{equation*}
\widetilde{\partial_{v}\phi_{1}} = \partial_{v}\phi_{1}-\left(\partial_{v}\phi_{1}\right)_{S},\quad \widetilde{\partial_{v}\phi_{2}} = \partial_{v}\phi_{2}-\left(\partial_{v}\phi_{2}\right)_{S}.
\end{equation*}
The renormalized quantities satisfy the following equations: \begin{align}
\partial_{u}\widetilde{\lambda} &= \frac{\nu}{r(1-\mu)}\left(\mu-\frac{Q^{2}}{r^{2}}\right)\widetilde{\lambda}+\frac{\nu}{r(1-\mu)}\left(\mu-\frac{Q^{2}}{r^{2}}\right)\mathcal{P}(\lambda)\nonumber\\&\quad -\mathcal{P}\left(\frac{\lambda\nu}{r(1-\mu)}\left(\mu-\frac{Q^{2}}{r^{2}}\right)\right),\allowdisplaybreaks\label{eq:renormalized rv equation}\\
\partial_{v}\widetilde{\nu}&=\frac{\lambda}{r(1-\mu)}\left(\mu-\frac{Q^{2}}{r^{2}}\right)\widetilde{\nu}+\frac{\lambda}{r(1-\mu)}\left(\mu-\frac{Q^{2}}{r^{2}}\right)\mathcal{P}(\nu),\allowdisplaybreaks\label{eq:renormalized ru equation}\\
\partial_{u}\left(\widetilde{\partial_{v}\phi_{1}}\right)+\frac{\nu}{r}\widetilde{\partial_{v}\phi_{1}}&=-\widetilde{\left(\frac{\nu}{r}\right)}\left(\partial_{v}\phi_{1}\right)_{S}-\widetilde{\left(\frac{\lambda}{r}\partial_{u}\phi_{1}\right)}+q_{0}\left(\frac{A_{u}}{r}\partial_{v}(r\phi_{2})+\widetilde{\left(\frac{Q}{r^{2}}\frac{\lambda\nu}{1-\mu}\phi_{2}\right)}\right),\allowdisplaybreaks\label{eq:renormalized dv-phi1 equation}\\\partial_{u}\left(\widetilde{\partial_{v}\phi_{2}}\right)+\frac{\nu}{r}\widetilde{\partial_{v}\phi_{2}}&=-\widetilde{\left(\frac{\nu}{r}\right)}\left(\partial_{v}\phi_{2}\right)_{S}-\widetilde{\left(\frac{\lambda}{r}\partial_{u}\phi_{2}\right)}-q_{0}\left(\frac{A_{u}}{r}\partial_{v}(r\phi_{1})+\widetilde{\left(\frac{Q}{r^{2}}\frac{\lambda\nu}{1-\mu}\phi_{1}\right)}\right),\allowdisplaybreaks\\
\partial_{v}\left(\widetilde{\partial_{u}\phi_{1}}\right)+\frac{\lambda}{r}\widetilde{\partial_{u}\phi_{1}} &= -\frac{\lambda}{r}\mathcal{P}(\partial_{u}\phi_{1})-\frac{\nu}{r}\widetilde{\partial_{v}\phi_{1}}-\frac{\nu}{r}\left(\partial_{v}\phi_{1}\right)_{S}+q_{0}\left(\frac{A_{u}\partial_{v}(r\phi_{2})}{r}+\frac{Q}{r^{2}}\frac{\lambda\nu}{1-\mu}\phi_{2}\right),\allowdisplaybreaks\label{eq:renormalized du-phi1 equation}\\
\partial_{v}\left(\widetilde{\partial_{u}\phi_{2}}\right)+\frac{\lambda}{r}\widetilde{\partial_{u}\phi_{2}} &= -\frac{\lambda}{r}\mathcal{P}(\partial_{u}\phi_{2})-\frac{\nu}{r}\widetilde{\partial_{v}\phi_{2}}-\frac{\nu}{r}\left(\partial_{v}\phi_{2}\right)_{S}-q_{0}\left(\frac{A_{u}\partial_{v}(r\phi_{1})}{r}+\frac{Q}{r^{2}}\frac{\lambda\nu}{1-\mu}\phi_{1}\right),\allowdisplaybreaks\\
\partial_{v}\widetilde{\mu}+\left(\frac{\lambda}{r}+\frac{r}{\lambda}\left\vert\partial_{v}\phi\right\vert^{2}\right)\widetilde{\mu}&=-\left(\frac{\lambda}{r}+\frac{r}{\lambda}\left\vert\partial_{v}\phi\right\vert^{2}\right)\mathcal{P}(\mu)+\frac{r}{\lambda}\left\vert\partial_{v}\phi\right\vert^{2}+\frac{\lambda}{r^{3}}Q^{2},\allowdisplaybreaks\label{eq: renormalized dv-mu equation}
\\
\partial_{v}\widetilde{Q} &= q_{0}r^{2}\left(\phi_{2}\partial_{v}\phi_{1}-\phi_{1}\partial_{v}\phi_{2}\right),\\
\partial_{v}A_{u}& = \frac{2Q}{r^{2}}\frac{\lambda\nu}{1-\mu}.
\end{align}
\subsubsection{Bootstrap assumption and its consequence}
Let $\mathcal{Q}^{(ex)}_{u_{0},v_{0}}$ be the region \begin{equation*}
    \{-e^{-1}\leq u\leq u_{0},\ 0\leq v\leq v_{0}\}.
\end{equation*}
To show that the solution for the Einstein--Maxwell-charged scalar field equations exists in the whole region $\Rmnum{1}$, we make the following bootstrap assumption in the region $\mathcal{Q}_{u_{0},v_{0}}^{(ex)}$ for some universal constant $B$:
\begin{align}
&\vert \widetilde{\nu}\vert\leq Bq^{2}\left(\frac{v}{(-u)^{1-q^{2}}}\right)^{1-\delta},\quad \vert\widetilde{\lambda}\vert\leq Bq^{2}(-u)^{q^{2}}\left(\frac{v}{(-u)^{1-q^{2}}}\right)^{1-2\delta},\quad \vert\widetilde{\mu}\vert\leq Bq^{2}\left(\frac{v}{(-u)^{1-q^{2}}}\right)^{1-\delta}, \label{bootstrap 1:label}\\&
\vert\widetilde{\partial_{u}\phi_{1}}\vert\leq Bq(-u)^{-1}\left(\frac{v}{(-u)^{1-q^{2}}}\right)^{1-\delta},\quad \left\vert\widetilde{\partial_{u}\phi_{2}}\right\vert\leq Bq(-u)^{-1}\left(\frac{v}{(-u)^{1-q^{2}}}\right)^{1-\delta},\label{bootstrap2:label}\\&\vert\widetilde{\partial_{v}\phi_{1}}\vert\leq B q(-u)^{q^{2}-1}\left(\frac{v}{(-u)^{1-q^{2}}}\right)^{1-2\delta},\quad \left\vert\widetilde{\partial_{v}\phi_{2}}\right\vert\leq Bq(-u)^{q^{2}-1}\left(\frac{v}{(-u)^{1-q^{2}}}\right)^{1-2\delta},\label{bootstrap3:label}\\&
\left\vert \widetilde{Q}\right\vert\leq B\vert q_{0}\vert q^{2}u^{2}\left\vert\log(-u)\right\vert\left(\frac{v}{(-u)^{1-q^{2}}}\right)^{1-\delta}.\label{bootstrap 4}
\end{align}
The following proposition is a consequence of the bootstrap assumption \eqref{bootstrap 1:label}-\eqref{bootstrap 4}.
\begin{proposition}
\label{prop:consequence of the bootstrap}
Under the bootstrap assumption \eqref{bootstrap 1:label}-\eqref{bootstrap 4}, for $q$ and $\underline{v}$ sufficiently small, we have the following estimates in the region $\mathcal{Q}_{u_{0},v_{0}}^{(ex)}$:\begin{align}
&-\frac{3}{2}<\nu<-\frac{1}{2},\quad \lambda\approx (-u)^{q^{2}},\quad r\approx (-u),\quad \vert\mu\vert\leq 2q^{2},\\&
\vert\partial_{u}\phi_{1}\vert\leq\frac{2q}{(-u)},\quad\left\vert\partial_{u}\phi_{2}\right\vert\leq \frac{2q}{(-u)},\quad \vert \phi\vert\leq 2q\vert\log(-u)\vert, \\&
\left\vert Q\right\vert\lesssim \vert q_{0}\vert q^{2}u^{2}\left\vert\log(-u)\right\vert,\quad \left\vert A_{u}\right\vert\lesssim \vert q_{0}\vert q^{2}(-u)\left\vert\log(-u)\right\vert\frac{v}{(-u)^{1-q^{2}}}.
\end{align}
\end{proposition}
\subsubsection{Closing the bootstrap argument}
To show that the solution can be extended to the whole region $\Rmnum{1}$, it suffices to show that for $q$ and $\underline{v}$ sufficiently small, we can improve the bootstrap estimate \eqref{bootstrap 1:label}-\eqref{bootstrap 4}.
\begin{proposition}
Under the bootstrap assumption \eqref{bootstrap 1:label}-\eqref{bootstrap 4} in the region $\mathcal{Q}_{u_{0},v_{0}}$, for $0<q\ll\delta\ll1$, we have\begin{align}
&\vert \widetilde{\nu}\vert\leq \frac{1}{2}Bq^{2}\left(\frac{v}{(-u)^{1-q^{2}}}\right)^{1-\delta},\quad \vert\widetilde{\lambda}\vert\leq \frac{1}{2}Bq^{2}(-u)^{q^{2}}\left(\frac{v}{(-u)^{1-q^{2}}}\right)^{1-2\delta},\quad \vert\widetilde\mu\vert\leq\frac{1}{2}Bq^{2}\left(\frac{v}{(-u)^{1-q^{2}}}\right)^{1-\delta},\\&
\vert\widetilde{\partial_{u}\phi}\vert\leq \frac{1}{2}Bq(-u)^{-1}\left(\frac{v}{(-u)^{1-q^{2}}}\right)^{1-\delta},\quad \vert\widetilde{\partial_{v}\phi}\vert\leq \frac{1}{2}Bq(-u)^{q^{2}-1}\left(\frac{v}{(-u)^{1-q^{2}}}\right)^{1-2\delta},\\&
\left\vert\widetilde{Q}\right\vert\leq\frac{1}{2}B\vert q_{0}\vert q^{2}u^{2}\vert\log(-u)\vert\left(\frac{v}{(-u)^{1-q^{2}}}\right)^{1-\delta},
\end{align}
where $B$ is the constant in \eqref{bootstrap 1:label}-\eqref{bootstrap 4}.
\label{proposition:close the bootstrap assumption in the region 1}
\end{proposition}
\begin{proof}
To estimate $\widetilde{\nu}$, we conjugate the equation \eqref{eq:renormalized ru equation} with $\left(\frac{v}{(-u)^{1-q^{2}}}\right)^{-2+2\delta}\widetilde{\nu}$. We have\begin{equation}
\begin{aligned}
\frac{1}{2}\partial_{v}\left(\left(\frac{v}{(-u)^{1-q^{2}}}\right)^{-2+2\delta}\widetilde{\nu}^{2}\right)+&\left(\frac{1-\delta}{v}-\frac{\lambda}{(1-\mu)r}\left(\mu-\frac{Q^{2}}{r^{2}}\right)\right)\left(\frac{v}{(-u)^{1-q^{2}}}\right)^{-2+2\delta}\widetilde\nu^{2}\\&
= \frac{\lambda}{(1-\mu)r}\left(\mu-\frac{Q^{2}}{r^{2}}\right)\mathcal{P}(\nu)\left(\frac{v}{(-u)^{1-q^{2}}}\right)^{-2+2\delta}\widetilde\nu.
\end{aligned}
\label{conjugate equation for ru}
\end{equation}
By Proposition \ref{prop:consequence of the bootstrap}, we have\begin{align*}
\frac{1-\delta}{v}-\frac{\lambda}{(1-\mu)r}\left(\mu-\frac{Q^{2}}{r^{2}}\right)&= \frac{1}{v}\left(1-\delta-v\frac{\lambda}{(1-\mu)r}\left(\mu-\frac{Q^{2}}{r^{2}}\right)\right)\\&\geq
\frac{1}{v}\left(1-\delta-C_{0}q^{2}\left(\frac{v}{(-u)^{1-q^{2}}}\right)\right)\\&\geq \frac{1}{v}\left(1-\delta-\underline{v}\right)\geq 0,
\end{align*}
provided that $\underline{v}$ and $q$ are sufficiently small. Hence, we can drop the second term on the left-hand side of the equation \eqref{conjugate equation for ru}. Then we have\begin{equation}
\begin{aligned}
\frac{1}{2}\left(\frac{v}{(-u)^{1-q^{2}}}\right)^{-2+2\delta}\widetilde{\nu}^{2}(u,v)&\leq \int_{0}^{v}\left\vert\frac{\lambda}{(1-\mu)r}\left(\mu-\frac{Q^{2}}{r^{2}}\right)\mathcal{P}(\nu)\left(\frac{v^{\prime}}{(-u)^{1-q^{2}}}\right)^{-2+2\delta}\widetilde{\nu}\right\vert dv^{\prime}\\&\lesssim\frac{q^{4}}{\delta}\left(\frac{v}{(-u)^{1-q^{2}}}\right)^{\delta}\lesssim \frac{q^{4}}{\delta}\underline{v}^{\delta}.
\end{aligned}
\end{equation}
By the assumption that $\underline{v}\ll\delta\ll1$, we have\begin{equation}
\vert\widetilde\nu\vert\leq \frac{1}{2}Bq^{2}\left(\frac{v}{(-u)^{1-q^{2}}}\right)^{1-\delta}.
\end{equation}
We only provide the proof of the improved bootstrap bound for $\widetilde{\partial_{u}\phi_{1}}$, as the estimate for $\widetilde{\partial_{u}\phi_{2}}$ will follow similarly. To estimate $\widetilde{\partial_{u}\phi_{1}}$, since $\frac{\lambda}{r}$ is already positive, we directly conjugate the equation \eqref{eq:renormalized du-phi1 equation} with $\widetilde{\partial_{u}\phi}$ and integrate over $[0,v]$. We have\begin{equation}
\begin{aligned}
\left(\widetilde{\partial_{u}\phi_{1}}\right)^{2}(u,v)\lesssim&\int_{0}^{v}\left\vert\frac{\lambda}{r}\mathcal{P}(\partial_{u}\phi_{1})\widetilde{\partial_{u}\phi_{1}}\right\vert+\left\vert\frac{\nu}{r}\widetilde{\partial_{v}\phi_{1}}\widetilde{\partial_{u}\phi_{1}}\right\vert+\left\vert\frac{\nu}{r}\left(\partial_{v}\phi_{1}\right)_{S}\widetilde{\partial_{u}\phi_{1}}\right\vert\\&+\left\vert q_{0}\right\vert\int_{0}^{v}\left\vert A_{u}\widetilde{\partial_{v}\phi_{2}}\widetilde{\partial_{u}\phi_{1}}\right\vert+\left\vert A_{u}\left(\partial_{v}\phi_{2}\right)_{S}\widetilde{\partial_{u}\phi_{1}}\right\vert+\left\vert A_{u}\frac{\lambda}{r}+\frac{Q}{r^{2}}\frac{\lambda\nu}{1-\mu}\right\vert\left\vert\phi_{2}\widetilde{\partial_{u}\phi_{1}}\right\vert.
\end{aligned}
\label{conjugate equation for du phi}
\end{equation}
To estimate the first term on the right-hand side of \eqref{conjugate equation for du phi}, we have \begin{align*}
\int_{0}^{v}\left\vert\frac{\lambda}{r}\mathcal{P}(\partial_{u}\phi_{1})\widetilde{\partial_{u}\phi_{1}}\right\vert dv^{\prime}&\lesssim  q^{2}(-u)^{-2}\int_{0}^{v}\frac{1}{(-u)^{1-q^{2}}}\left(\frac{v^{\prime}}{(-u)^{1-q^{2}}}\right)^{1-\delta}dv^{\prime}\\&\lesssim q^{2}(-u)^{-2}\left(\frac{v}{(-u)^{1-q^{2}}}\right)^{2-2\delta}\underline{v}^{\delta}.
\end{align*}
Using the fact that $\underline{v}\ll\delta\ll1$, we have \begin{equation*}
    \int_{0}^{v}\left\vert\frac{\lambda}{r}\mathcal{P}(\partial_{u}\phi_{1})\widetilde{\partial_{u}\phi_{1}}\right\vert dv^{\prime}\leq \frac{B^{2}q^{2}}{100}(-u)^{-2}\left(\frac{v}{(-u)^{1-q^{2}}}\right)^{2-2\delta}.
\end{equation*}
Similarly, using the bootstrap assumptions \eqref{bootstrap 1:label}-\eqref{bootstrap 4} and Proposition \ref{prop:consequence of the bootstrap}, we have \begin{align*}
    \int_{0}^{v}\left\vert\frac{\nu}{r}\widetilde{\partial_{v}\phi_{1}}\widetilde{\partial_{u}\phi_{1}}\right\vert\leq&\frac{B^{2}q^{2}}{100}(-u)^{-2}\left(\frac{v}{(-u)^{1-q^{2}}}\right)^{2-2\delta},\\
    \left\vert q_{0}\right\vert\int_{0}^{v}\left\vert A_{u}\widetilde{\partial_{v}\phi_{2}}\widetilde{\partial_{u}\phi_{1}}\right\vert\leq&\frac{B^{2}q^{2}}{100}(-u)^{-2}\left(\frac{v}{(-u)^{1-q^{2}}}\right)^{2-2\delta},\\
    \left\vert q_{0}\right\vert\int_{0}^{v}\left\vert A_{u}\frac{\lambda}{r}+\frac{Q}{r^{2}}\frac{\lambda\nu}{1-\mu}\right\vert\left\vert\phi_{2}\widetilde{\partial_{u}\phi_{1}}\right\vert\leq&\frac{B^{2}q^{2}}{100}(-u)^{-2}\left(\frac{v}{(-u)^{1-q^{2}}}\right)^{2-2\delta}.
\end{align*}

For the third term on the right-hand side of \eqref{conjugate equation for du phi}, we have\begin{align*}
&\int_{0}^{v}\left\vert\frac{\lambda}{r}\left(\partial_{v}\phi_{1}\right)_{S}\widetilde{\partial_{u}\phi_{1}}\right\vert \\\lesssim&\frac{q}{(-u)^{2}}\int_{0}^{v}\left(\frac{v^{\prime}}{(-u)^{1-q^{2}}}\right)^{1-\delta}\left\vert\left(\partial_{v}\phi_{1}\right)_{S}\right\vert dv^{\prime}\\\lesssim&
\frac{1}{(-u)^{2}}\int_{0}^{v}(-u)^{q^{2}-1}\left(\frac{v^{\prime}}{(-u)^{1-q^{2}}}\right)^{1-\delta}\left[1-\left(\frac{v^{\prime}}{(-u)^{1-q^{2}}}\right)^{\frac{q^{2}}{1-q^{2}}}\left(\frac{1}{\underline{v}}\right)^{\frac{q^{2}}{1-q^{2}}}\right]dv^{\prime}\\\lesssim&\frac{1}{u^{2}}\left(\frac{v}{(-u)^{1-q^{2}}}\right)^{2-2\delta}\int_{0}^{(-u)^{1-q^{2}}\underline{v}}\left(\frac{v^{\prime}}{(-u)^{1-q^{2}}}\right)^{-1+\delta}\left[1-\left(\frac{v^{\prime}}{(-u)^{1-q^{2}}}\right)^{\frac{q^{2}}{1-q^{2}}}\left(\frac{1}{\underline{v}}\right)^{\frac{q^{2}}{1-q^{2}}}\right] \frac{dv^{\prime}}{(-u)^{1-q^{2}}}\\\lesssim&\frac{q^{2}}{\delta((1-q^{2})\delta+q^{2})}\underline{v}^{\delta}\frac{1}{u^{2}}\left(\frac{v}{(-u)^{1-q^{2}}}\right)^{2-2\delta}.
\end{align*}
where we have used the estimate \eqref{eq: pointwise estimate for auxiliary function}. Hence, for $\underline{v}\ll\delta\ll1$, we have \begin{equation*}
\int_{0}^{v}\left\vert\frac{\lambda}{r}\left(\partial_{v}\phi_{1}\right)_{S}\widetilde{\partial_{u}\phi_{1}}\right\vert\leq \frac{B^{2}q^{2}}{100}\frac{1}{u^{2}}\left(\frac{v}{(-u)^{1-q^{2}}}\right)^{2-2\delta}.
\end{equation*}
Similarly, we have \begin{equation*}
    \left\vert q_{0}\right\vert\int_{0}^{v}\left\vert A_{u}\left(\partial_{v}\phi_{2}\right)_{S}\widetilde{\partial_{u}\phi_{1}}\right\vert\leq\frac{B^{2}q^{2}}{100}(-u)^{-2}\left(\frac{v}{(-u)^{1-q^{2}}}\right)^{2-2\delta}.
\end{equation*}
Putting everything together, we have \begin{equation*}
    \left\vert\widetilde{\partial_{u}\phi_{1}}\right\vert\leq \frac{1}{2}\frac{Bq}{(-u)}\left(\frac{v}{(-u)^{1-q^{2}}}\right)^{1-\delta}.
\end{equation*}
To estimate $\widetilde{Q}$, we have
\begin{equation}
\begin{aligned}
\left\vert\widetilde{Q}\right\vert\lesssim &\vert q_{0}\vert u^{2}\int_{0}^{v}\left\vert\phi_{1}\partial_{v}\phi_{2}\right\vert+\left\vert\phi_{2}\partial_{v}\phi_{1}\right\vert dv^{\prime}\\\lesssim& \vert q_{0}\vert u^{2}\int_{0}^{v}\left\vert\phi_{1}\widetilde{\partial_{v}\phi_{2}}\right\vert+\left\vert\phi_{2}\widetilde{\partial_{v}\phi_{1}}\right\vert+\left\vert\phi_{1}\left(\partial_{v}\phi_{2}\right)_{S}\right\vert+\left\vert\phi_{2}\left(\partial_{v}\phi_{1}\right)_{S}\right\vert dv^{\prime}.
\end{aligned}
\label{eq: first equation for the charge in region 1}
\end{equation}
For the first term on the right-hand side of \eqref{eq: first equation for the charge in region 1}, we have \begin{align*}
    u^{2}\int_{0}^{v}\left\vert\phi_{1}\widetilde{\partial_{v}\phi_{2}}\right\vert dv^{\prime}\lesssim &q^{2}u^{2}\int_{0}^{v}\left\vert\log(-u)\right\vert(-u)^{q^{2}-1}\left(\frac{v^{\prime}}{(-u)^{1-q^{2}}}\right)^{1-2\delta} dv^{\prime}\\\lesssim& q^{2}u^{2}\left\vert\log(-u)\right\vert\left(\frac{v}{(-u)^{1-q^{2}}}\right)^{2-2\delta}\\\lesssim&q^{2}u^{2}\left\vert\log(-u)\right\vert\left(\frac{v}{(-u)^{1-q^{2}}}\right)^{1-\delta}\underline{v}^{1-\delta}.
\end{align*}
Similarly, we have \begin{equation*}
    u^{2}\int_{0}^{v}\left\vert\phi_{2}\widetilde{\partial_{v}\phi_{1}}\right\vert\lesssim q^{2}u^{2}\left\vert\log(-u)\right\vert\left(\frac{v}{(-u)^{1-q^{2}}}\right)^{1-\delta}\underline{v}^{1-\delta}.
\end{equation*}
For the last term on the right-hand side of \eqref{eq: first equation for the charge in region 1}, using \eqref{eq: pointwise estimate for auxiliary function}, we have \begin{align*}
    u^{2}\int_{0}^{v}\left\vert\phi_{2}\left(\partial_{v}\phi_{1}\right)_{S}\right\vert dv^{\prime}\lesssim &(-u)^{2}\left\vert\log(-u)\right\vert\frac{1}{(-u)^{1-q^{2}}}\int_{0}^{v}1-\left(\frac{v}{(-u)^{1-q^{2}}}\right)^{\frac{q^{2}}{1-q^{2}}}\left(\frac{1}{\underline{v}}\right)^{\frac{q^{2}}{1-q^{2}}}dv^{\prime}\\\lesssim&(-u)^{2}\left\vert\log(-u)\right\vert\frac{v^{1-\delta}}{(-u)^{1-q^{2}}}\int_{0}^{(-u)^{1-q^{2}}\underline{v}}v^{-1+\delta}-\frac{v^{\frac{q^{2}}{1-q^{2}}+\delta-1}}{(-u)^{q^{2}}}\left(\frac{1}{\underline{v}}\right)^{\frac{q^{2}}{1-q^{2}}}dv\\\lesssim&\frac{q^{2}}{\delta^{2}}\underline{v}^{\delta}(-u)^{2}\left\vert\log(-u)\right\vert\left(\frac{v}{(-u)^{1-q^{2}}}\right)^{1-\delta}.
\end{align*}
Similarly, we have \begin{equation*}
    u^{2}\int_{0}^{v}\left\vert\phi_{1}\left(\partial_{v}\phi_{2}\right)_{S}\right\vert dv^{\prime}\lesssim \frac{q^{2}}{\delta^{2}}\underline{v}^{\delta} u^{2}\left\vert\log(-u)\right\vert\left(\frac{v}{(-u)^{1-q^{2}}}\right)^{1-\delta}.
\end{equation*}
Putting everything together, we have \begin{equation*}
    \left\vert\widetilde{Q}\right\vert\leq\frac{1}{2}B\vert q_{0}\vert q^{2}u^{2}\left\vert\log(-u)\right\vert.
\end{equation*}

To estimate $\widetilde{\lambda}$, since the coefficient before the zero-order term of $\widetilde{\lambda}$ in \eqref{eq:renormalized rv equation} already has a favorable sign when $q$ is sufficiently small, we directly conjugate the equation with $\widetilde{\lambda}$. Then we have\begin{equation}
\begin{aligned}
\widetilde{\lambda}^{2}(u,v)\lesssim& \int_{-e^{-1}}^{u}\left\vert\frac{\nu}{r(1-\mu)}\left(\mu-\frac{Q^{2}}{r^{2}}\right)-\mathcal{P}\left(\frac{\nu}{r(1-\mu)}\left(\mu-\frac{Q^{2}}{r^{2}}\right)\right)\right\vert\mathcal{P}(\lambda)\left\vert\widetilde{\lambda}\right\vert du^{\prime}\\\lesssim&q^{4}\int_{-e^{-1}}^{u}v^{2-3\delta}(-u^{\prime})^{-(1-q^{2})(2-3\delta)+2q^{2}-1}du^{\prime}\\\lesssim& q^{4}(-u)^{2q^{2}}\left(\frac{v}{(-u)^{1-q^{2}}}\right)^{2-4\delta}\underline{v}^{\delta}.
\end{aligned}
\end{equation}
Hence we have \begin{equation}
\left\vert\widetilde{\lambda}\right\vert\leq\frac{1}{2}Bq^{2}(-u)^{q^{2}}\left(\frac{v}{(-u)^{1-q^{2}}}\right)^{1-2\delta}.
\end{equation}
Next, we improve the bootstrap bounds for $\widetilde{\partial_{v}\phi_{1}}$ and $\widetilde{\partial_{v}\phi_{2}}$. To estimate $\widetilde{\partial_{v}\phi_{1}}$, we conjugate the equation \eqref{eq:renormalized dv-phi1 equation} with $$(-u)^{2-2q^{2}}\left(\frac{v}{(-u)^{1-q^{2}}}\right)^{-2+4\delta}\widetilde{\partial_{v}\phi_{1}}.$$ Then we have\begin{equation}
\begin{aligned}
&\frac{1}{2}\partial_{u}\left((-u)^{2-2q^{2}}\left(\frac{v}{(-u)^{1-q^{2}}}\right)^{-2+4\delta}\left(\widetilde{\partial_{v}\phi_{1}}\right)^{2}\right)\\&+\left((2-2\delta)(1-q^{2})\frac{1}{(-u)}+\frac{\nu}{r}\right)(-u)^{2-2q^{2}}\left(\frac{v}{(-u)^{1-q^{2}}}\right)^{-2+4\delta}\left(\widetilde{\partial_{v}\phi_{1}}\right)^{2} \\=& -(-u)^{2-2q^{2}}\left(\frac{v}{(-u)^{1-q^{2}}}\right)^{-2+4\delta}\widetilde{\left(\frac{\nu}{r}\right)}\left(\partial_{v}\phi_{1}\right)_{S}\widetilde{\partial_{v}\phi_{1}}-(-u)^{2-2q^{2}}\left(\frac{v}{(-u)^{1-q^{2}}}\right)^{-2+4\delta}\widetilde{\left(\frac{\lambda}{r}\partial_{u}\phi_{1}\right)}\widetilde{\partial_{v}\phi_{1}}\\&+q_{0}(-u)^{2-2q^{2}}\left(\frac{v}{(-u)^{1-q^{2}}}\right)^{-2+4\delta}\left(\frac{A_{u}}{r}\partial_{v}(r\phi_{2})-\widetilde{\left(\frac{Q\Omega^{2}\phi_{2}}{4r^{2}}\right)}\right)\widetilde{\partial_{v}\phi_{1}}.
\end{aligned}
\label{eq: differential equation for dvphi1 after conjugation}
\end{equation}
Since the coefficient\begin{equation*}
\frac{(2-2\delta)(1-q^{2})}{(-u)}+\frac{\nu}{r}>0
\end{equation*}
by the consequence of the bootstrap assumption \eqref{bootstrap 1:label}-\eqref{bootstrap 4},  we have\begin{equation}
\begin{aligned}
(-u)^{2-2q^{2}}\left(\frac{v}{(-u)^{1-q^{2}}}\right)^{-2+4\delta}\left(\widetilde{\partial_{v}\phi_{1}}\right)^{2}(u,v)\lesssim&v^{-2+4\delta}\left(\widetilde{\partial_{v}\phi_{1}}\right)^{2}(-e^{-1},v)\\&+\int_{-e^{-1}}^{u}\left\vert\text{RHS of \eqref{eq: differential equation for dvphi1 after conjugation}}\right\vert.
\end{aligned}
\label{conjugate equation for dv phi}
\end{equation}
For the right-hand side of \eqref{eq: differential equation for dvphi1 after conjugation}, we have \begin{equation}
    \left\vert\text{RHS of \eqref{eq: differential equation for dvphi1 after conjugation}}\right\vert\lesssim  q^{2}(-u)^{-1}\left(\frac{v}{(-u)^{1-q^{2}}}\right)^{\delta}+q(-u)^{-q^{2}}\left(\frac{v}{(-u)^{1-q^{2}}}\right)^{2\delta}\left\vert\left(\partial_{v}\phi_{1}\right)_{S}\right\vert. \label{eq: estimate for the rhs of dvphi1 estimate}
\end{equation}
For the first term on the right-hand side of \eqref{eq: estimate for the rhs of dvphi1 estimate}, we have \begin{equation*}
    \int_{-e^{-1}}^{u}q^{2}(-u^{\prime})^{-1}\left(\frac{v}{(-u^{\prime})^{1-q^{2}}}\right)^{\delta} du^{\prime}\lesssim q^{2}\underline{v}^{\delta}.
\end{equation*}
For the second term on the right-hand side of \eqref{eq: estimate for the rhs of dvphi1 estimate}, using \eqref{eq: pointwise estimate for auxiliary function}, we have \begin{align*}
    &q\int_{-e^{-1}}^{u}(-u^{\prime})^{-q^{2}}\left(\frac{v}{(-u^{\prime})^{1-q^{2}}}\right)^{2\delta}\left\vert(\partial_{v}\phi_{1})_{S}\right\vert du^{\prime}\\\lesssim& q\int_{-e^{-1}}^{u}(-u^{\prime})^{-1}\left(\frac{v}{(-u^{\prime})^{1-q^{2}}}\right)^{2\delta}(-u^{\prime})^{q^{2}-1}\left(1-\left(\frac{v}{(-u^{\prime})^{1-q^{2}}}\right)^{\frac{q^{2}}{1-q^{2}}}\left(\frac{1}{\underline{v}}\right)^{\frac{q^{2}}{1-q^{2}}}\right)du^{\prime}\\\lesssim&
    \frac{q^{2}}{2\delta(1-q^{2})\left(2\delta(1-q^{2})+q^{2}\right)}\left(\frac{v}{(-u)^{1-q^{2}}}\right)^{2\delta}\\\lesssim&
    \frac{q^{2}}{2\delta(1-q^{2})\left(2\delta(1-q^{2})+q^{2}\right)}\underline{v}^{2\delta}.
\end{align*}
Hence, we can conclude that \begin{equation*}
    \left\vert\widetilde{\partial_{v}\phi_{1}}\right\vert\leq \frac{Bq}{100}(-u)^{q^{2}-1}\left(\frac{v}{(-u)^{1-q^{2}}}\right)^{1-2\delta}.
\end{equation*}
It remains to estimate $\widetilde{\mu}$. Since the coefficient before the zeroth-order term $\widetilde{\mu}$ in \eqref{eq: renormalized dv-mu equation} is positive, we directly conjugate \eqref{eq: renormalized dv-mu equation} with $\widetilde{\mu}$. Then we have\begin{equation}
\begin{aligned}
\widetilde{\mu}^{2}(u,v)\lesssim \int_{0}^{v}\left\vert\left(\frac{\lambda}{r}+\frac{r}{\lambda}\left\vert\partial_{v}\phi\right\vert^{2}\right)\mathcal{P}(\mu)\widetilde{\mu}\right\vert+\left\vert\frac{r}{\lambda}\left\vert\partial_{v}\phi\right\vert^{2}\widetilde{\mu}\right\vert+\left\vert\frac{\lambda}{r^{3}}Q^{2}\widetilde{\mu}\right\vert dv^{\prime}.
\end{aligned}
\label{conjugate equation for mu}
\end{equation}
For the first term on the right-hand side of \eqref{conjugate equation for mu}, we have the estimate:\begin{align*}
&\int_{0}^{v}\left\vert\left(\frac{\lambda}{r}+\frac{r}{\lambda}\left\vert\partial_{v}\phi\right\vert^{2}\right)\mathcal{P}(\mu)\widetilde{\mu}\right\vert\\\lesssim&q^{4}\int_{0}^{v}\left(\frac{v^{\prime}}{(-u)^{1-q^{2}}}\right)^{1-\delta}\left((-u)^{q^{2}-1}+(-u)^{1-q^{2}}\left\vert\widetilde{\partial_{v}\phi}\right\vert^{2}+(-u)^{1-q^{2}}\left\vert\left(\partial_{v}\phi\right)_{S}\right\vert^{2}\right)\\\lesssim&q^{4}\left(\frac{v}{(-u)^{1-q^{2}}}\right)^{2-\delta}+q^{2}\int_{0}^{v}\left(\frac{v^{\prime}}{(-u)^{1-q^{2}}}\right)^{1-\delta}(-u)^{q^{2}-1}\left(1-\left(\frac{v^{\prime}}{(-u)^{1-q^{2}}}\right)^{\frac{q^{2}}{1-q^{2}}}\left(\frac{1}{\underline{v}}\right)^{\frac{q^{2}}{1-q^{2}}}\right)^{2}dv^{\prime}\\\lesssim&q^{2}\left(\frac{v}{(-u)^{1-q^{2}}}\right)^{2-2\delta}\int_{0}^{(-u)^{1-q^{2}}\underline{v}}\left(\frac{v^{\prime}}{(-u)^{1-q^{2}}}\right)^{-1+\delta}(-u)^{q^{2}-1}\left(1-\left(\frac{v^{\prime}}{(-u)^{1-q^{2}}}\right)^{\frac{q^{2}}{1-q^{2}}}\left(\frac{1}{\underline{v}}\right)^{\frac{q^{2}}{1-q^{2}}}\right)^{2}dv^{\prime}\\&+q^{4}\left(\frac{v}{(-u)^{1-q^{2}}}\right)^{2-\delta}\\\lesssim& q^{4}\left(\frac{v}{(-u)^{1-q^{2}}}\right)^{2-2\delta}\underline{v}^{\delta}+q^{2}\left(\frac{v}{(-u)^{1-q^{2}}}\right)^{2-2\delta}\underline{v}^{\delta}\left(\frac{1}{\delta}-\frac{2}{\delta+\frac{q^{2}}{1-q^{2}}}+\frac{1}{\delta+\frac{2q^{2}}{1-q^{2}}}\right)\\\lesssim& q^{4}\underline{v}^{\delta}\left(\frac{v}{(-u)^{1-q^{2}}}\right)^{2-2\delta}.
\end{align*}
For the second term on the right-hand side of \eqref{conjugate equation for mu}, similarly we have\begin{equation*}
\int_{0}^{v}\left\vert\frac{r}{\lambda}\left\vert\partial_{v}\phi\right\vert^{2}\widetilde{\mu}\right\vert\lesssim q^{4}\underline{v}^{\delta}\left(\frac{v}{(-u)^{1-q^{2}}}\right)^{2-2\delta}.
\end{equation*}
For the last term on the right-hand side of \eqref{conjugate equation for mu}, we have \begin{equation*}
    \int_{0}^{v}\left\vert\frac{\lambda}{r^{3}}Q^{2}\widetilde{\mu}\right\vert dv^{\prime}\lesssim q^{4}\underline{v}^{\delta}\left(\frac{v}{(-u)^{1-q^{2}}}\right)^{2-2\delta}.
\end{equation*}
Hence, we can conclude that \begin{equation*}
    \left\vert\widetilde{\mu}\right\vert\leq \frac{1}{2}Bq^{2}\left(\frac{v}{(-u)^{1-q^{2}}}\right)^{1-\delta}.
\end{equation*}
Putting everything together, we can conclude the proof of this proposition.
\end{proof}

\subsection{Construction of region II}
In this section, we construct the solution in the region $\Rmnum{2}$:$$\left\{\underline{v}\leq\frac{v}{(-u)^{1-q^{2}}}\leq \Lambda^{1-q^{2}}\right\},$$ where $\Lambda$ is a large number which will be determined later.\begin{figure}[htbp]
    \centering
    \begin{tikzpicture}[ scale=0.8, >=Latex, every node/.style={font=\small} ]
\coordinate (O) at (0,0);
\coordinate (A) at (4,-4);
\coordinate (B) at (6,6);
\coordinate (C) at (10,2);
\draw (O)--(A);
\draw[dashed] (O)--(B);
\draw (A)--(C);
\draw[dashed] (B) -- (C)
    node[midway, above, sloped]
    {\rotatebox{90}{$\mathscr{I}$}};
\node[left] at (O) { $\mathcal{O}$};
\coordinate (A1) at (6,-2);
\coordinate (A2) at (8,0);
\draw[dashed] (O)--(A1);
\draw[dashed] (O)--(A2);
\coordinate (A3) at (3,0);
\coordinate (A4) at (1.5,1.5);
\draw[dashed] (A3)--(A4);
\coordinate (L1) at (2.5,-1.5);
\coordinate (L2) at (3.5,-0.5);

\coordinate (L3) at (1.5,0.5);
\node at (L3) {III};
\coordinate (L4) at (5,3);
\node at (L4){IV};
\fill[gray!30] (O) -- (A1) -- (A2) -- cycle;
\node at (L1){I};
\node at (L2) {II};
\end{tikzpicture}
    \caption{$R_{II}$}
    \label{fig: region 2}
\end{figure} We mainly adopt the method in \cite{singh2024construction,rodnianski2023naked,cicortas2024discretely}. Since in this region we are away from the past null cone emanating from the naked singularity $(0,0)$, we transform back to the untwisted self-similar gauge:\begin{equation*}
V: = v^{\frac{1}{1-q^{2}}},\quad U: = u.
\end{equation*}
Then under the $(U,V)$ gauge, the region $\Rmnum{2}$ becomes: $$\{\underline{v}^{\frac{1}{1-q^{2}}}\leq\frac{V}{(-U)}\leq \Lambda\}.$$

By Proposition \ref{proposition:close the bootstrap assumption in the region 1}, in the $(U,V)$ gauge, we have the following proposition.
\begin{proposition}
On the cone $\left\{\frac{V}{(-U)} = \underline{v}^{\frac{1}{1-q^{2}}}\right\}$, we have\begin{align}
&\left\vert\partial_{U}r+1\right\vert\lesssim q^{2}\underline{v}^{1-\delta},\quad \left\vert\frac{1}{1-q^{2}}V^{q^{2}}\partial_{V}r(U,V)-\partial_{v}r(U,0)\right\vert\lesssim q^{2}\underline{v}^{1-2\delta}(-U)^{q^{2}},\\&\left\vert\mu(U,V)-\mu(u,0)\right\vert\lesssim q^{2}\underline{v}^{1-\delta},\quad 
\left\vert\partial_{U}\phi-\partial_{u}\phi(U,0)\right\vert\lesssim\frac{q}{(-U)}\underline{v}^{1-\delta},\quad \left\vert\partial_{V}\phi(U,V)\right\vert\lesssim \frac{q\underline{v}^{1-2\delta-\frac{q^{2}}{1-q^{2}}}}{(-U)},\\&
\left\vert Q(U,V)-Q(U,0)\right\vert\lesssim \vert q_{0}\vert q^{2}U^{2}\left\vert\log(-U)\right\vert\underline{v}^{1-\delta},\quad \left\vert A_{U}(U,V)\right\vert\lesssim  \vert q_{0}\vert q^{2}(-U)\left\vert\log(-U)\right\vert\underline{v}.
\end{align}
\end{proposition}
\begin{proof}
The estimate for $\partial_{V}\phi$ follows from the fact that $\left(\partial_{v}\phi\right)_{S} = 0$ on $\left\{\frac{v}{(-u)^{1-q^{2}}} = \underline{v}\right\}$. All the other estimates follow directly from Proposition \ref{proposition:close the bootstrap assumption in the region 1}.
\end{proof}

We define the region $$\mathcal{Q}_{u_{0},v_{0}}^{\Rmnum{2}}: = \mathcal{Q}_{u_{0},v_{0}}^{(ex)}\cap\left\{\underline{v}^{\frac{1}{1-q^{2}}}\leq\frac{V}{(-U)}\leq \Lambda\right\}.$$ 
To construct the spacetime in the region $\Rmnum{2}$, we make the following bootstrap assumption:\begin{align}
&e^{-K\frac{V}{(-U)}}\vert\partial_{U}r+1\vert\leq B_{1}q^{2-10\delta},\label{eq:bootstrap1 in the region 2}\quad 
e^{-K\frac{V}{(-U)}}\left\vert\frac{1}{1-q^{2}}V^{q^{2}}\partial_{V}r(U,V)-\partial_{v}r(U,0)\right\vert\leq B_{1}(-U)^{q^{2}}q^{2-10\delta},\\&
e^{-K\frac{V}{(-U)}}\vert\partial_{U}\phi\vert\leq \frac{B_{1}q^{1-2\delta}}{(-U)},\quad 
e^{-K\frac{V}{(-U)}}\left\vert\left(\frac{V}{(-U)}\right)^{q^{2}}\partial_{V}\phi\right\vert\leq \frac{B_{1}q^{1-2\delta}}{(-U)},\label{eq: bootstrap 2 in the region2}\\&
e^{-K\frac{V}{(-U)}}\vert\mu\vert\leq B_{1}q^{2-8\delta}.\label{eq:bootstrap5 in the region 2}
\end{align}
where $K$ is a large number such that $e^{K\Lambda}q^{\delta}\ll1$ (recall that in the previous section, we only assumed $\underline{v}\ll\delta\ll1$).

Immediately we have the following consequence of the bootstrap assumption \eqref{eq:bootstrap1 in the region 2}-\eqref{eq:bootstrap5 in the region 2}:
\begin{proposition}
\label{prop: consequence of the bootstrap in region 2}
In the region $\mathcal{Q}_{u_{0},v_{0}}^{\Rmnum{2}}$, under the bootstrap assumption \eqref{eq:bootstrap1 in the region 2}-\eqref{eq:bootstrap5 in the region 2}, there exists a constant $C$ such that\begin{align}
&-1-\frac{1}{2}Cq^{2-12\delta}\leq\partial_{U}r(U,V)\leq-1+2Cq^{2-12\delta},\label{eq: bootstrap consequence in the region 2 for dur}\\&
\frac{\partial_{v}r(U,0)}{2(-U)^{q^{2}}}\left(\frac{V}{(-U)}\right)^{-q^{2}}\leq\partial_{V}r(U,V)\leq \frac{2{\partial_{v}r}(U,0)}{(-U)^{q^{2}}}\left(\frac{V}{(-U)}\right)^{-q^{2}},\label{eq: bootstrap consequence in the region 2 for dvr}\\&
r\approx (-U)\left(1+\left(\frac{V}{(-U)}\right)^{1-q^{2}}\right),\label{eq: bootstrap consequence in the region2 for r}\quad 
\vert\mu\vert\leq q^{2-9\delta},\\&
\left\vert\phi\right\vert(U,V)\lesssim q^{1-3\delta}\left(\frac{V}{(-U)}\right)^{1-q^{2}}\left\vert\log(-U)\right\vert,\\& 
\left\vert Q(U,V)\right\vert\lesssim \vert q_{0}\vert q^{2-6\delta}\left(\frac{V}{(-U)}\right)^{2(1-q^{2})}\left(1+\left(\frac{V}{(-U)}\right)^{2(1-q^{2})}\right)U^{2}\left\vert\log(-U)\right\vert,\\& \left\vert A_{U}(U,V)\right\vert\lesssim \vert q_{0}\vert q^{2-6\delta}\left(\frac{V}{-U}\right)^{3(1-q^{2})}(-U)\left\vert\log(-U)\right\vert.
\end{align}
\end{proposition}
\begin{proof}
The estimates \eqref{eq: bootstrap consequence in the region 2 for dur}-\eqref{eq: bootstrap consequence in the region 2 for dvr} follow directly from our bootstrap assumption. For $r(U,V)$, we have \begin{align*}
    r(U,V) =& r(U,V)-r\left(U,(-U)\underline{v}^{\frac{1}{1-q^{2}}}\right)+r\left(U,(-U)\underline{v}^{\frac{1}{1-q^{2}}}\right)-r(U,0)+r(U,0)\\=&
    \int_{(-U)\underline{v}^{\frac{1}{1-q^{2}}}}^{V}\partial_{V}r(U,V^{\prime})dV^{\prime}+\int_{0}^{(-U)^{1-q^{2}}\underline{v}}\partial_{v}r(U,v^{\prime})dv^{\prime}+(-U).
\end{align*}
Then we have \begin{equation*}
   0\leq  r(U,V)-(-U)\leq (-U)\frac{2\partial_{v}r(U,0)}{(-U)^{q^{2}}}\left(\frac{V}{(-U)}\right)^{1-q^{2}}.
\end{equation*}
This proves \eqref{eq: bootstrap consequence in the region2 for r}. To estimate $\left\vert\phi\right\vert$, we have \begin{align*}
    \left\vert\phi\right\vert(U,V)\leq&\left\vert\phi(U,V)-\phi\left(U,(-U)\underline{v}^{\frac{1}{1-q^{2}}}\right)\right\vert+\left\vert\phi\left(U,(-U)\underline{v}^{\frac{1}{1-q^{2}}}\right)-\phi(U,0)\right\vert+\left\vert\phi(U,0)\right\vert\\\leq& \int_{(-U)\underline{v}^{\frac{1}{1-q^{2}}}}^{V}\left\vert\partial_{V}\phi\right\vert dV^{\prime}+\int_{0}^{(-U)^{1-q^{2}}\underline{v}}\left\vert\partial_{v}\phi\right\vert dv+q\left\vert\log(-U)\right\vert\\\leq&Cq^{1-3\delta}\left(\frac{V}{(-U)}\right)^{1-q^{2}}+q\left\vert\log(-U)\right\vert\\\leq&Cq^{1-3\delta}\left(\frac{V}{(-U)}\right)^{1-q^{2}}\left\vert\log(-U)\right\vert,
\end{align*}
where the third inequality is due to the bootstrap assumption \eqref{eq: bootstrap 2 in the region2} and the inequality \eqref{eq: L1 estimate for the auxilary function}. 
To estimate $Q$, we have \begin{align*}
    \left\vert Q\right\vert(U,V)\leq&\left\vert Q(U,V)-Q\left(U,(-U)\underline{v}^{\frac{1}{1-q^{2}}}\right)\right\vert+\left\vert Q\left(U,(-U)\underline{v}^{\frac{1}{1-q^{2}}}\right)-Q\left(U,0\right)\right\vert+\left\vert Q(U,0)\right\vert\\\lesssim&
    \vert q_{0}\vert\int_{(-U)\underline{v}^{\frac{1}{1-q^{2}}}}^{V}r^{2}\left\vert\phi\right\vert\left\vert\partial_{V}\phi\right\vert dV^{\prime}+\vert q_{0}\vert\int_{0}^{(-U)^{1-q^{2}}\underline{v}}r^{2}\left\vert \phi\right\vert\left\vert\partial_{v}\phi\right\vert dv^{\prime}+\vert q_{0}\vert q^{2}u^{2}\left\vert\log(-U)\right\vert\\\lesssim&\vert q_{0}\vert q^{2-6\delta}\left(\frac{V}{(-U)}\right)^{2(1-q^{2})}\left(1+\left(\frac{V}{(-U)}\right)^{2(1-q^{2})}\right)U^{2}\left\vert\log(-U)\right\vert.
\end{align*}
To estimate $A_{U}$, we have \begin{align*}
    \left\vert A_{U}(U,V)\right\vert =&\left\vert A_{U}(U,V)-A_{U}\left(U,(-U)\underline{v}^{\frac{1}{1-q^{2}}}\right)+A_{U}\left(U,(-U)\underline{v}^{\frac{1}{1-q^{2}}}\right)\right\vert\\\leq&
    \int_{(-U)\underline{v}^{\frac{1}{1-q^{2}}}}^{V}\left\vert\partial_{V}A_{U}(U,V^{\prime})\right\vert dV^{\prime}+\left\vert A_{u}\left(U,(-U)^{1-q^{2}}\underline{v}\right)\right\vert\\\leq &\int_{(-U)\underline{v}^{\frac{1}{1-q^{2}}}}^{V}\frac{2\left\vert Q\right\vert\left\vert\partial_{U}r\right\vert\left\vert\partial_{V}r\right\vert}{r^{2}(1-\mu)}dV^{\prime}+\int_{0}^{(-U)^{1-q^{2}}\underline{v}}\frac{2\left\vert Q\right\vert\left\vert\partial_{v}r\right\vert\left\vert\partial_{u}r\right\vert}{r^{2}(1-\mu)}dv^{\prime}\\\lesssim&\vert q_{0}\vert q^{2-6\delta}\left(\frac{V}{-U}\right)^{3(1-q^{2})}(-U)\left\vert\log(-U)\right\vert.
\end{align*}
\end{proof}

By the standard local existence result, the solution constructed in the region $\Rmnum{1}$ can be extended uniquely in the region $\Rmnum{2}$ if we can improve the above bootstrap assumptions \eqref{eq:bootstrap1 in the region 2}-\eqref{eq:bootstrap5 in the region 2}.\begin{proposition}
\label{prop: closing the bootstrap in region 2}
Assuming the bootstrap assumptions \eqref{eq:bootstrap1 in the region 2}-\eqref{eq:bootstrap5 in the region 2} hold in the region $\mathcal{Q}_{u_{0},v_{0}}^{\Rmnum{2}}$ for some constant $C$, we have the following estimates improving the original bootstrap bounds in $Q_{u_{0},v_{0}}^{\Rmnum{2}}$ \begin{align}
&e^{-K\frac{V}{(-U)}}\vert\partial_{U}r+1\vert\leq \frac{1}{2}B_{1}q^{2-10\delta},\quad
e^{-K\frac{V}{(-U)}}\left\vert\frac{1}{1-q^{2}}V^{q^{2}}\partial_{V}r(U,V)-\partial_{v}r(U,0)\right\vert\leq \frac{1}{2}B_{1}(-U)^{q^{2}}q^{2-10\delta},\\&
e^{-K\frac{V}{(-U)}}\vert\partial_{U}\phi\vert\leq \frac{B_{1}q^{1-2\delta}}{2(-U)},\quad 
e^{-K\frac{V}{(-U)}}\left\vert\left(\frac{V}{(-U)}\right)^{q^{2}}\partial_{V}\phi\right\vert\leq \frac{B_{1}q^{1-2\delta}}{2(-U)},\quad
e^{-K\frac{V}{(-U)}}\vert\mu\vert\leq \frac{1}{2}B_{1}q^{2-8\delta}.
\end{align}
\label{prop:close the bootstrap in the region 2}
\end{proposition}
\begin{proof}
The approach of the proof is similar to that used in the proof of Proposition \ref{proposition:close the bootstrap assumption in the region 1}. To estimate $\partial_{V}r$, let \begin{equation*}
    \partial_{low}\mathcal{Q}^{(ex)}_{\Rmnum{2}}: = \left\{\frac{V}{(-U)} = \underline{v}^{\frac{1}{1-q^{2}}},\ 0\leq V\leq\underline{v}^{\frac{1}{1-q^{2}}}\right\}\cup\left\{U = -e^{-1},\ \underline{v}^{\frac{1}{1-q^{2}}}\leq V\leq \Lambda\right\}
\end{equation*}
be the lower boundary of the region $\Rmnum{2}$. We denote the point on $\partial_{low}\mathcal{Q}^{(ex)}_{\Rmnum{2}}$ by $(b(V),V)$. 

To estimate $\partial_{V}r$, we write down the equation for $\frac{1}{1-q^{2}}V^{q^{2}}\partial_{V}r(U,V)-\partial_{v}r(U,0)$:
\begin{equation}
\begin{aligned}
&\partial_{U}\left(\frac{1}{1-q^{2}}V^{q^{2}}\partial_{V}r(U,V)-\partial_{v}r(U,0)\right)-\frac{\partial_{U}r}{r(1-\mu)}\left(\mu-\frac{Q^{2}}{r^{2}}\right)\left(\frac{1}{1-q^{2}}V^{q^{2}}\partial_{V}r(U,V)-\partial_{v}r(U,0)\right)\\=&\left(\frac{\partial_{U}r}{r(1-\mu)}\left(\mu-\frac{Q^{2}}{r^{2}}\right)-\mathcal{P}\left(\frac{\partial_{U}r}{r(1-\mu)}\left(\mu-\frac{Q^{2}}{r^{2}}\right)\right)\right)\partial_{v}r(U,0).
\end{aligned}
\label{eq:conjugate equation for dv r in the region 2}
\end{equation}
Conjugating the equation \eqref{eq:conjugate equation for dv r in the region 2} with $$e^{-2K\frac{V}{(-U)}}\left(\frac{1}{1-q^{2}}V^{q^{2}}\partial_{V}r-\partial_{v}r(U,0)\right),$$integrating from $\partial_{low}\mathcal{Q}^{(ex)}_{\Rmnum{2}}$, and using the smallness of $\frac{1}{1-q^{2}}V^{q^{2}}\partial_{V}r(U,V)-\partial_{v}r(U,0)$ on the boundary, we have
\begin{equation}
\begin{aligned}
&e^{-2K\frac{V}{(-U)}}\left\vert\frac{1}{1-q^{2}}V^{q^{2}}\partial_{V}r(U,V)-\partial_{v}r(U,0)\right\vert^{2}\allowdisplaybreaks\\\lesssim&
q^{4}\underline{v}^{2-4\delta}(-U)^{2q^{2}}+\int_{\partial_{low}\mathcal{Q}^{(ex)}_{\Rmnum{2}}}^{U}e^{-2K\frac{V}{(-U^{\prime})}}\left\vert\frac{\partial_{U}r}{r(1-\mu)}\left(\mu-\frac{Q^{2}}{r^{2}}\right)-\mathcal{P}\left(\frac{\partial_{U}r}{r(1-\mu)}\left(\mu-\frac{Q^{2}}{r^{2}}\right)\right)\right\vert\vert\mathcal{P}(\lambda)\vert\\&\hspace{5cm}\times\left\vert\frac{1}{1-q^{2}}V^{q^{2}}\partial_{V}r(U^{\prime},V)-\partial_{v}r(U^{\prime},0)\right\vert dU^{\prime}\allowdisplaybreaks\\\lesssim&q^{4}\underline{v}^{2-4\delta}(-U)^{2q^{2}}+q^{2-10\delta}\int_{\partial_{low}\mathcal{Q}^{(ex)}_{\Rmnum{2}}}^{U}(-U^{\prime})^{2q^{2}}e^{-K\frac{V}{(-U^{\prime})}}\left\vert\frac{\partial_{U}r}{r(1-\mu)}\left(\mu-\frac{Q^{2}}{r^{2}}\right)-\mathcal{P}\left(\frac{\partial_{U}r}{r(1-\mu)}\left(\mu-\frac{Q^{2}}{r^{2}}\right)\right) \right\vert\allowdisplaybreaks\\\lesssim& q^{4}\underline{v}^{2-4\delta}(-U)^{2q^{2}}+q^{4-19\delta}\Lambda^{8(1-q^{2})}V\int_{\partial_{low}\mathcal{Q}^{(ex)}_{\Rmnum{2}}}^{U}(-U^{\prime})^{2q^{2}-2}\left(\frac{V}{(-U^{\prime})}\right)^{-1}dU^{\prime}\allowdisplaybreaks\\\lesssim&q^{4}\underline{v}^{2-4\delta}(-U)^{2q^{2}}+q^{4-19\delta}\Lambda^{8(1-q^{2})}\frac{\Lambda}{\underline{v}}(-U)^{2q^{2}}.
\end{aligned}
\end{equation}
where the third inequality follows from the consequence of the bootstrap assumptions \eqref{eq:bootstrap1 in the region 2}-\eqref{eq:bootstrap5 in the region 2}. Hence, by taking $q$ sufficiently small, we have \begin{equation*}
    e^{-K\frac{V}{(-U)}}\left\vert\frac{1}{1-q^{2}}V^{q^{2}}\partial_{V}r(U,V)-\partial_{v}r(U,0)\right\vert\leq \frac{1}{2}B_{1}q^{2-10\delta}(-U)^{q^{2}}.
\end{equation*}

To estimate $\partial_{V}\phi$, conjugating with $e^{-2K\frac{V}{(-U)}}V^{2q^{2}}\overline{\partial_{V}\phi}$ and taking the real part, we have\begin{equation}
\begin{aligned}
&\frac{1}{2}\partial_{U}\left(e^{-2K\frac{V}{(-U)}}\left\vert V^{q^{2}}\partial_{V}\phi\right\vert^{2}\right)+\left(K\frac{V}{U^{2}}+\frac{\partial_{U}r}{r}\right)e^{-2K\frac{V}{(-U)}}\left\vert V^{q^{2}}\partial_{V}\phi\right\vert^{2} \\=& -e^{-2K\frac{V}{(-U)}}\frac{\partial_{V}r}{r}V^{q^{2}}\Re\left(\partial_{U}\phi\left({V^{q^{2}}\overline{\partial_{V}\phi}}\right)\right)+\frac{1}{r}\Re\left(iq_{0}V^{2q^{2}}\left(A_{U}\partial_{V}(r\phi)-\frac{Q\Omega^{2}}{4r}\phi\right)\overline{\partial_{V}\phi}\right).\label{eq:conjugate of dv phi in the region 2}
\end{aligned}
\end{equation}
Since $K$ is sufficiently large, we have\begin{equation*}
K\frac{V}{U^{2}}+\frac{\partial_{U}r}{r}\geq K\frac{V}{U^{2}}-\frac{2}{(-U)}\geq \left(K\underline{v}^{1-q^{2}}-2\right)\frac{1}{(-U)}>\frac{K}{2}\underline{v}^{1-q^{2}}\frac{1}{(-U)}.
\end{equation*}
Directly integrating \eqref{eq:conjugate of dv phi in the region 2} and using the smallness of the initial data on the outgoing null cone $\Sigma_{-e^{-1}}$, we have\begin{equation}
\begin{aligned}
&e^{-2K\frac{V}{(-U)}}\left\vert V^{q^{2}}\partial_{V}\phi\right\vert^{2}+K\underline{v}^{1-q^{2}}\int_{b(V)}^{U}\frac{1}{(-U^{\prime})}e^{-2K\frac{V}{(-U^{\prime})}}\left(V^{q^{2}}\partial_{V}\phi\right)^{2}dU^{\prime}\allowdisplaybreaks\\\leq&\frac{q^{2}}{(-U)^{2-2q^{2}}}+\int_{b(V)}^{U}\frac{V^{q^{2}}\partial_{V}r}{r}\left\vert\partial_{U}\phi\right\vert e^{-2K\frac{V}{(-U^{\prime})}}V^{q^{2}}\left\vert\partial_{V}\phi\right\vert dU^{\prime}\\&+\int_{b(V)}^{U}e^{-2K\frac{V}{(-U)^{\prime}}}\left(\left\vert \frac{A_{U}}{r}\partial_{V}(r\phi)\right\vert+\left\vert\frac{\partial_{U}r\partial_{V}r}{1-\mu}\frac{Q}{r^{2}}\phi\right\vert\right)V^{2q^{2}}\left\vert\partial_{V}\phi\right\vert dU^{\prime}\allowdisplaybreaks\\\leq&
\frac{q^{2-4\delta}}{(-U)^{2-2q^{2}}}q^{4\delta}+\frac{1}{4}K\underline{v}^{1-q^{2}}\int_{b(V)}^{U}\frac{1}{(-U^{\prime})}e^{-2K\frac{V}{(-U^{\prime})}}\left(V^{q^{2}}\partial_{V}\phi\right)^{2}dU^{\prime}\\&+\frac{1}{K}\underline{v}^{q^{2}-1}\int_{b(V)}^{U}(-U^{\prime})^{2q^{2}-1}e^{-2K\frac{V}{(-U^{\prime})}}\left(\partial_{U}\phi\right)^{2}dU^{\prime}\\&+\frac{1}{K}\underline{v}^{q^{2}-1}\int_{b(V)}^{U}V^{2q^{2}}(-U^{\prime})e^{-2K\frac{V}{(-U)}}\left(\left\vert\frac{A_{U}}{r}\partial_{V}(r\phi)\right\vert^{2}+\left\vert\frac{\partial_{U}r\partial_{V}r}{1-\mu}\frac{Q}{r^{2}}\phi\right\vert^{2}\right)\allowdisplaybreaks\\\leq&\frac{Cq^{2-4\delta}}{(-U)^{2-2q^{2}}}q^{4\delta}+\frac{C}{K}\underline{v}^{q^{2}-1}\frac{q^{2-4\delta}}{(-U)^{2-2q^{2}}}+\frac{1}{4}K\underline{v}^{1-q^{2}}\int_{b(V)}^{U}\frac{1}{(-U^{\prime})}e^{-2K\frac{V}{(-U^{\prime})}}\left(V^{q^{2}}\partial_{V}\phi\right)^{2}dU^{\prime}.
\end{aligned}
\end{equation}
By the largeness of $K$, we can absorb the third term on the right-hand side of the third inequality into the second term on the left. Then we have\begin{equation}
e^{-K\frac{V}{(-U)}}\left\vert\left(\frac{V}{(-U)}\right)^{q^{2}}\partial_{V}\phi\right\vert\leq\frac{B_{1}q^{1-2\delta}}{2(-U)}.
\end{equation}
To estimate $\partial_{U}r$, we have\begin{equation}
\begin{aligned}
&\partial_{V}\left(e^{-K\frac{V}{(-U)}}\left(\partial_{U}r+1\right)\right)+\left(\frac{K}{(-U)}-\frac{\partial_{V}r}{r(1-\mu)}\left(\mu-\frac{Q^{2}}{r^{2}}\right)\right)e^{-K\frac{V}{(-U)}}\left(\partial_{U}r+1\right)\\ =& -\frac{\partial_{V}r}{(1-\mu)r}\left(\mu-\frac{Q^{2}}{r^{2}}\right)e^{-K\frac{V}{(-U)}}.
\end{aligned}
\label{eq:conjugate equation of du r in the region 2}
\end{equation}
Since \begin{equation*}
\frac{K}{(-U)}-\frac{\mu}{1-\mu}\frac{\partial_{V}r}{r}\geq \frac{K}{(-U)}-\frac{Cq^{2-9\delta}\underline{v}^{-\frac{q^{2}}{1-q^{2}}}}{(-U)}>0,
\end{equation*}
multiplying \eqref{eq:conjugate equation of du r in the region 2} by $e^{-K\frac{V}{(-U)}}\left(\partial_{U}r+1\right)$ and integrating from $\frac{V}{(-U)} = \underline{v}^{\frac{1}{1-q^{2}}}$, we have
\begin{equation}
\begin{aligned}
e^{-2K\frac{V}{(-U)}}\left(\partial_{U}r+1\right)^{2}\leq&Cq^{4}\underline{v}^{2-2\delta}+\int_{(-U)\underline{v}^{\frac{1}{1-q^{2}}}}^{V}\left\vert\frac{\partial_{V}r}{(1-\mu)r}\left(\mu-\frac{Q^{2}}{r^{2}}\right)\right\vert e^{-2K\frac{V^{\prime}}{(-U)}}\left\vert\partial_{U}r+1\right\vert dV^{\prime}\\\leq&Cq^{4}\underline{v}^{2-2\delta}+C\int_{(-U)\underline{v}^{\frac{1}{1-q^{2}}}}^{V}q^{4-18\delta}\frac{(-U)^{q^{2}-1}}{(V^{\prime})^{q^{2}}}dV^{\prime}\\&+C\int_{(-U)\underline{v}^{\frac{1}{1-q^{2}}}}^{V}q^{6-22\delta}\frac{(-U)^{q^{2}-1}}{(V^{\prime})^{q^{2}}}\left(\frac{V^{\prime}}{(-U)}\right)^{2(1-q^{2})}\left(1+\left(\frac{V^{\prime}}{(-U)}\right)^{2(1-q^{2})}\right)dV^{\prime}\\\leq&Cq^{4}\underline{v}^{2-2\delta}+Cq^{4-18\delta}\Lambda^{1-q^{2}}+Cq^{6-22\delta}\Lambda^{5(1-q^{2})}.
\end{aligned}
\end{equation}
Therefore, we have \begin{equation*}
    e^{-K\frac{V}{(-U)}}\left\vert\partial_{U}r+1\right\vert\leq\frac{1}{2}B_{1}q^{2-10\delta}.
\end{equation*}
To estimate $\partial_{U}\phi$, we conjugate with $e^{-2K\frac{V}{(-U)}}\overline{\partial_{U}\phi}$ and take the real part:\begin{equation}
\begin{aligned}
&\frac{1}{2}\partial_{V}\left(e^{-2K\frac{V}{(-U)}}\left\vert\partial_{U}\phi\right\vert^{2}\right)+\left(\frac{K}{(-U)}+\frac{\partial_{V}r}{r}\right)e^{-2K\frac{V}{(-U)}}\left\vert\partial_{U}\phi\right\vert^{2} \\=& -\Re\left(\frac{\partial_{U}r\partial_{V}\phi}{r}e^{-2K\frac{V}{(-U)}}\overline{\partial_{U}\phi}\right)-\Re\left(iq_{0}\left(\frac{A_{U}}{r}\partial_{V}(r\phi)+\frac{\partial_{U}r\partial_{V}r}{(1-\mu)r}\frac{Q}{r^{2}}\phi\right)e^{-2K\frac{V}{(-U)}}\overline{\partial_{U}\phi}\right).
\end{aligned}\label{eq:conjugate equation for du-phi in the region 2}
\end{equation}
Directly integrating \eqref{eq:conjugate equation for du-phi in the region 2} and applying the Cauchy--Schwarz inequality, we have\begin{equation}
\begin{aligned}
&e^{-2K\frac{V}{(-U)}}\left\vert\partial_{U}\phi\right\vert^{2}+\frac{K}{(-U)}\int_{(-U)\underline{v}^{\frac{1}{1-q^{2}}}}^{V}e^{-2K\frac{V^{\prime}}{(-U)}}\vert\partial_{U}\phi\vert^{2}dV^{\prime}\\\lesssim&\frac{q^{2}}{(-U)^{2}}+\int_{(-U)\underline{v}^{\frac{1}{1-q^{2}}}}^{V}\left\vert\frac{\partial_{U}r\partial_{V}\phi}{r}\right\vert e^{-2K\frac{V^{\prime}}{(-U)}}\vert\partial_{U}\phi\vert dV^{\prime}\\&+\int^{V}_{(-U)\underline{v}^{\frac{1}{1-q^{2}}}}\left(\left\vert\frac{A_{U}}{r}\partial_{V}(r\phi)\right\vert+\left\vert\frac{\partial_{U}r\partial_{V}r}{(1-\mu)r}\frac{Q}{r^{2}}\phi\right\vert\right)e^{-2K\frac{V^{\prime}}{(-U)}}\left\vert\partial_{U}\phi\right\vert dV^{\prime}\\\lesssim&\frac{q^{2-4\delta}}{(-U)^{2}}q^{4\delta}+\frac{K}{2(-U)}\int_{(-U)\underline{v}^{\frac{1}{1-q^{2}}}}^{V}e^{-2K\frac{V^{\prime}}{(-U)}}\left\vert\partial_{U}\phi\right\vert^{2}dV^{\prime}+q^{2-4\delta}\frac{\Lambda^{2-2q^{2}}}{K}\frac{1}{U^{2}}.
\end{aligned}
\end{equation}
Hence, for $K$ sufficiently large, we have\begin{equation}
e^{-K\frac{V}{(-U)}}\left\vert\partial_{U}\phi\right\vert\leq\frac{B_{1}q^{1-2\delta}}{2(-U)}.
\end{equation}
It remains to close the estimate for $\mu$. We conjugate the equation \eqref{eq: v-transport equation for mu} with $e^{-2K\frac{V}{(-U)}}\mu$ and integrate:\begin{equation}
\begin{aligned}
e^{-2K\frac{V}{(-U)}}\mu^{2}\lesssim&q^{4}+\int_{(-U)\underline{v}^{\frac{1}{1-q^{2}}}}^{V}\left(\frac{r}{\partial_{V}r}\left\vert\partial_{V}\phi\right\vert^{2}+\frac{\partial_{V}r}{r^{3}}Q^{2}\right)e^{-2K\frac{V^{\prime}}{(-U)}}\vert\mu\vert dV^{\prime}\\\lesssim&
q^{4}+q^{4-12\delta}\Lambda^{2(1-q^{2})}.
\end{aligned}
\end{equation}
Hence, we can conclude the proof.
\end{proof}
\subsection{Construction of region III}
In this section, we consider the region:\begin{equation*}
R_{III}: = \left\{(U,V),\ -\frac{V}{\Lambda}\leq U<0,\ 0\leq V\leq \frac{1}{2}\right\}.
\end{equation*}
\begin{figure}[htbp]
    \centering
    \begin{tikzpicture}[ scale=0.8, >=Latex, every node/.style={font=\small} ]
\coordinate (O) at (0,0);
\coordinate (A) at (4,-4);
\coordinate (B) at (6,6);
\coordinate (C) at (10,2);
\draw (O)--(A);
\draw[dashed] (O)--(B);
\draw (A)--(C);
\draw[dashed] (B) -- (C)
    node[midway, above, sloped]
    {\rotatebox{90}{$\mathscr{I}$}};
\node[left] at (O) { $\mathcal{O}$};
\coordinate (A1) at (6,-2);
\coordinate (A2) at (8,0);
\draw[dashed] (O)--(A1);
\draw[dashed] (O)--(A2);
\coordinate (A3) at (3,0);
\coordinate (A4) at (1.5,1.5);
\draw[dashed] (A3)--(A4);
\coordinate (L1) at (2.5,-1.5);
\coordinate (L2) at (3.5,-0.5);

\coordinate (L3) at (1.5,0.5);

\coordinate (L4) at (5,3);
\node at (L4){IV};
\fill[gray!30] (O) -- (A4) -- (A3) -- cycle;
\node at (L1){I};
\node at (L2) {II};
\node at (L3) {III};
\end{tikzpicture}
    \caption{$R_{III}$}
    \label{fig: region 3}
\end{figure}
First, by Proposition \ref{prop:close the bootstrap in the region 2}, we have the following estimate:
\begin{proposition}
\label{prop: estimate on the cone of the right boundary of region 2}
On the hypersurface $\left\{\frac{V}{(-U)}=\Lambda\right\}$, we have\begin{align}
&r\approx V,\quad \left\vert\partial_{U}r+1\right\vert\leq q^{2-11\delta},\quad 
\left\vert\frac{1}{1-q^{2}}V^{q^{2}}\partial_{V}r(U,V)-\partial_{v}r(U,0)\right\vert\leq q^{2-11\delta}(-U)^{q^{2}},\\&\left\vert\phi\right\vert\lesssim q^{1-4\delta}\left\vert\log V\right\vert,\quad 
\left\vert\partial_{U}\phi\right\vert\leq \frac{q^{1-3\delta}}{(-U)},\quad
\left\vert\partial_{V}\phi\right\vert\leq\frac{q^{1-3\delta}}{(-U)},\quad 
\vert\mu\vert\leq q^{2-9\delta},\\&
\left\vert Q(U,V)\right\vert\leq\vert q_{0}\vert q^{2-7\delta}V^{2}\left\vert\log V\right\vert,\quad \left\vert A_{U}(U,V)\right\vert\leq \vert q_{0}\vert q^{2-7\delta}V\left\vert\log V\right\vert.
\end{align}
\end{proposition}
\begin{proof}
By Proposition \ref{prop: consequence of the bootstrap in region 2}, on the hypersurface $\left\{\frac{V}{(-U)} = \Lambda\right\}$, we have \begin{equation*}
    r\approx (-U)(1+\Lambda^{1-q^{2}}) = V\left(\Lambda^{-1}+\Lambda^{-q^{2}}\right)\approx V.
\end{equation*}
All the other estimates follow by an argument similar to that used in Proposition \ref{prop: closing the bootstrap in region 2}.
\end{proof}
In particular, the bound for $\partial_{V}r$ on $\left\{\frac{V}{(-U)} = \Lambda\right\}$ implies \begin{equation}
\begin{aligned}
    \left\vert\partial_{V}r-\frac{\partial_{v}r(U,0)}{(-U)^{q^{2}}}\right\vert\leq& \left\vert\partial_{V}r-(1-q^{2})\frac{\partial_{v}r(U,0)}{\Lambda^{q^{2}}(-U)^{q^{2}}}\right\vert+\left\vert\frac{1-q^{2}}{\Lambda^{q^{2}}}-1\right\vert\left\vert\frac{\partial_{v}r(U,0)}{(-U)^{q^{2}}}\right\vert\\\leq& q^{2-11\delta}(-U)^{q^{2}},
    \end{aligned}
\end{equation}
where we have used $ q\ll 1$.

Now we make the following bootstrap assumption in the region $\mathcal{Q}^{(ex)}_{u_{0},v_{0}}\cap R_{III}$ for some constant $0<p<\frac{1}{2}$ and $B_{2}>0$:
\begin{align}
&0<-\frac{\Lambda^{p}}{\left(\frac{V}{(-U)}\right)^{p}}\partial_{U}r\leq 2,\quad \label{eq:bootstrap 1 in the region 3}
\left\vert\partial_{V}r(U,V)-\frac{\partial_{v}r(-\frac{V}{\Lambda},0)}{(\frac{V}{\Lambda})^{q^{2}}}\right\vert\leq q^{2-12\delta},\\&
\left\vert V\partial_{V}\phi\right\vert\leq q^{1-4\delta},\quad 
\left\vert V\left(\frac{V}{(-U)}\right)^{-2p}\partial_{U}\phi\right\vert\leq q^{1-4\delta},\quad
\vert\mu\vert\leq q^{2-10\delta}.\label{eq:bootstrap 2 in the region 3}
\end{align}
We have the following consequences of the bootstrap assumption \eqref{eq:bootstrap 1 in the region 3}-\eqref{eq:bootstrap 2 in the region 3}:
\begin{proposition}
In the region $\mathcal{Q}_{u_{0},v_{0}}^{(ex)}\cap R_{\Rmnum{3}}$, under the bootstrap assumption \eqref{eq:bootstrap 1 in the region 3}-\eqref{eq:bootstrap 2 in the region 3}, we have:
\begin{align}
    &r\approx V,\quad (-\partial_{U}r)\leq\Lambda^{-p}\left(\frac{V}{(-U)}\right)^{p},\quad \partial_{V}r(U,V)\approx \frac{\partial_{v}r(-\frac{V}{\Lambda},0)}{(\frac{V}{\Lambda})^{q^{2}}}\approx 1,\\&
\left\vert\phi\right\vert\leq q^{1-5\delta}\left\vert\log V\right\vert,\quad \left\vert\partial_{V}\phi\right\vert\leq\frac{q^{1-4\delta}}{V},\quad \left\vert\partial_{U}\phi\right\vert\leq \frac{q^{1-4\delta}}{V}\left(\frac{V}{(-U)}\right)^{2p},\quad \left\vert\mu\right\vert\leq q^{2-10\delta},\\&
\left\vert Q(U,V)\right\vert\leq \vert q_{0}\vert q^{2-10\delta}V^{2}\left\vert\log V\right\vert,\quad \left\vert A_{U}\right\vert\leq \vert q_{0}\vert q^{2-10\delta}(-U)^{-p}V^{1+p}\left\vert\log V\right\vert.
\end{align}

\end{proposition}
\begin{proof}
The estimates for $\partial_{U}r$, $\partial_{V}r$, $\partial_{U}\phi$, $\partial_{V}\phi$, and $\mu$ follow directly from the bootstrap assumptions. For $r$, we have \begin{align*}
    r(U,V)=r\left(\frac{V}{-\Lambda},V\right)+\int_{\frac{V}{-\Lambda}}^{U}\partial_{U}r(U^{\prime},V)dU^{\prime}\approx V.
\end{align*}
To estimate $\phi$, we have \begin{align*}
    \left\vert\phi(U,V)\right\vert\leq&\left\vert\phi\left(\frac{V}{-\Lambda},V\right)\right\vert+\int_{\frac{V}{-\Lambda}}^{U}\left\vert\partial_{U}\phi\right\vert dU^{\prime}\\\leq&Cq^{1-4\delta}\left\vert\log V\right\vert+Cq^{1-4\delta}\left(-\left(\frac{V}{-U}\right)^{2p-1}+\Lambda^{2p-1}\right)\\\leq& q^{1-5\delta}\left\vert\log V\right\vert.
\end{align*}
For the spacetime charge $Q$, we have \begin{align*}
    \left\vert Q(U,V)\right\vert\leq \left\vert Q(U,\Lambda(-U))\right\vert+\int_{\Lambda(-U)}^{V}\vert q_{0}\vert r^{2}\left\vert\phi\right\vert\left\vert\partial_{V}\phi\right\vert dV^{\prime}\lesssim\vert q_{0}\vert q^{2-\frac{19}{2}\delta}V^{2}\left\vert\log V\right\vert.
\end{align*}
For $A_{U}$, we have \begin{align*}
    \left\vert A_{U}(U,V)\right\vert\leq& \left\vert A_{U}(U,\Lambda(-U))\right\vert+\int_{\Lambda(-U)}^{V}\left\vert\frac{2Q\partial_{U}r\partial_{V}r}{(1-\mu)r^{2}}\right\vert\\\leq& \vert q_{0}\vert q^{2-7\delta}(-U)\left\vert\log(-U)\right\vert+\vert q_{0}\vert q^{2-\frac{19}{2}\delta}V\vert\log V\vert\\\leq&\vert q_{0}\vert q^{2-10\delta}(-U)^{-p}V^{1+p}\vert\log V\vert.
\end{align*}
\end{proof}
By the standard local well-posedness result, we can close the construction in the region $\Rmnum{3}$ if we can prove the following proposition:
\begin{proposition}
\label{prop: estimate in the region 3}
Assume that \eqref{eq:bootstrap 1 in the region 3}-\eqref{eq:bootstrap 2 in the region 3} hold in the region $\mathcal{Q}_{u_{0},v_{0}}^{\Rmnum{3}}$ for some large constant $C$. Then in the same region, we have\begin{align}
&0<-\Lambda^{p}\left(\frac{V}{(-U)}\right)^{-p}\partial_{U}r\leq\frac{3}{2},\quad \left\vert\partial_{V}r(U,V)-\frac{\partial_{v}r(-\frac{V}{\Lambda},0)}{(\frac{V}{\Lambda})^{q^{2}}}\right\vert\leq\frac{1}{2}q^{2-12\delta},\\&\left\vert V\left(\frac{V}{(-U)}\right)^{-2p}\partial_{U}\phi\right\vert\leq\frac{1}{2}q^{1-4\delta},\quad
\left\vert V\partial_{V}\phi\right\vert\leq\frac{1}{2}q^{1-4\delta},\quad
\vert\mu\vert\leq\frac{1}{2}q^{2-10\delta}.
\end{align}
\end{proposition}
\begin{proof}
To estimate $\partial_{U}\phi$, conjugating with $\left(\frac{V}{(-U)}\right)^{-4p}r\partial_{U}\overline\phi$ and taking the real part, we have\begin{equation}
\label{eq:conjugate equation for dv phi in the region 3}
\begin{aligned}
&\frac{1}{2}\partial_{V}\left(\left(\frac{V}{(-U)}\right)^{-4p}\left\vert r\partial_{U}\phi\right\vert^{2}\right)+\frac{2p}{V}\left(\frac{V}{(-U)}\right)^{-4p}\left\vert r\partial_{U}\phi\right\vert^{2} \\\leq&\left(\frac{V}{(-U)}\right)^{-4p}\left\vert r\partial_{U}\phi\right\vert\left\vert\partial_{U}r\right\vert\left\vert\partial_{V}\phi\right\vert+\left\vert q_{0}\right\vert\left(\frac{V}{(-U)}\right)^{-4p}\left\vert r\partial_{U}\phi\right\vert\left(\left\vert A_{U}\right\vert\left\vert\partial_{V}(r\phi)\right\vert+\frac{\vert Q\vert}{r}\frac{\left\vert\partial_{U}r\partial_{V}r\right\vert}{(1-\mu)}\left\vert\phi\right\vert\right)
\end{aligned}
\end{equation}
Integrating \eqref{eq:conjugate equation for dv phi in the region 3}, we have\begin{equation}
\begin{aligned}
\left(\frac{V}{(-U)}\right)^{-4p}\left\vert r\partial_{U}\phi\right\vert^{2}\leq& \Lambda^{2-4p}q^{2-6\delta}+\int_{\Lambda(-U)}^{V}\left(\frac{V^{\prime}}{(-U)}\right)^{-4p}\left\vert r\partial_{U}\phi\right\vert\left\vert\partial_{U}r\right\vert\left\vert\partial_{V}\phi\right\vert dV^{\prime}
\\&+\vert q_{0}\vert\int_{\Lambda(-U)}^{V}\left(\frac{V^{\prime}}{(-U)}\right)^{-4p}\left\vert r\partial_{U}\phi\right\vert\left(\left\vert A_{U}\right\vert\left\vert\partial_{V}(r\phi)\right\vert+\frac{\vert Q\vert}{r}\frac{\left\vert\partial_{U}r\partial_{V}r\right\vert}{1-\mu}\vert \phi\vert\right)dV^{\prime}\\\leq&\Lambda^{2-4p}q^{2-6\delta}+C\Lambda^{-p}q^{2-8\delta}\int_{\Lambda(-U)}^{V}\left(\frac{V^{\prime}}{(-U)}\right)^{-p}\frac{1}{V^{\prime}}dV^{\prime}\\&+Cq^{2-8\delta}q^{2-10\delta}\int_{\Lambda(-U)}^{V}V^{\prime}\vert\log V^{\prime}\vert^{2}\left(\frac{V^{\prime}}{(-U)}\right)^{-p}dV^{\prime}\\\leq&\Lambda^{2-4p}q^{2-6\delta}+C\Lambda^{-p}q^{2-8\delta}\left(\Lambda^{-p}-\left(\frac{V}{(-U)}\right)^{-p}\right).
\end{aligned}
\end{equation}
Hence, we can choose a suitable $p$ such that $C\Lambda^{-p}\ll1$. Then we have
\begin{equation}
\left\vert V\left(\frac{V}{(-U)}\right)^{-2p}\partial_{U}\phi\right\vert\leq\frac{1}{2}q^{1-4\delta}.
\end{equation}
To estimate $\partial_{V}\phi$, conjugating with $r\partial_{V}\bar{\phi}$ and taking the real part, we have\begin{equation}
\frac{1}{2}\partial_{U}\left(\left\vert r\partial_{V}\phi\right\vert^{2}\right) \leq \left\vert\partial_{V}r\right\vert\left\vert r\partial_{V}\phi\right\vert\left\vert\partial_{U}\phi\right\vert+\vert q_{0}\vert\left\vert r\partial_{V}\phi\right\vert\left(\left\vert A_{U}\right\vert\left\vert\partial_{V}(r\phi)\right\vert+\frac{\vert Q\vert}{r}\frac{\left\vert\partial_{U}r\partial_{V}r\right\vert}{1-\mu}\vert\phi\vert\right).\label{eq:conjugate equation for dv-phi in the region 3}
\end{equation}
Integrating \eqref{eq:conjugate equation for dv-phi in the region 3}, we have\begin{equation}
\begin{aligned}
\left\vert V\partial_{V}\phi\right\vert^{2}\leq& 2\Lambda^{2}q^{2-8\delta}q^{2\delta}+2\int_{-\frac{V}{\Lambda}}^{U}\left\vert\partial_{V}r\right\vert\left\vert r\partial_{V}\phi\right\vert\left\vert\partial_{U}\phi\right\vert dU^{\prime}\\&+2\vert q_{0}\vert\int_{\frac{V}{-\Lambda}}^{U}\left\vert r\partial_{V}\phi\right\vert\left(\vert A_{U}\vert\left\vert\partial_{V}(r\phi)\right\vert+\frac{\vert Q\vert}{r}\frac{\left\vert\partial_{U}r\partial_{V}r\right\vert}{1-\mu}\vert\phi\vert\right)dU^{\prime}\\\leq&2\Lambda^{2}q^{2-8\delta}q^{2\delta}+Cq^{2-8\delta}\int_{-\frac{V}{\Lambda}}^{U}V^{2p-1}\left(-U^{\prime}\right)^{-2p}dU^{\prime}\\\leq&
2\Lambda^{2}q^{2-8\delta}q^{2\delta}+Cq^{2-8\delta}\Lambda^{2p-1}\\\leq& \frac{1}{4}q^{2-8\delta}.
\end{aligned}
\end{equation}
Hence, we can close the bootstrap estimate for $\partial_{V}\phi$.

To estimate $\partial_{U}r$, conjugating with $\Lambda^{p}\left(\frac{V}{(-U)}\right)^{-p}$, we have\begin{equation}
\partial_{V}\left(\Lambda^{p}\left(\frac{V}{(-U)}\right)^{-p}\partial_{U}r\right)+\frac{p}{V}\Lambda^{p}\left(\frac{V}{(-U)}\right)^{-p}\partial_{U}r= \Lambda^{p}\left(\frac{V}{(-U)}\right)^{-p}\frac{\partial_{U}r\partial_{V}r}{(1-\mu)r}\left(\mu-\frac{Q^{2}}{r^{2}}\right).
\end{equation}
For $q$ sufficiently small, we have \begin{equation*}
    \frac{p}{V}>\frac{\partial_{V}r}{(1-\mu)r}\left(\mu-\frac{Q^{2}}{r^{2}}\right).
\end{equation*}
Therefore, we have\begin{equation*}
    \partial_{V}\left(\log\left(-\Lambda^{p}\left(\frac{V}{(-U)}\right)^{-p}\partial_{U}r\right)\right) = -\left(\frac{p}{V}-\frac{\partial_{V}r}{(1-\mu)r}\left(\mu-\frac{Q^{2}}{r^{2}}\right)\right)<0.
\end{equation*}
Directly integrating the above equation, we have \begin{equation*}
    0<-\Lambda^{p}\left(\frac{V}{(-U)}\right)^{-p}\partial_{U}r<\frac{3}{2}.
\end{equation*}
This closes the bootstrap argument for $\partial_{U}r$.

To estimate $\partial_{V}r$, we have\begin{equation}
\partial_{U}\left(\partial_{V}r(U,V)-\frac{\partial_{v}r(-\frac{V}{\Lambda},0)}{\left(\frac{V}{\Lambda}\right)^{q^{2}}}\right) = \frac{\mu-\frac{Q^{2}}{r^{2}}}{1-\mu}\frac{\partial_{U}r}{r}\left(\partial_{V}r(U,V)-\frac{\partial_{v}r(-\frac{V}{\Lambda},0)}{(\frac{V}{\Lambda})^{q^{2}}}\right)+\frac{\mu-\frac{Q^{2}}{r^{2}}}{1-\mu}\frac{\partial_{U}r}{r}\frac{\partial_{v}r(-\frac{V}{\Lambda},0)}{(\frac{V}{\Lambda})^{q^{2}}}.
\end{equation}
Conjugating the above equation with $\left(\partial_{V}r(U,V)-\frac{\partial_{v}r(-\frac{V}{\Lambda},0)}{(\frac{V}{\Lambda})^{q^{2}}}\right)$ and integrating, we have\begin{equation}
\begin{aligned}
\frac{1}{2}\left(\partial_{V}r(U,V)-\frac{{\partial_{v}r}(-\frac{V}{\Lambda},0)}{(\frac{V}{\Lambda})^{q^{2}}}\right)^{2}&\leq q^{4-24\delta}q^{2\delta}+Cq^{4-22\delta}\Lambda^{-p}\int_{-\frac{V}{\Lambda}}^{U}\frac{\left(\frac{V}{(-U^{\prime})}\right)^{p}}{V}dU^{\prime}\\&\leq
q^{4-24\delta}q^{2\delta}+q^{4-24\delta}q^{2\delta}\frac{1}{p}\Lambda^{p-1},
\end{aligned}
\end{equation}
Hence, we can close the bootstrap estimate for $\partial_{V}r$.

It remains to estimate $\mu$. Recall the transport equation for $\mu$:
\begin{equation*}
\partial_{V}\mu+\left(\frac{\partial_{V}r}{r}+\frac{r}{\partial_{V}r}\left\vert\partial_{V}\phi\right\vert^{2}\right)\mu = \frac{r}{\partial_{V}r}\left\vert\partial_{V}\phi\right\vert^{2}+\frac{\partial_{V}r}{r^{3}}Q^{2}.
\end{equation*}
Conjugating with $V^{2p}\mu$ and integrating, we have\begin{equation}
\label{estimate for mu in region 3}
\begin{aligned}
&\frac{1}{2}V^{2p}\mu^{2}+\int_{\Lambda(-U)}^{V}\left(\frac{\partial_{V}r}{r}-\frac{p}{V^{\prime}}+\frac{r}{\partial_{V}r}\left(\partial_{V}\phi\right)^{2}\right)\left((V^{\prime})^{p}\mu\right)^{2}dV^{\prime}\\=&
\int_{\Lambda(-U)}^{V}\left(\frac{r}{\partial_{V}r}\left(\partial_{V}\phi\right)^{2}+\frac{\partial_{V}r}{r^{3}}Q^{2}\right)(V^{\prime})^{2p}\mu dV^{\prime}\\\leq&q^{4-18\delta}V^{2p}+q^{4-18\delta}\int_{\Lambda(-U)}^{V}(V^{\prime})^{2p-1}dV^{\prime}\\\leq&q^{4-18\delta}V^{2p}+\frac{1}{2p}q^{4-18\delta}V^{2p}.
\end{aligned}
\end{equation}
We can choose a suitable $p$ such that the bulk term on the left-hand side of the above estimate is positive. Hence, we have \begin{equation*}
    \left\vert\mu\right\vert\leq \frac{1}{2}q^{2-10\delta}.
\end{equation*}
This concludes the proof of this proposition.
\end{proof}

\subsection{Construction of region IV}
In this section, we construct the solution in the asymptotically flat region: \begin{equation*}
    R_{IV} =\left\{(U,V),\ -e^{-1}\leq U<0,\ \frac{V}{(-U)}\geq \Lambda,\ V\geq\frac{1}{2}\right\}.
\end{equation*}
We further decompose the region $R_{IV}$ into two regions: \begin{align*}
   & R_{IV}^{(1)} =\left\{(U,V),\ -e^{-1}\leq U<0,\ \frac{V}{(-U)}\geq\Lambda,\ \frac{1}{2}\leq V\leq e^{-1}\Lambda\right\},\\&
   R_{IV}^{(2)} = \left\{(U,V),\ -e^{-1}\leq U<0,\ V\geq e^{-1}\Lambda\right\}.
\end{align*}
\begin{figure}[htbp]
    \centering
    \begin{tikzpicture}[ scale=0.8, >=Latex, every node/.style={font=\small} ]
\coordinate (O) at (0,0);
\coordinate (A) at (4,-4);
\coordinate (B) at (6,6);
\coordinate (C) at (10,2);
\draw (O)--(A);
\draw[dashed] (O)--(B);
\draw (A)--(C);
\draw[dashed] (B) -- (C)
    node[midway, above, sloped]
    {\rotatebox{90}{$\mathscr{I}$}};
\node[left] at (O) { $\mathcal{O}$};
\coordinate (A1) at (6,-2);
\coordinate (A2) at (8,0);
\draw[dashed] (O)--(A1);
\draw[dashed] (O)--(A2);
\coordinate (A3) at (3,0);
\coordinate (A4) at (1.5,1.5);
\draw[dashed] (A3)--(A4);
\coordinate (L1) at (2.5,-1.5);
\coordinate (L2) at (3.5,-0.5);

\coordinate (L3) at (1.5,0.5);

\coordinate (L4) at (5,3);

\fill[gray!30] (A2) -- (A3) -- (A4)--(B)--(C) -- cycle;
\node at (L1){I};
\node at (L2) {II};
\node at (L3) {III};
\coordinate (f) at (4,4);
\draw[dashed] (f)--(A2);
\coordinate (L5) at (4,2);
\coordinate (L6) at (7,3);
\node at (L5) {$R_{IV}^{(1)}$};
\node at (L6){$R_{IV}^{(2)}$};
\end{tikzpicture}
    \caption{$R_{IV}$}
    \label{fig: region 4}
\end{figure}
We will first extend our exterior solution to the region $R_{IV}^{(1)}$ and then to the region $R_{IV}^{(2)}$.

By Proposition~\ref{prop: estimate on the cone of the right boundary of region 2} and Proposition~\ref{prop: estimate in the region 3}, on the boundary of the region $R_{IV}^{(1)}$ $$\partial_{low}R_{IV}^{(1)}: = \{V = \frac{1}{2},\ -\frac{1}{2}\frac{1}{\Lambda}\leq U<0\}\cup \left\{\frac{V}{(-U)} = \Lambda,\ V\geq\frac{1}{2}\right\},$$
we have \begin{align*}
    &r\approx V,\quad 0<-\partial_{U}r\lesssim \left(\frac{V}{(-U)}\right)^{p},\quad \partial_{V}r\approx 1,\quad \vert \mu\vert\lesssim q^{2-11\delta},\\&
    \left\vert V\partial_{U}\phi\right\vert\lesssim q^{1-5\delta}(-U)^{-2p},\qquad \left\vert V^{2}\partial_{V}\phi\right\vert\lesssim q^{1-5\delta},\quad \vert\phi\vert\lesssim q^{1-5\delta}\left\vert\log V\right\vert,\\&
    \vert Q\vert\lesssim \vert q_{0}\vert q^{2-11\delta}V,\quad \vert A_{U}\vert\lesssim \vert q_{0}\vert q^{2-11\delta}V(-U)^{-p}.
\end{align*}
By the choice of our initial data on the outgoing cone, on the boundary of the region $R_{IV}^{(2)}$\begin{equation*}
    \partial_{low}R_{IV}^{(2)}: = \left\{U = -e^{-1},\ V\geq e^{-1}\Lambda\right\},
\end{equation*}
we have \begin{equation*}
    \vert\partial_{V}\phi\vert\leq \epsilon V^{-2},\quad \partial_{V}r\approx 1.
\end{equation*}
Assuming $q\ll\epsilon$, we make the following bootstrap assumption in the region $R_{IV}\cap \mathcal{Q}_{u_{0},v_{0}}^{(ex)}$: \begin{align}
    &0<-\partial_{U}r\leq B_{2}\left(\frac{V}{(-U)}\right)^{p},\quad \frac{1}{B_{2}}\leq \partial_{V}r\leq B_{2},\quad \vert\mu\vert\leq B_{2}\epsilon,\label{eq: bootstrap in region 4.1-1}\\&
    \left\vert V\partial_{U}\phi\right\vert\leq B_{2}\epsilon(-U)^{-2p},\quad \left\vert V^{2}\partial_{V}\phi\right\vert\leq B_{2}\epsilon.\label{eq: bootstrap in region 4.1-2}
\end{align}
Assuming the above bootstrap bound, we have the following estimates on $(r,Q,A_{U},\phi)$ in the region $R_{IV}\cap \mathcal{Q}_{u_{0},v_{0}}^{(ex)}$.
\begin{proposition}
    In the region $R_{IV}^{(1)}\cap \mathcal{Q}_{u_{0},v_{0}}^{(ex)}$, under the bootstrap assumption~\eqref{eq: bootstrap in region 4.1-1}--\eqref{eq: bootstrap in region 4.1-2}, we have \begin{align*}
        r\approx V,\quad \vert\phi\vert\lesssim \epsilon,\quad \vert Q\vert\lesssim\vert q_{0}\vert \epsilon^{2}V,\quad \vert A_{U}\vert\lesssim \vert q_{0}\vert \epsilon^{2}\left(\frac{V}{(-U)}\right)^{p}.
    \end{align*}
    \begin{proof}
        The estimate for $(r,\phi)$ follows directly from our bootstrap assumptions. To estimate $Q$, we have \begin{equation*}
            \left\vert\partial_{V}Q\right\vert \leq\vert q_{0}\vert r^{2}\vert\phi\vert\vert\partial_{V}\phi\vert\lesssim \vert q_{0}\vert\epsilon^{2}.
        \end{equation*}
        Then we can deduce the bound for $Q$. For $A_{U}$, we have \begin{equation*}
            \left\vert\partial_{V}A_{U}\right\vert\lesssim\left\vert\frac{Q\partial_{U}r\partial_{V}r}{(1-\mu)r^{2}}\right\vert.
        \end{equation*}
        Therefore, we have \begin{equation*}
            \vert A_{U}\vert\lesssim\vert q_{0}\vert \epsilon^{2}\left(\frac{V}{(-U)}\right)^{p}.
            \end{equation*}
            This concludes the proof.
    \end{proof}
\end{proposition}
We can prove the following proposition, improving the above bootstrap bounds.
\begin{proposition}
    Assuming the bootstrap bounds~\eqref{eq: bootstrap in region 4.1-1}--\eqref{eq: bootstrap in region 4.1-2} in the region $R_{IV}\cap \mathcal{Q}_{u_{0},v_{0}}^{(ex)}$ for some constant $B_{2}$, we can improve the bounds in the same region to \begin{align*}
        &0<-\partial_{U}r<\frac{1}{2}B_{2}\left(\frac{V}{(-U)}\right)^{p},\quad \frac{2}{B_{2}}\leq \partial_{V}r\leq \frac{1}{2}B_{2},\quad \vert \mu\vert\leq\frac{1}{2}B_{2}\epsilon,\\&
        \left\vert V\partial_{U}\phi\right\vert\leq\frac{1}{2}B_{2}\epsilon(-U)^{-2p},\quad \left\vert V^{2}\partial_{V}\phi\right\vert\leq\frac{1}{2}B_{2}\epsilon.
    \end{align*}
\end{proposition}
\begin{proof}
The proof of this proposition relies on the strict hierarchy of the estimates. We first improve the bounds for $V$-derivatives, since we can use the largeness of $\Lambda$. Then we use these improved bounds to control $U$-derivatives.

For $\partial_{V}r$, we have \begin{equation*}
        \left\vert\partial_{U}\partial_{V}r\right\vert = \left\vert\frac{\partial_{U}r\partial_{V}r}{(1-\mu)r}\left(\mu-\frac{Q^{2}}{r^{2}}\right)\right\vert\lesssim B_{2}^{3}\epsilon V^{-1}\left(\frac{V}{(-U)}\right)^{p}.
    \end{equation*}
    Integrating from the boundary $\left\{U = -e^{-1}\right\}$, we have \begin{equation*}
        \partial_{V}r\approx C+CB_{2}^{3}\epsilon\Lambda^{-1+p}\int_{-e^{-1}}^{U}(-U^{\prime})^{-p}dU^{\prime}\approx C,
    \end{equation*}
    where $C$ is a constant only depending on the initial data on the outgoing cone.

    Integrating from the boundary $\left\{\frac{V}{(-U)} = \Lambda\right\}$, we have \begin{equation*}
        \partial_{V}r(U,V)\approx C+CB_{2}^{3}\epsilon V^{-1+p}\int_{-\frac{V}{\Lambda}}^{U}(-U^{\prime})^{-p}dU^{\prime}\approx C+CB_{2}^{3}\epsilon V^{-1+p}\left(\frac{V}{\Lambda}\right)^{1-p}\approx C.
    \end{equation*}
Therefore, we can improve the bootstrap bound for $\partial_{V}r$.

For $\partial_{V}\phi$, we have \begin{align*}
    \left\vert\partial_{U}\left(r\partial_{V}\phi\right) \right\vert=& \left\vert-\partial_{V}r\partial_{U}\phi-iq_{0}\left(A_{U}\partial_{V}(r\phi)+\frac{Q\partial_{U}r\partial_{V}r}{r(1-\mu)}\phi\right)\right\vert\\\lesssim& B_{2}^{2}\epsilon V^{-1}(-U)^{-2p}+\vert q_{0}\vert^{2}\epsilon^{2}B_{2}^{3}V^{-1}\left(\frac{V}{(-U)}\right)^{p}.
\end{align*}
Integrating from $\partial_{low}R_{IV}^{(1)}\cup\partial_{low}R_{IV}^{(2)}$, arguing as in the estimate for $\partial_{V}r$, and taking $\epsilon$ sufficiently small, we can improve the bootstrap bound for $\partial_{V}\phi$: \begin{equation}
    \left\vert V^{2}\partial_{V}\phi\right\vert\leq\frac{B_{2}}{\Lambda^{p}}\epsilon.
\end{equation}

Next, we estimate $\partial_{U}r$. By the Raychaudhuri equation~\eqref{eq:spherical symmetric equaion1}, we have that $\partial_{U}r<0$. Using the equation\begin{equation*}
        \partial_{V}\partial_{U}r = \frac{\partial_{U}r\partial_{V}r}{r(1-\mu)}\left(\mu-\frac{Q^{2}}{r^{2}}\right),
    \end{equation*}
    we have \begin{equation*}
        \partial_{V}\left(\log(-\partial_{U}r)\right) = \frac{\partial_{V}r}{r(1-\mu)}\left(\mu-\frac{Q^{2}}{r^{2}}\right)\lesssim B_{2}\epsilon\frac{1}{V}.
    \end{equation*}
Integrating the above equation from the boundary $\partial_{low}R_{IV}^{(1)}$, we have \begin{equation*}
    \log(-\partial_{U}r)(U,V) \leq \log(-\partial_{U}r)|_{\partial_{low}R_{IV}^{(1)}}+CB_{2}\epsilon\log(\vert V\vert).
\end{equation*}
Therefore, we have \begin{equation*}
        0<-\partial_{U}r<\frac{1}{2}B_{2}\left(\frac{V}{(-U)}\right)^{p}.
    \end{equation*}
    To close the bootstrap argument for $\partial_{U}\phi$, we have \begin{align*}
        \left\vert\partial_{V}(r\partial_{U}\phi)\right\vert = &\left\vert-\partial_{U}r\partial_{V}\phi-iq_{0}\left(A_{U}\partial_{V}(r\phi)+\frac{Q\partial_{U}r\partial_{V}r}{r(1-\mu)}\phi\right)\right\vert\\\lesssim& \frac{B_{2}^{2}\epsilon}{\Lambda^{p}}V^{-2}\left(\frac{V}{(-U)}\right)^{p}.
    \end{align*}
Integrating from the boundary $\partial_{low}R_{IV}^{(1)}$, we can close the bootstrap bound for $\partial_{U}\phi$.

Finally, to close the bootstrap estimate for $\mu$, arguing similarly to~\eqref{estimate for mu in region 3}, we have \begin{align*}
    &\frac{1}{2}V^{2p}\mu^{2}+\int_{\Lambda(-U)}^{V}\left(\frac{\partial_{V}r}{r}-\frac{p}{V^{\prime}}+\frac{r}{\partial_{V}r}\left(\partial_{V}\phi\right)^{2}\right)\left((V^{\prime})^{p}\mu\right)^{2}dV^{\prime}\\=&\left(\frac{1}{2}(\cdot)^{2p}\mu^{2}(\cdot,U)\right)\Bigg|_{\partial_{low}R_{IV}^{(1)}}+
\int_{\partial_{low}R_{IV}^{(1)}}^{V}\left(\frac{r}{\partial_{V}r}\left(\partial_{V}\phi\right)^{2}+\frac{\partial_{V}r}{r^{3}}Q^{2}\right)(V^{\prime})^{2p}\mu dV^{\prime}\\\leq&q^{4-22\delta}+\frac{\epsilon^{2}}{\Lambda^{p}}\int_{\Lambda(-U)}^{V}(V^{\prime})^{2p-1}dV^{\prime}.
\end{align*}
Therefore, we have \begin{equation*}
    \vert\mu\vert\leq\frac{1}{2}B_{2}\epsilon.
\end{equation*}
This concludes the proof.
\end{proof}
The upper bound for the spacetime charge $Q$ provided by the preceding bootstrap estimates does not guarantee that $Q$ remains finite as one approaches the null infinity, and consequently, does not rule out divergence of the Hawking mass. In the next proposition, we improve the estimates for both $Q$ and $m$.
\begin{proposition}
    For each fixed $U\in[-e^{-1},0]$, the spacetime charge $Q$ and the Hawking mass $m$ have a limit when $V\rightarrow\infty$.
\end{proposition}
\begin{proof}
    Since $\vert\partial_{U}\phi\vert(U,V)\lesssim V^{-1}(-U)^{-2p}$, integrating from the outgoing initial data on $\{U = -e^{-1}\}$, we have \begin{equation*}
        \lim_{V\rightarrow\infty}\phi(U,V) = 0.
    \end{equation*}
    Therefore, we can expand $\phi$ as \begin{equation*}
        \phi = \frac{\Phi(U)}{r}+O(V^{-2}).
    \end{equation*}
    Considering the transport equation for $Q$: \begin{align*}
        \partial_{V}Q = q_{0}r^{2}\Im\left(\phi\overline{\partial_{V}\phi}\right) \approx& q_{0}V^{2}\Im\left(\left(\frac{\Phi(U)}{r}+O(V^{-2})\right)\overline{\left(-\frac{\Phi(U)\partial_{V}r}{r^{2}}+O(V^{-3})\right)}\right) \\\lesssim&q_{0}V^{2}\Im\left(-\frac{\vert\Phi(U)\vert^{2}\partial_{V}r}{r^{2}}+O(V^{-4})\right)\\\lesssim& q_{0}V^{-2}.
        \end{align*}
    Therefore we have \begin{equation*}
        \lim_{V\rightarrow\infty}Q(U,V)<\infty.
    \end{equation*}
    Recall the transport equation for the Hawking mass $m$: 
    \begin{equation*}
        \partial_{V}m +\frac{r}{\partial_{V}r}\left\vert\partial_{V}\phi\right\vert^{2}m = \frac{1}{2}\left(\frac{r^{2}}{\partial_{V}r}\vert\partial_{V}\phi\vert^{2}+\frac{\partial_{V}r}{r^{2}}Q^{2}\right).
    \end{equation*}
    Since the right-hand side of the above equation is now integrable, we can directly integrate the above equation to conclude \begin{equation*}
        \lim_{V\rightarrow\infty}m(U,V)<\infty.
    \end{equation*}
    This concludes the proof.
\end{proof}
\appendix
\section{Proof of BV local well-posedness in Proposition~\ref{prop: BVlwp}}
\label{app: BVLWP}
In this section, we prove Proposition~\ref{prop: BVlwp}. In the setting of Proposition~\ref{prop: BVlwp}, we assume\begin{align*}
    &\partial_{u}r(u,v_{0}) = -1,\ \partial_{u}\phi(u,v_{0})\in BV_{u},\ A_{u}(u,v_{0}) = 0,\\&
    \partial_{v}r(u_{0},v) = 1,\ \partial_{v}\phi(u_{0},v)\in BV_{v},\\&
    r(u_{0},v_{0}) = r_{0}>(u_{1}-u_{0}),\ \mu(u_{0},v_{0}) = \mu_{0},\ \phi(u_{0},v_{0}) = \phi_{0},\ Q(u_{0},v_{0}) = Q_{0}.
\end{align*}
By the standard local well-posedness result, for $\partial_{u}\phi(u,v_{0})\in C^{1}$ and $\partial_{v}\phi(u,v_{0})\in C^{1}$, there exists a unique $C^{1}$ solution arising from the given initial data in the domain \begin{equation*}
    [u_{0},u_{1}]\times [v_{0},v_{0}+\delta)
\end{equation*}
for $\delta$ sufficiently small. Moreover, we can choose $\delta$ sufficiently small such that in that region \begin{equation*}
    r\approx 1,\ -\partial_{u}r\approx 1,\ \partial_{v}r\approx 1.
\end{equation*}
Let $\theta = r\partial_{v}\phi$ and $\zeta = r\partial_{u}\phi$. By the standard limiting argument, it suffices to establish the bounds for $TV_{v}[\theta](u)$ and $TV_{u}[\zeta](v)$ in terms of $TV_{v}[\theta](u_{0})$ and $TV_{u}[\zeta](v_{0})$, where $TV$ means the total variation.
\subsection{A priori bounds}
\label{appsec: A priori bounds}
In this section, we derive some useful a priori bounds. To simplify the notation, let $I$ be any constant depending on the quantities on the initial characteristic hypersurfaces and $P$ be any polynomial.

\paragraph{A priori bounds for $\phi$} For $\phi$, by the fundamental theorem of calculus, we have \begin{align*}
    \phi(u,v) = &\phi(u,v_{0})+\int_{v_{0}}^{v}\partial_{v}\phi(u,\bar{v})d\bar{v}.
\end{align*}
Therefore, we have \begin{align*}
    &\sup_{v}\Vert\phi\Vert_{L_{u}^{\infty}}\leq\Vert \phi(\cdot,v_{0})\Vert_{L_{u}^{\infty}}+\delta\sup_{v}\Vert \theta\Vert_{L_{u}^{\infty}}\leq\Vert \phi(\cdot,v_{0})\Vert_{L_{u}^{\infty}}+\delta \Vert\theta(\cdot,v_{0})\Vert_{L_{u}^{\infty}}+\delta \sup_{u}TV_{v}[\theta]\lesssim I+\delta \sup_{u}TV_{v}[\theta],\\&
    \sup_{v}TV_{u}[\phi](v)= \sup_{v}\int_{u_{0}}^{u_{1}}\frac{\vert \zeta\vert}{r}du\lesssim \sup_{v}\Vert\zeta(\cdot,v)\Vert_{L_{u}^{\infty}}\lesssim I+\sup_{v}TV_{u}[\zeta],\\&
    \sup_{u}TV_{v}[\phi](u) = \sup_{u}\int_{v_{0}}^{v_{0}+\delta}\frac{\vert\theta\vert}{r}dv\lesssim \delta\sup_{u}\Vert\theta(u,\cdot)\Vert_{L_{v}^{\infty}}\lesssim \delta I+\delta\sup_{u}TV_{v}[\theta].
\end{align*}
\paragraph{A priori bounds for $Q$}
Using the transport equation~\eqref{eq:v-transport equation for Q} for $Q$, we have \begin{align*}
    Q(u,v) = Q(u,v_{0})+q_{0}\int_{v_{0}}^{v}r^{2}\Im(\phi\overline{\partial_{v}\phi})d\bar{v}.
\end{align*}
Therefore, we have 
\begin{equation}
\begin{aligned}
   \sup_{v} \Vert Q\Vert_{L_{u}^{\infty}}(v)\leq& \Vert Q(\cdot,v_{0})\Vert_{L_{u}^{\infty}}+\delta \left(\sup_{v}\Vert\phi\Vert_{L_{u}^{\infty}}\right)\left(\sup_{v}\Vert \theta\Vert_{L_{u}^{\infty}}\right)\\\lesssim& I+I\vert q_{0}\vert\left(\delta \sup_{u}TV_{v}[\theta]+\delta^{2}\left(\sup_{u}TV_{v}[\theta]\right)^{2}\right)\\=&I+\delta P(\sup_{u}TV_{v}[\theta],\delta).
\end{aligned}
\label{appeq: estimate for Linfty Q}
\end{equation}
For the $TV_{u}$-norm of $Q$, we have \begin{align*}
    \sup_{v}TV_{u}[Q]\lesssim& I+\delta \left(\sup_{u}TV_{v}[\theta]+\sup_{v}TV_{u}[\theta]+\sup_{v}TV_{u}[\zeta]+(\sup_{v}TV_{u}[\zeta])(\sup_{u}TV_{v}[\theta])\right)\\&+\delta^{2}\left((\sup_{u}TV_{v}[\theta])^{2}+(\sup_{u}TV_{v}[\theta])(\sup_{v}TV_{u}[\theta])\right).
\end{align*}
For $TV_{v}[Q]$, we have \begin{equation}
\begin{aligned}
    \sup_{u}TV_{v}[Q] \lesssim \int_{v_{0}}^{v_{0}+\delta} \Vert\phi\Vert_{L_{u}^{\infty}}\Vert \theta\Vert_{L_{u}^{\infty}}dv\lesssim& \delta\left(I+\delta\sup_{u}TV_{v}[\theta]\right)\left(I+\sup_{u}TV_{v}[\theta]\right)\\\lesssim&
    \delta\left(I+\sup_{u}TV_{v}[\theta]+\delta\left(\sup_{u}TV_{v}[\theta]\right)^{2}\right)\\\lesssim&\delta I+\delta P\left(\sup_{u}TV_{v}[\theta],\delta\right).
    \end{aligned}
\end{equation}
\paragraph{A priori bounds for $A_{u}$ and $TV_{v}[A_{u}]$}
Using the equation~\eqref{eq:spherical symmetric equation last}, we have \begin{align*}
    A_{u}(u,v) = \int_{v_{0}}^{v}\frac{2Q}{r^{2}}\frac{\lambda\nu}{1-\mu}d\bar{v}.
\end{align*}
Then we have \begin{align}
   \sup_{v} \Vert A_{u}\Vert_{L_{u}^{\infty}}\lesssim\delta \sup_{v}\Vert Q\Vert_{L_{u}^{\infty}} \lesssim \delta I+\delta I\left(\delta \sup_{u}TV_{v}[\theta]+\delta^{2}(\sup_{u}TV_{v}[\theta])^{2}\right)=\delta I+\delta P\left(\sup_{u}TV_{v}[\theta],\delta\right).
   \label{appeq: final estimate for Linfty Au}
\end{align}
For $TV_{v}[A_{u}]$, we have \begin{equation}
    \sup_{u}TV_{v}[A_{u}] = \sup_{u}\int_{v_{0}}^{v_{0}+\delta}\left\vert\frac{2Q}{r^{2}}\frac{\lambda\nu}{1-\mu}\right\vert dv\lesssim \sup_{v}\Vert Q\Vert_{L_{u}^{\infty}}\lesssim I+\delta P\left(\sup_{u}TV_{v}[\theta],\delta\right).
\end{equation}
\paragraph{A priori bounds for $TV_{u}[\theta]$ and $TV_{v}[\zeta]$}
For $TV_{u}[\theta]$, we have \begin{align*}
    TV_{u}[\theta](v) = \int_{u_{0}}^{u_{1}}\vert\partial_{u}\theta\vert(u,v)du\lesssim &\Vert \zeta\Vert_{L_{u}^{\infty}}(v)+\vert q_{0}\vert\left(\Vert A_{u}\Vert_{L_{u}^{\infty}}(v)\left(\Vert \theta\Vert_{L_{u}^{\infty}}+\Vert\phi\Vert_{L_{u}^{\infty}}\right)(v)+\Vert Q\Vert_{L_{u}^{\infty}}\Vert\phi\Vert_{L_{u}^{\infty}}(v)\right)\\\lesssim&I+TV_{u}[\zeta](v)+\delta\sup_{u}TV_{v}[\theta]+\delta^{2}\left(\sup_{u}TV_{v}[\theta]\right)^{2}+\delta^{3}\left(\sup_{u}TV_{v}[\theta]\right)^{3}.
\end{align*} 
Therefore, we have \begin{equation}
\begin{aligned}
    \sup_{v} TV_{u}[\theta]\lesssim& I+\sup_{v}TV_{u}[\zeta]+\delta \sup_{u}TV_{v}[\theta]+\delta^{2}\left(\sup_{u}TV_{v}[\theta]\right)^{2}+\delta^{3}\left(\sup_{u}TV_{v}[\theta]\right)^{3}\\\lesssim& I+\sup_{v}TV_{u}[\zeta]+\delta\sup_{u}TV_{v}[\theta]+\delta^{3}\left(\sup_{u}TV_{v}[\theta]\right)^{3}\\\lesssim& I+\sup_{v}TV_{u}[\zeta]+\delta P\left(\sup_{u}TV_{v}[\theta],\delta\right).
    \end{aligned}\label{eq: final a prior bound for TVutheta}
\end{equation}
Using the above bound~\eqref{eq: final a prior bound for TVutheta}, we have \begin{equation}
    \label{eq: final a priori bound for TvuQ}
    \begin{aligned}
        \sup_{v}TV_{u}[Q]\lesssim &I+\delta\left(\sup_{u}TV_{v}[\theta]+\sup_{v}TV_{u}[\zeta]+(\sup_{v}TV_{u}[\zeta])(\sup_{u}TV_{v}[\theta])\right)\\&+\delta^{2}(\sup_{u}TV_{v}[\theta])^{2}+\delta^{3}\left(\sup_{u}TV_{v}[\theta]\right)^{3}+\delta^{4}\left(\sup_{u}TV_{v}[\theta]\right)^{4}\\\lesssim&I+\delta P\left(\sup_{u}TV_{v}[\theta],\sup_{v}TV_{u}[\zeta],\delta\right).
    \end{aligned}
\end{equation}
Similarly, for $TV_{v}[\zeta]$, we have \begin{equation}
    \sup_{u}TV_{v}[\zeta]\lesssim \delta\left(I+P\left(\sup_{u}TV_{v}[\theta],\delta\right)\right).
\end{equation}
\paragraph{A priori bounds for $m$ and $\mu$}
Using the equation~\eqref{eq: v-transport equation for m}, we have \begin{equation}
    m(u,v) = e^{-\int_{v_{0}}^{v}\frac{\theta^{2}}{r\lambda}d\bar{v}}m(u,v_{0})+\frac{1}{2}\int_{v_{0}}^{v}e^{\int_{v}^{\bar{v}}\frac{\theta^{2}}{r\lambda}d\widetilde{v}}\left(\frac{\theta^{2}}{\lambda}+\frac{\lambda Q^{2}}{r^{2}}\right)d\bar{v}.\label{appeq: expression for m in appendix}
\end{equation}
Then we have 
\begin{equation}
\begin{aligned}
\sup_{v}\Vert\mu\Vert_{L_{u}^{\infty}}\approx\sup_{v}\Vert m\Vert_{L_{u}^{\infty}}\lesssim I+\delta(\sup_{v}\Vert \theta\Vert_{L_{u}^{\infty}}^{2}+\sup_{v} \Vert Q\Vert_{L_{u}^{\infty}}^{2})\lesssim& I+\delta \sup_{u}TV_{v}[\theta]+\delta^{4}(\sup_{u}TV_{v}[\theta])^{4}\\\lesssim& I+\delta P\left(\sup_{u}TV_{v}[\theta],\delta\right).
\end{aligned}
\label{appeq: Linfty bound for mu and m}
\end{equation}
\paragraph{A priori bounds for $TV_{u}[\lambda]$ and $TV_{v}[\nu]$}
For $TV_{u}[\lambda]$, we have 
\begin{equation}
\begin{aligned}
    \sup_{v}TV_{u}[\lambda]\lesssim \sup_{v}TV_{u}[\log\lambda]\lesssim \sup_{v}\int_{u_{0}}^{u_{1}}\vert\partial_{u}\log\lambda\vert&\lesssim\int_{u_{0}}^{u_{1}}\frac{-\nu}{1-\mu}\left(\mu+\frac{Q^{2}}{r^{2}}\right)du\\&\lesssim\sup_{v} \Vert\mu\Vert_{L_{u}^{\infty}}+\sup_{v}\Vert Q\Vert_{L_{u}^{\infty}}^{2}\\&\lesssim I+\delta\sup_{u}TV_{v}[\theta]+\delta^{4}(\sup_{u}TV_{v}[\theta])^{4}\\&\lesssim I+\delta P\left(\sup_{u}TV_{v}[\theta],\delta\right).
\end{aligned}
\label{appeq: TVu bound for lambda}
\end{equation}
Similarly, for $TV_{v}[\nu]$, we have \begin{equation}
    \sup_{u}TV_{v}[\nu]\lesssim \delta \left(I+\delta\sup_{u}TV_{v}[\theta]+\delta^{4}(\sup_{u}TV_{v}[\theta])^{4}\right)\lesssim \delta I+\delta P\left(\sup_{u}TV_{v}[\theta],\delta\right).\label{appeq: TV v norm of nu}
\end{equation}
\paragraph{A priori bounds for $TV_{u}[m]$ and $TV_{u}[\mu]$}
Next, we estimate $TV_{u}[m]$. By~\eqref{appeq: expression for m in appendix}, we have \begin{align}
    \sup_{v}TV_{u}[m]\lesssim I+\delta \left(\sup_{v}TV_{u}[\zeta]\right)P\left(\sup_{u}TV_{v}[\theta],\delta\right)+\delta P\left(\sup_{u}TV_{v}[\theta],\delta\right).\label{appeq: TVu[m]}
\end{align}
For $TV_{u}[\mu]$, we have \begin{align}
    \sup_{v}TV_{u}[\mu]\leq\sup_{v}TV_{u}[m]+\sup_{v}\Vert m\Vert_{L_{u}^{\infty}}\lesssim I+\delta\left(\sup_{v}TV_{u}[\zeta]\right)P(\sup_{u}TV_{v}[\theta],\delta)+\delta P(\sup_{u}TV_{v}[\theta],\delta).
    \label{appeq: TVu[mu]}
\end{align}
\paragraph{A priori bounds for $TV_{v}[m]$ and $TV_{v}[\mu]$}
For $TV_{v}[m]$, we have \begin{equation}
    \begin{aligned}
      \sup_{u}TV_{v}[m]  = \sup_{u}\int_{v_{0}}^{v_{0}+\delta}\vert\partial_{v}m\vert dv\lesssim&\sup_{u}\int_{v_{0}}^{v_{0}+\delta} \left\vert\frac{\theta^{2}}{r\lambda}m\right\vert+\left\vert\frac{\theta^{2}}{\lambda}\right\vert+\left\vert\frac{\lambda Q^{2}}{r^{2}}\right\vert dv\\\lesssim&\delta\left(\sup_{v}\Vert\theta\Vert_{L_{u}^{\infty}}^{2}\sup_{v}\Vert m\Vert_{L_{u}^{\infty}}+\sup_{v}\Vert\theta\Vert_{L_{u}^{\infty}}^{2}+\Vert Q\Vert_{L_{u}^{\infty}}^{2}\right)\\\lesssim&\delta\left(I+P\left(\sup_{u}TV_{v}[\theta],\delta\right)\right).
    \end{aligned}
    \label{appeq: TVv[m]}
\end{equation}
For $TV_{v}[\mu]$, we have \begin{equation}
    \label{appeq: TVv[mu]}
    \sup_{u}TV_{v}[\mu] = \sup_{u}TV_{v}\left[\frac{2m}{r}\right]\lesssim \sup_{u} TV_{v}[m]+\sup_{v}\Vert m\Vert_{L_{u}^{\infty}}TV_{v}r\lesssim \delta\left(I+P\left(\sup_{u}TV_{v}[\theta],\delta\right)\right).
\end{equation}
\paragraph{A priori bounds for $TV_{u}[\nu]$ and $TV_{v}[\lambda]$}
Using the equation~\eqref{eq:wave equation for r}, we have \begin{align*}
    \log(-\nu)(u,v) = \int_{v_{0}}^{v}\frac{\lambda}{(1-\mu)r}\left(\mu-\frac{Q^{2}}{r^{2}}\right)d\bar{v}.
\end{align*}
Therefore, we have \begin{equation*}
    \sup_{v}TV_{u}[\nu]\lesssim \sup_{v}TV_{u}[\log(-\nu)]\lesssim \delta\left(I+\delta \left(\sup_{v}TV_{u}[\zeta]\right)P\left(\sup_{u}TV_{v}[\theta],\delta\right)+\delta P\left(\sup_{u}TV_{v}[\theta],\delta\right)\right).
\end{equation*}
Similarly for $TV_{v}[\lambda]$, we have \begin{equation}
    \sup_{u}TV_{v}[\lambda]\lesssim \sup_{u}TV_{v}[\log\lambda]\lesssim\delta I+\delta P\left(\sup_{u}TV_{v}[\theta],\delta\right).
\end{equation}
\paragraph{A priori bounds for $TV_{u}[A_{u}]$}
For $TV_{u}[A_{u}]$, we have \begin{equation}
    \begin{aligned}
        \sup_{v}TV_{u}[A_{u}]\lesssim&\delta \sup_{v} TV_{u}\left[\frac{Q\lambda\nu}{r^{2}(1-\mu)}\right]\\\lesssim&\delta I+\delta^{2}P\left(\sup_{u}TV_{v}[\theta],\sup_{v}TV_{u}[\zeta],\delta\right).
    \end{aligned}
\end{equation}
\subsection{Estimates on $TV_{v}[\theta]$ and $TV_{u}[\zeta]$}
In this section, we assume the following bootstrap bounds:\begin{equation}
    \label{appeq: bootstrap assumption}
    \sup_{u}TV_{v}[\theta]\leq 2TV_{v}[\theta](u_{0}),\ \sup_{v}TV_{u}[\zeta]\leq 2 TV_{u}[\zeta](v_{0}).
\end{equation}

By the wave equation~\eqref{eq: pre wave equation for phi} for $\phi$, we have \begin{align}
    \partial_{v}\zeta =&-\frac{\nu}{r}\theta-iq_{0}\left(A_{u}\theta+A_{u}\lambda\phi+\frac{Q}{r}\frac{\lambda\nu}{1-\mu}\phi\right),\label{eq: transport equation for zeta}\\
    \partial_{u}\theta = &-\frac{\lambda}{r}\zeta-iq_{0}\left(A_{u}\theta+A_{u}\lambda\phi+\frac{Q}{r}\frac{\lambda\nu}{1-\mu}\phi\right).\label{eq: transport equation for theta}
\end{align}
Integrating the equation~\eqref{eq: transport equation for zeta} and taking the total variation, we have \begin{equation}
    \begin{aligned}
        TV_{u}[\zeta](v)\leq& \delta\left(\sup_{v}TV_{u}\left[\frac{\nu}{r}\theta\right](v)+\vert q_{0}\vert \sup_{v}TV_{u}[A_{u}\theta](v)+\vert q_{0}\vert\sup_{v}TV_{u}[A_{u}\lambda\phi](v)+\vert q_{0}\vert\sup_{v} TV_{u}\left[\frac{Q}{r}\frac{\lambda\nu}{1-\mu}\phi\right](v)\right)\\&+TV_{u}[\zeta](v_{0}).
    \end{aligned}
\end{equation}
For $TV_{u}\left[\frac{\nu}{r}\theta\right]$, we have \begin{align*}
    \sup_{v}TV_{u}\left[\frac{\nu}{r}\theta\right](v)\lesssim&\sup_{v}TV_{u}[\theta](v)+\sup_{v}\Vert\theta\Vert_{L_{u}^{\infty}}(v)TV_{u}[\nu](v)+\sup_{v}\Vert\theta\Vert_{L_{u}^{\infty}}(v)\\\lesssim&I+P\left(\sup_{u}TV_{v}[\theta],\sup_{v}TV_{u}[\zeta],\delta\right).
\end{align*}
We can estimate \begin{equation*}
    \sup_{v}TV_{u}[A_{u}\theta],\ \sup_{v}TV_{u}[A_{u}\lambda\phi],\ \sup_{v}TV_{u}\left[\frac{Q}{r}\frac{\lambda\nu}{1-\mu}\phi\right]
\end{equation*}
similarly. Therefore, we can conclude that \begin{equation*}
    \sup_{v}TV_{u}[\zeta]\lesssim TV_{u}[\zeta](v_{0})+\delta I+\delta P\left(\sup_{v}TV_{u}[\zeta],\sup_{u}TV_{v}[\theta],\delta\right).
\end{equation*}
Taking $\delta$ sufficiently small, we have \begin{equation*}
    \sup_{v}TV_{u}[\zeta]\leq \frac{3}{2}TV_{u}[\zeta](v_{0}).
\end{equation*}
This improves the bootstrap assumption~\eqref{appeq: bootstrap assumption} for $TV_{u}[\zeta]$. 

For $TV_{v}[\theta]$, integrating the equation~\eqref{eq: transport equation for theta} and taking the total variation, we have \begin{align*}
    TV_{v}[\theta](u)\leq \sup_{u}TV_{v}\left[\frac{\lambda}{r}\zeta\right]+\vert q_{0}\vert\sup_{u} TV_{v}[A_{u}\theta]+\vert q_{0}\vert \sup_{u}TV_{v}[A_{u}\lambda\phi]+\sup_{u}TV_{v}\left[\frac{Q}{r}\frac{\lambda\nu}{1-\mu}\phi\right]+TV_{v}[\theta](u_{0}).
\end{align*}
Since the $TV_{v}$ a priori bounds established in Appendix~\ref{appsec: A priori bounds} have $\delta$-smallness, taking $\delta$ sufficiently small will improve the bootstrap bounds for $TV_{v}[\theta]$. This concludes the proof.

\bibliographystyle{plain}
\bibliography{ref.bib}

\end{document}